\documentclass[groupedaddress,superscriptaddress,onecolumn,notitlepage,nofootinbib]{revtex4-1}
\usepackage{epsfig}
\usepackage[normalem]{ulem}
\usepackage[dvipsnames]{xcolor}
\usepackage{xr}
\usepackage{graphicx,subfigure}
\usepackage{tensor}
\usepackage{appendix}
\usepackage{cancel}

\usepackage[latin1]{inputenc}
\usepackage{amsthm}
\usepackage{amssymb}
\usepackage{amsmath}
\usepackage{bbold}
\usepackage{bbm}
\usepackage{graphicx}
\usepackage{nonfloat}
\usepackage[pdftex]{hyperref}
\usepackage{braket}
\usepackage{dsfont}
\usepackage{mathdots}
\usepackage{mathtools}
\usepackage{enumerate}
\usepackage{csquotes}
\usepackage{stmaryrd}
\usepackage[cal=boondox]{mathalfa}
\usepackage{stackengine}
\usepackage{scalerel}
\usepackage{lipsum}
\usepackage{array}
\usepackage{makecell}
\newcolumntype{x}[1]{>{\centering\arraybackslash}p{#1}}
\usepackage{tikz}
\usetikzlibrary{shapes.geometric, arrows, decorations.pathreplacing, angles, quotes}
\usepackage{booktabs}
\usepackage{xfrac}
\usepackage{siunitx}
\usepackage{centernot}
\usepackage{comment}
\usepackage[normalem]{ulem}
\usepackage{lmodern}

\newtheorem{thm}{Theorem}
\newtheorem*{thm*}{Theorem}
\makeatletter
\newcommand{\setthmtag}[1]{
  \let\oldthethm\thethm
  \renewcommand{\thethm}{#1}
  \g@addto@macro\endthm{
    \addtocounter{thm}{-1}
    \global\let\thethm\oldthethm}
  }
\makeatother

\newtheorem{prop}[thm]{Proposition}
\newtheorem*{prop*}{Proposition}
\newtheorem{lemma}[thm]{Lemma}
\newtheorem*{lemma*}{Lemma}
\newtheorem{cor}[thm]{Corollary}
\newtheorem*{cor*}{Corollary}

\newtheorem*{cj*}{Conjecture}
\newtheorem{Def}[thm]{Definition}
\newtheorem*{Def*}{Definition}

\makeatletter
\def\thmhead@plain#1#2#3{%
  \thmname{#1}\thmnumber{\@ifnotempty{#1}{ }\@upn{#2}}%
  \thmnote{ {\the\thm@notefont#3}}}
\let\thmhead\thmhead@plain
\makeatother

\theoremstyle{definition}
\newtheorem{rem}[thm]{Remark}
\newtheorem*{note}{Note}

\newcommand{\bb}{\begin{equation}}
\newcommand{\bbb}{\begin{equation*}}
\newcommand{\ee}{\end{equation}}
\newcommand{\eee}{\end{equation*}}
\newcommand*{\coloneqq}{\mathrel{\vcenter{\baselineskip0.5ex \lineskiplimit0pt \hbox{\scriptsize.}\hbox{\scriptsize.}}} =}
\newcommand*{\eqqcolon}{= \mathrel{\vcenter{\baselineskip0.5ex \lineskiplimit0pt \hbox{\scriptsize.}\hbox{\scriptsize.}}}}
\newcommand\floor[1]{\left\lfloor#1\right\rfloor}

\newcommand{\texteq}[1]{\stackrel{\mathclap{\scriptsize \mbox{#1}}}{=}}
\newcommand{\textleq}[1]{\stackrel{\mathclap{\scriptsize \mbox{#1}}}{\leq}}

\newcommand{\textgeq}[1]{\stackrel{\mathclap{\scriptsize \mbox{#1}}}{\geq}}
\newcommand{\ketbra}[1]{\ket{#1}\!\!\bra{#1}}

\newcommand{\sumno}{\sum\nolimits}

\newcommand{\Alpha}{\alpha}

\newcommand{\tcr}[1]{{\color{red} #1}}
\newcommand{\tcb}[1]{{\color{blue} #1}}

\newcommand{\id}{\mathds{1}}
\newcommand{\R}{\mathds{R}}
\newcommand{\N}{\mathds{N}}
\newcommand{\C}{\mathds{C}}

\DeclareMathOperator{\Tr}{Tr}
\DeclareMathOperator{\rk}{rk}

\DeclareMathOperator{\co}{conv}

\DeclareMathOperator{\Span}{span}
\DeclareMathAlphabet{\pazocal}{OMS}{zplm}{m}{n}

\DeclareMathOperator{\supp}{supp}

\newcommand{\lsmatrix}{\left(\begin{smallmatrix}}
\newcommand{\rsmatrix}{\end{smallmatrix}\right)}
\newcommand\smatrix[1]{{%
  \scriptsize\arraycolsep=0.4\arraycolsep\ensuremath{\begin{pmatrix}#1\end{pmatrix}}}}

\stackMath
\newcommand\xxrightarrow[2][]{\mathrel{%
  \setbox2=\hbox{\stackon{\scriptstyle#1}{\scriptstyle#2}}%
  \stackunder[5pt]{%
    \xrightarrow{\makebox[\dimexpr\wd2\relax]{$\scriptstyle#2$}}%
  }{%
   \scriptstyle#1\,%
  }%
}}

\newcommand{\tends}[2]{\xxrightarrow[\! #2 \!]{\mathrm{#1}}}
\newcommand{\tendsn}[1]{\xxrightarrow[\! n\rightarrow \infty\!]{\mathrm{#1}}}

\newcommand{\ctends}[3]{\xxrightarrow[\raisebox{#3}{$\scriptstyle #2$}]{\raisebox{-0.7pt}{$\scriptstyle \mathrm{#1}$}}}

\stackMath

\makeatletter
\newcommand*\rel@kern[1]{\kern#1\dimexpr\macc@kerna}
\newcommand*\widebar[1]{%
  \begingroup
  \def\mathaccent##1##2{%
    \rel@kern{0.8}%
    \overline{\rel@kern{-0.8}\macc@nucleus\rel@kern{0.2}}%
    \rel@kern{-0.2}%
  }%
  \macc@depth\@ne
  \let\math@bgroup\@empty \let\math@egroup\macc@set@skewchar
  \mathsurround\z@ \frozen@everymath{\mathgroup\macc@group\relax}%
  \macc@set@skewchar\relax
  \let\mathaccentV\macc@nested@a
  \macc@nested@a\relax111{#1}%
  \endgroup
}

\newcommand{\CC}{\pazocal{C}}

\newcommand{\DM}{D^M\!}
\newcommand{\DKL}{D_{K\! L}\!}
\newcommand{\M}{\pazocal{M}}
\newcommand{\X}{\pazocal{X}}
\newcommand{\NC}{N_r} 
\newcommand{\NCi}{N_r^\infty} 
\newcommand{\NCM}{N^M_r\!} 
\newcommand{\NCMi}{N^{M,\infty}_r\!} 

\newcommand{\HH}{\pazocal{H}}
\newcommand{\T}[1]{\pazocal{T}_{\mathrm{sa}}(\HH_{#1})}
\newcommand{\Tp}[1]{\pazocal{T}_{\mathrm{sa}}^+(\HH_{#1})}
\newcommand{\D}[1]{\pazocal{D}(\HH_{#1})}
\newcommand{\K}[1]{\pazocal{K}_{\mathrm{sa}}(\HH_{#1})}
\newcommand{\Kp}[1]{\pazocal{K}_{\mathrm{sa}}^+(\HH_{#1})}
\newcommand{\B}[1]{\pazocal{B}_{\mathrm{sa}}(\HH_{#1})}
\newcommand{\disp}{\mathcal{D}}

\newcommand{\tsa}{\pazocal{T}_{\mathrm{sa}}}

\newcommand{\hilb}{\pazocal{H}}
\newcommand{\CCM}{\pazocal{C}_m}
\newcommand{\CCMFD}{\pazocal{C}_m^{\mathrm{FD}}}
\newcommand{\CCFD}{\pazocal{C}_1^{\mathrm{FD}}}

\allowdisplaybreaks

\renewcommand{\tcr}[1]{{\color{red!80!black}#1}}

\let\tcb\relax
\let\tcr\relax

\begin{document}

\title{Asymptotic state transformations of continuous variable resources}

\author{Giovanni Ferrari}
\affiliation{Dipartimento di Fisica e Astronomia Galileo Galilei,
Universit\`a degli studi di Padova, via Marzolo 8, 35131 Padova, Italy}
\affiliation{Institut f\"{u}r Theoretische Physik und IQST, Universit\"{a}t Ulm, Albert-Einstein-Allee 11, D-89069 Ulm, Germany}

\author{Ludovico Lami}
\email{ludovico.lami@gmail.com}
\affiliation{Institut f\"{u}r Theoretische Physik und IQST, Universit\"{a}t Ulm, Albert-Einstein-Allee 11, D-89069 Ulm, Germany}

\author{Thomas Theurer}
\affiliation{Institut f\"{u}r Theoretische Physik und IQST, Universit\"{a}t Ulm, Albert-Einstein-Allee 11, D-89069 Ulm, Germany}

\author{Martin B. Plenio}
\affiliation{Institut f\"{u}r Theoretische Physik und IQST, Universit\"{a}t Ulm, Albert-Einstein-Allee 11, D-89069 Ulm, Germany}

\begin{abstract}
We study asymptotic state transformations in continuous variable quantum resource theories. In particular, we prove that monotones displaying lower semicontinuity and strong superadditivity can be used to bound asymptotic transformation rates in these settings. This removes the need for asymptotic continuity, which cannot be defined in the traditional sense for infinite-dimensional systems. We consider three applications, to the resource theories of (I)~optical nonclassicality, (II)~entanglement, and (III)~quantum thermodynamics. In cases~(II) and~(III), the employed monotones are the (infinite-dimensional) squashed entanglement and the free energy, respectively. For case~(I), we consider the measured relative entropy of nonclassicality and prove it to be lower semicontinuous and strongly superadditive. \tcr{One of our main technical contributions, and a key tool to establish these results, is a handy variational expression for the measured relative entropy of nonclassicality}. Our technique then yields computable upper bounds on asymptotic transformation rates, including those achievable under linear optical elements. We also prove a number of results which \tcr{guarantee that} the measured relative entropy of nonclassicality \tcr{is} bounded on any physically meaningful state and easily computable for some classes of states of interest, e.g., Fock diagonal states. We conclude by applying our findings to the problem of cat state manipulation and noisy Fock state purification.
\end{abstract}

\maketitle

\section{Introduction}

In recent years, the paradigm of quantum resource theories has established itself as the main framework to analyze and assess the operational usefulness of quantum resources~\cite{Bennett-RT, Coecke2016, quantum-resource-theories-review}. The general setting involves two sets of objects that are considered easily accessible: free states and free operations. Once these have been identified, the resource content of a state is determined by its transformation properties under free operations~\cite[Section~V]{quantum-resource-theories-review}. In the long-established tradition of classical~\cite{Shannon, CT} as well as quantum~\cite{Bennett-distillation, Plenio-Virmani, Horodecki-review} information theory, in this work we consider ultimate limitations on those transformation properties, and thus look at the asymptotic setting. Namely, we study free approximate conversion of a large number of copies of the initial state $\rho$ into as many copies of the target state $\sigma$ as possible, under the constraint that the approximation error vanishes asymptotically. The resulting transformation rate $R(\rho\!\to\!\sigma)$ can be turned into a whole family of resource quantifiers: for a fixed resourceful state $\sigma$ (respectively, $\rho$), the function $R(\,\cdot\!\to\!\sigma)$ (respectively, $R(\rho\! \to\! \cdot\,)^{-1}$) is a resource quantifier with a solid operational interpretation. In entanglement theory, for example, considering free all those transformations that can be implemented with local operations assisted by classical communication (LOCC) and choosing as fixed states Bell pairs, the above procedure leads to the distillable entanglement and the entanglement cost, respectively~\cite[Section~XV]{Horodecki-review}.


Since exact computations of asymptotic transformation rates are often challenging, it is important to \tcb{establish rigorous bounds on} them. In finite-dimensional resource theories, \tcb{it is possible do so as follows}: if $G$ is a resource monotone, i.e., a function from quantum states to the set of nonnegative real numbers that does not increase under free operations, the inequality $R(\rho\!\to\!\sigma)\leq \frac{G(\rho)}{G(\sigma)}$ holds if $G$ is (i)~additive on multiple copies of a state, and (ii)~asymptotically continuous~\cite{Vidal2000, Horodecki2000, Donald2002, Horodecki2001, Synak2006} (see also~\cite[Section~VI.A.5]{quantum-resource-theories-review}). Property~(i) can be enforced by regularization~\cite[Section~VI.A.4]{quantum-resource-theories-review}, and~(ii) turns out to hold for many monotones in finite-dimensional systems. For infinite-dimensional resource theories, this approach is however not viable, because the conventional definition of asymptotic continuity, which involves the dimension $d$ of the underlying Hilbert space, becomes meaningless. \tcb{And indeed, in infinite dimensions many monotones --- especially those based on entropic quantities --- are discontinuous everywhere~\cite{Wehrl-convergence, Wehrl, Eisert2002}. A weaker version of asymptotic continuity can be restored by imposing an energy constraint~\cite{Shirokov-AFW-1, Shirokov-AFW-2, Shirokov-AFW-3, Eisert2002}, yet doing so still does not result in any bound on the transformation rates, because the free operations employed are a priori \emph{not} required to be (uniformly) energy-constrained. Due to this seemingly merely technical complication, we are not aware of any general technique to upper bound asymptotic transformation rates in infinite-dimensional quantum resource theories prior to our work. We dub this state of affairs the ``asymptotic continuity catastrophe''.}

%
%

\tcb{This situation is particularly undesirable because infinite-dimensional systems, especially quantum harmonic oscillators, are ubiquitous in physics, and --- as suggested by quantum field theory --- perhaps fundamental. The optical modes that underlie the flourishing field of continuous variable (CV) quantum technologies~\cite{Braunstein-review, CERF, weedbrook12} 
are a prime example, but harmonic oscillators appear whenever the behavior of a physical system close to equilibrium is approximated to second order.}


Here, we devise a \tcb{simple yet} general way to circumvent \tcb{the asymptotic continuity catastrophe}, and establish rigorous bounds on transformation rates that are \tcb{equally} valid for  
finite- and infinite-dimensional quantum resource theories. Our approach relies on monotones $G$ that satisfy, in addition to~(i), also (ii')~lower semicontinuity, which is much weaker than~(ii) and does not depend upon the Hilbert space dimension, and (iii)~strong superaddivity, i.e., $G(\rho_{AB})\geq G(\rho_A) + G(\rho_B)$.  We show how (i),~(ii'), and~(iii) combined imply the sought general bound $R(\rho\!\to\!\sigma)\leq \frac{G(\rho)}{G(\sigma)}$ on the transformation rate (Theorem~\ref{general_bound_rates_thm}).

We then study three main applications, to the \tcb{infinite-dimensional} resource theories of: (I)~optical nonclassicality~\cite{Sperling2015, Tan2017, Yadin2018, NC-review, taming-PRA, taming-PRL}; (II)~quantum entanglement~\cite{Bennett-distillation, Plenio-Virmani, Horodecki-review}; and (III)~continuous variable quantum thermodynamics~\cite{Brandao-thermo,thermo-review}. Each of these applications rests upon a different strongly superadditive monotone, namely (I)~the measured relative entropy of nonclassicality, introduced and studied here, (II)~the squashed entanglement~\cite{tucci1999, squashed, faithful, rel-ent-sq, Shirokov-sq}, and (III)~the free energy~\cite{Brandao-thermo}.
Albeit strong subadditivity was known to hold for the latter two monotones, \tcb{it is only thanks to our Theorem~\ref{general_bound_rates_thm} that we are able to employ this mathematical property to deduce an upper bound on asymptotic transformation rates. To the extent of our knowledge, ours are the first such bounds for any and thus in particular the above infinite-dimensional resource theories.}

\tcb{From the technical standpoint, our main results are obtained in the context of nonclassicality~(I).} Here, one of our key contributions is the proof of a handy variational expression for the measured relative entropy of nonclassicality. This result relies on a careful application of Sion's Theorem, based, in turn, on a crucial \tcb{and carefully made} choice of topology on the space of trace class operators acting on an infinite-dimensional Hilbert space.
\tcb{This} \tcb{variational expression} has several implications, most notably: (a)~it immediately implies strong superadditivity, allowing us to deduce (b)~the sought bound on asymptotic transformation rates; and finally (c)~it points to a technique for estimating, up to arbitrary precision, said upper bound.
Another useful result established here is the finiteness of our nonclassicality monotones for any state with either finite energy or finite Wehrl entropy. Other nonclassicality monotones, such as the standard robustness of nonclassicality~\cite{taming-PRA,taming-PRL}, often diverge, giving little to no information about the actual resource content of a state.

This manuscript is organized as follows. In Section~\ref{inf dimensional section} we introduce the basic notation and concepts of quantum mechanics for infinite-dimensional systems that will be used in the remainder of the work.
Then, Section~\ref{qrt section} features a concise review of the mathematical framework of quantum resource theories.
In Section~\ref{main results section} we introduce our main results, in particular Theorems~\ref{general_bound_rates_thm},~\ref{main result thm}, and~\ref{bound_rates_NCMi_thm}, and briefly discuss their implications. The proofs start from Section~\ref{preliminary results section}, where we initiate the study of our monotones. In Section~\ref{proofs 1 section} we establish our first main result, Theorem~\ref{general_bound_rates_thm}, together with some of its consequences. Section~\ref{proofs 2 section} is devoted instead to the more involved proof of Theorems~\ref{main result thm} and~\ref{bound_rates_NCMi_thm}. In Section~\ref{secondary res section} we explore further properties of our nonclassicality monotones. In the subsequent Section~\ref{applications section} we apply them to study a wealth of examples, including noisy Fock states, Schr\"{o}dinger cat states, and squeezed states; we also test our bounds on rates for the case of distillation and dilution of Fock states and cat states in the theory of nonclassicality. Appendix~\ref{restr_asymp_cont_sec} is concerned with various restricted notions of asymptotic continuity for infinite-dimensional systems and with their limitations in proving general bounds on asymptotic transformation rates, which further motivates our analysis.

\section{Quantum mechanics for infinite-dimensional systems}\label{inf dimensional section}

With every quantum system one associates a Hilbert space $\HH$; in this work we will be mainly concerned with the case where $\dim\HH=\infty$, i.e., with infinite-dimensional quantum systems. When dealing with (both finite- and infinite-dimensional) open quantum systems, both physical states and observables are described in terms of linear self-adjoint operators acting on $\HH$ and fulfilling specific properties. As we will discuss in a moment, the structure of operator spaces is more complex in the infinite-dimensional case than in the finite-dimensional one, where many of them coincide. We start by setting the basic notation and definitions and introduce the objects that we will use in the following.

\subsection{Notation and definitions} \label{subsec_notation}

For a generic linear self-adjoint operator $X$ acting on a Hilbert space $\HH$ one can define the \textbf{operator norm} as follows: 
    \bbb
    \tcb{
    \|X\|_\infty\coloneqq\sup_{\ket{\psi}\in\HH\setminus\{0\}}\frac{|\braket{\psi|X|\psi}|}{\braket{\psi|\psi}}\,.
    }
    \eee
Operators with finite operator norm, i.e., $\|B\|_\infty<\infty$ are said to be \textbf{bounded}. \tcb{Moreover, an operator is said to be \textbf{trace class} if its \textbf{trace norm}
\bbb
\|T\|_1\coloneqq\sum_{n=0}^\infty\braket{\phi_n|\sqrt{T^\dagger T}|\phi_n}
\eee
is finite, where $\{\phi_n\}_n$ is any orthogonal basis of $\pazocal{H}$. For trace class operators we can define the \textbf{trace} as:
\bbb
\Tr [T]\coloneqq\sum_{n=0}^\infty \braket{\phi_n|T|\phi_n}\,,
\eee
which is independent of the choice of the orthogonal basis $\{\phi_n\}_n$. The trace norm of a trace class operator $T$ can then be written as $\|T\|_1=\Tr\left[\sqrt{T^\dagger T}\right]$.}

We are now ready to list below the most relevant operator spaces:
\begin{itemize}
	\item $\B{}$: the Banach space of bounded self-adjoint operators on $\HH$;
	\item $\T{}$: the Banach space of self-adjoint trace class operators on $\HH$;
	\item $\K{}$: the Banach space of self-adjoint \textbf{compact} operators on $\HH$, defined as the closure with respect to the operator norm of $\T{}$;
	\item $\D{}$: the set of \textbf{density operators} (i.e., positive semidefinite trace class operators with trace $1$) on $\HH$;
	\item $\Tp{}$: the cone of positive semidefinite (and hence self-adjoint) trace class operators on $\HH$;
	\item $\Kp{}$: the cone of positive semidefinite (and hence self-adjoint) compact operators on $\HH$;
\end{itemize}

One has that $\T{} \subseteq \K{} \subseteq \B{}$, with equality iff $\HH$ is finite dimensional. Also, the duality relation $\T{}^* = \B{}$ holds at the level of Banach spaces. We remind the reader that the dual of a Banach space $X$ equipped with a norm $\|\cdot\|_X$ is the vector space of all linear functionals $\varphi:X\to \R$ such that $\|\varphi\|_{X^*}\coloneqq \sup_{\|x\|_X\leq 1} |\varphi(x)|<\infty$, equipped with the norm $\|\cdot\|_{X^*}$.

A \textbf{quantum state} on a quantum system $A$ with Hilbert space $\HH_A$ is represented by a density operator $\rho_A\in \D{A}$. \textbf{Quantum channels} from $A$ to $B$, where $A,B$ are quantum systems, are completely positive trace preserving maps $\Lambda:\T{A}\to \T{B}$. For a quantum channel $\Lambda:\T{A}\to \T{B}$, the \textbf{adjoint} $\Lambda^\dag$ is the linear map $\Lambda^\dag: \B{B}\to \B{A}$ defined by $\Tr\left[T_A \Lambda^\dag (X_B)\right] \coloneqq \Tr\left[ \Lambda(T_A) X_B\right]$ for all $T_A\in \T{A}$ and $X_B\in \B{B}$. Among the simplest examples of quantum channels are \textbf{quantum measurements}, represented by \textbf{positive operator-valued measures} (POVM), i.e., finite collections $\M = \{E_x\}_{x\in \X}$ of positive semidefinite (bounded) operators $E_x \geq 0$ that obey the normalization rule $\sum_x E_x=\id$. Any quantum measurement can be written as a trace-preserving map by making use of classical flags $\{\ket{\phi_x}\}_{x\in \X}$: $\rho\mapsto\sumno_x\Tr[\rho E_x]\rho_x\otimes\ketbra{\phi_x}$, where $\rho_x$ is the output state in case the outcome $x$ is measured.


It is well known that the topological structure of infinite-dimensional spaces is much richer than in the finite-dimensional case. There is a wealth of topologies that can be defined on infinite-dimensional Banach spaces, and in particular on the operator spaces discussed above~\cite{wiki-operator-topologies}. In light of this fact, and for later convenience, we provide here a quick guide:
\begin{itemize}
	\item the \textbf{weak operator topology} on $\B{}$ (and hence on $\T{}$ and $\D{}$) is the coarsest topology that makes all functionals $A \mapsto \braket{\psi |A| \psi}$ continuous, for all $\ket{\psi}\in \HH$;
	\item the \textbf{weak* topology} on $\T{}$ is the coarsest topology that makes all functionals $T\mapsto \Tr[TK]$ continuous, for all $K\in \K{}$;
	\item the \textbf{weak topology} on $\T{}$ is the coarsest topology that makes all functionals $T\mapsto \Tr[TA]$ continuous, for all $A\in \B{}$;
	\item the \textbf{trace norm topology} on $\T{}$ is the one induced by the trace norm $\|\cdot\|_1$;
	\item the \textbf{operator norm topology} on $\B{}$ is the one induced by the operator norm $\|\cdot\|_\infty$.
\end{itemize}
The role of the weak* topology on $\T{}$ will play a special role for us (cf.~Lemma~\ref{Cm tilde compact lemma}).

\begin{rem} \label{w* rem}
The weak* topology is the topology induced by the Banach space $\K{}$ on its dual $\K{}^*=\T{}$. Therefore, by the Banach--Alaoglu theorem the unit ball $B_{\T{}}\coloneqq \{T\in \T{}: \|T\|_1\leq 1\}$ of $\T{}$ is weak* compact. This fact will be crucial for one of the main results of the work.
\end{rem}

We conclude this section by stating some useful facts about operator topologies. \tcb{We start by noting the following remarkable lemma, originally discovered by Davies~\cite[Lemma~4.3]{Davies1969} --- see also the `gentle measurement lemma' by Winter~\cite[Lemma~9]{VV1999} for a refined version.}

\begin{lemma}[{\cite[Lemma~4.3]{Davies1969}}] \label{swot lemma}
For a net\footnote{\tcb{A net $(x_\alpha)_\alpha$ on some set $\pazocal{X}$ is any function of the form $x:A\to \pazocal{X}$, where $A$ is a directed set, i.e., a set $A$ equipped with a preorder $\leq$ such that for all $a,b\in A$ there exists $c\in A$ with the property that $a\leq c$ and $b\leq c$.}} $(\omega_\alpha)_\alpha\subseteq \Tp{}$ of positive semidefinite trace class operators, if $\omega_\alpha \ctends{wot}{\alpha}{2pt} \omega \in \Tp{}$ in the weak operator topology, and moreover $\Tr[\omega_\alpha] \ctends{}{\,\alpha\,}{2pt} \Tr[\omega]$, then $\omega_\alpha \ctends{\,n\,}{\alpha}{2pt} \omega$ in norm.
\end{lemma}


Since two topologies are equal if and only if they have the same convergent nets, it is immediate to deduce the following.

\begin{cor} \label{swot cor}
The weak topology and the norm topology coincide on $\Tp{}$. They also coincide with the weak operator topology on $\D{}$.
\end{cor}


\begin{rem}
The norm topology does not coincide with the weak operator topology on $\Tp{}$. For instance, the sequence of Fock states $\left( \ketbra{n}\right)_n$ converges to $0$ in the weak operator topology, but it is not convergent in the norm topology (for instance because it is not of Cauchy type).
\end{rem}

\subsection{Continuous variable systems}

Among all infinite-dimensional quantum systems, a central role is played by \textbf{continuous variable (CV) systems}, and here, perhaps most notably, by finite collections of harmonic oscillators. The Hilbert space corresponding to an $m$-mode CV system is composed of all square-integrable complex-valued functions on the Euclidean space $\R^m$, denoted with $\HH_m=L^2(\R^m)$; one can then identify $\HH_m\simeq \HH_1^{\otimes m}$. Note that we will adopt the convention $\hbar=1$ hereafter. Then, the \textbf{canonical operators} $x_j$ and $p_j\coloneqq -i \frac{\partial }{\partial x_j}$ ($j=1,\ldots, m$) satisfy the \textbf{canonical commutation relations} $[x_j,x_k]= 0 = [p_j,p_k]$ and $[x_j,p_k]=i \delta_{jk} \id$, with $\id$ denoting the identity over $\HH_m$. It is customary to define the annihilation and creation operators by
\bb
a_j\coloneqq \frac{x_j + i p_j}{\sqrt2}\, ,\qquad a_j^\dag \coloneqq \frac{x_j - i p_j}{\sqrt2}\, .
\ee
In terms of $a_j,a_j^\dag$, the canonical commutation relations take the form $[a_j,a_k]\equiv 0$, $[a_j, a_k^\dag] = \delta_{jk} \id$.

On a single-mode system, \textbf{Fock states} are defined for $k\in \N$ by $\ket{k}\coloneqq \frac{1}{\sqrt{k!}} (a^\dag)^k\ket{0}$, where $\ket{0}$ is the \textbf{vacuum state}. For $\alpha\in \C$, the associated \textbf{coherent state} takes the form~\cite{Schroedinger1926-coherent, Klauder1960, Glauber1963, Sudarshan1963}
\bb
\ket{\alpha}\coloneqq e^{-\frac{|\alpha|^2}{2}}\sum_{k=0}^\infty \frac{\alpha^k}{\sqrt{k!}}\, \ket{k}\, .
\label{coherent}
\ee
Extending these definitions to multimode systems is \tcr{quite straightforward}. For $k = (k_1,\ldots, k_m)^\intercal\in \N^m$, one sets $\ket{k}\coloneqq \bigotimes_{j=1}^m \ket{k_j}$; analogously, for $\Alpha = (\alpha_1,\ldots, \alpha_m)^\intercal \in \C^m$, a multimode coherent state is defined by $\ket{\alpha} \coloneqq \bigotimes_{j=1}^m \ket{\alpha_j}$.

The \textbf{displacement operators} form a special family of unitary operators acting on $\HH_m$. For $\Alpha\in \C^m$, they are defined by
\bb
\disp(\Alpha)\coloneqq \exp \left[ \sumno_{j=1}^m \left(\alpha_j a_j^\dag - \alpha_j^* a_j\right) \right] .
\label{disp}
\ee
They satisfy the identity
\bb
\disp(\alpha) \disp(\beta) = e^{\frac12\left(\alpha^\intercal \beta^* - \alpha^\dag \beta\right)}\,\disp(\alpha+\beta)\, ,
\label{Weyl}
\ee
called the \textbf{Weyl form of the canonical commutation relations}, for all $\alpha,\beta\in \C^m$, and they yield coherent states upon acting on the vacuum, i.e.,
\bb
\disp(\alpha) \ket{0} = \ket{\alpha} \qquad \forall\ \alpha\in \C^m\, .
\label{disp_vacuum}
\ee
For an arbitrary trace class operator $T\in \T{m}$, its \textbf{characteristic function} $\chi_T:\C^m\to \C$ is given by
\bb
\chi_T(\alpha) \coloneqq \Tr [T\disp(\alpha)]\, .
\label{chi}
\ee
For a $m$-mode quantum state $\rho\in\D{m}$, a quantity which is intimately related to its characteristic function is the \textbf{Husimi $\mathbf{Q}$-function} $Q:\C^m\to \C$, defined by $Q_\rho(\alpha) \coloneqq \frac{1}{\pi^m} \braket{\alpha|\rho|\alpha}$~\cite{Husimi}.

\subsection{Entropies and relative entropies}\label{entropies subsection}

The (von Neumann) \textbf{entropy} of some positive semidefinite trace class operator $A\in \tsa^+(\hilb)$ can be defined as
\bb
S(A) \coloneqq - \Tr\left[A \log_2 A \right] .
\label{entropy}
\ee
Note that this is a well-defined although possibly infinite quantity. One way to make sense of the expression~\eqref{entropy} is via the infinite sum $S(A) = \sum_i (- a_i \log_2 a_i)$, where $A=\sum_i a_i \ketbra{a_i}$ is the spectral decomposition of $A$\, where we convene that $0\log_2 0=0$.
Since $a_i\tends{}{i\to\infty} 0$ because $A$ is trace class, the terms of the above sum are eventually positive. Hence, the sum itself can be assigned a well-defined value, possibly $+\infty$. An alternative approach is to define is the \textbf{Wehrl entropy} instead:
\bb
S_W(\rho)\coloneqq -\int d^{2m}\alpha\,Q_\rho(\alpha)\log_2 \left(\pi^m\, Q_\rho(\alpha)\right)\,.
\label{Wehrl}
\ee
It is well known that $S_W(\rho)\geq S(\rho)$ for any quantum state $\rho\in\D{}$.

The \textbf{relative entropy} between two positive $A,B\in \tsa^+(\hilb)$ is usually written as~\cite{Umegaki1962, PETZ-ENTROPY}
\begin{equation} \label{rel ent}
    D(A\|B) \coloneqq \Tr\left[A (\log_2 A - \log_2 B)\right] .
\end{equation}
Again, the above expression is well defined and possibly infinite~\cite{Lindblad1973}. To see why, we represent it as the infinite sum $D(A \| B) \coloneqq \sum_{i,j} \left| \braket{a_i | b_j}\right|^2 \left( a_i \log_2 a_i - a_i \log_2 b_j + \log_2(e) (b_j - a_i) \right) + \log_2(e) \Tr[A-B]$, where $A = \sum_i a_i \ketbra{a_i}$ and $B = \sum_j b_j \ketbra{b_j}$ are the spectral decompositions of $A$ and $B$, respectively, and we assume that only terms with $a_i>0$ and $b_j>0$ are included. As above, we follow the convention of setting $0\log_2 0=0$, and we set $D(A\|B)=+\infty$ if there exist two indices $i$ and $j$ with $a_i>0$, $b_j=0$, and $\braket{a_i | b_j}\neq 0$. As detailed in~\cite{Lindblad1973}, the convexity of $a\mapsto a\log_2 a$ implies that all terms of the above infinite sum are non-negative, making the expression well defined. In light of the above discussion, it is not difficult to realize that a necessary condition for $D(A \|B)$ to be finite is that $\supp A \subseteq \supp B$. Thus, up to projecting everything onto a subspace we will often assume that $B$ is faithful, i.e., that $B>0$. \tcb{The relative entropy can be endowed with an operational interpretation in the context of asymmetric hypothesis testing~\cite{Hiai1991, Buscemi2019, Wang2019}.}

An alternative approach to the quest for defining a quantum relative entropy could be that of bringing the problem back to the classical setting by means of quantum measurements. Namely, for a state $\rho$ and a measurement $\M = \{E_x\}_{x\in \X}$, we define the associated outcome probability distribution on $\X$ as $P^{\M}_\rho (x)\coloneqq \Tr\left[ \rho E_x\right]$. Remembering that for two classical probability distributions $p$ and $q$ the \textbf{Kullback--Leibler divergence} is given by $\DKL (p\|q)\coloneqq \sum_x p_x (\log_2 p_x - \log_2 q_x)$~\cite{Kullback-Leibler}, let us define the \textbf{measured relative entropy} between any two states $\rho$ and $\sigma$ as~\cite{Donald1986, Hiai1991}
\bb
\DM (\rho\|\sigma) \coloneqq \sup_{\M} \DKL \left( P^{\M}_\rho \big\| P^{\M}_\rho\right) .
\label{measured re}
\ee
It is known that $\DM(\rho\|\sigma)\leq D(\rho\|\sigma)$ for all pairs of states $\rho,\sigma$~\cite{Donald1986}. Recently, extending a result by Petz~\cite{Petz-old}, Berta et al.\ have shown that for finite-dimensional systems equality holds if and only if $[\rho,\sigma]=0$~\cite{Berta2017}.

\section{Quantum Resource Theories}\label{qrt section}

In this section we introduce a general notion of quantum resource theory, and some related concepts and results. Note that our definition is slightly different from that in the recent review by Chitambar and Gour~\cite[Definition~1]{quantum-resource-theories-review}, in that we require also parallel composition (i.e., tensor product) of free operations to be free.

\begin{Def} \label{RT}
A \textbf{quantum resource theory} (QRT) is a pair $\mathcal{R} = \left( \mathcal{S}, \mathcal{F} \right)$, where $\mathcal{S}$ is a family of quantum systems that is closed under tensor products, in the sense that $A,B\in \mathcal{S}$ implies that $AB \coloneqq A\otimes B \in \mathcal{S}$; and contains the trivial system $1$ with Hilbert space $\C$, while $\mathcal{F}$, called the set of free operations, is a mapping that assigns to every pair of systems $A,B\in \mathcal{S}$ a set of channels from system $A$ to $B$. Such a set will be denoted with $\mathcal{F} (A\to B)$.\footnote{Here, classical registers are thought of as special $d$-dimensional quantum systems $X$ ($d<\infty$) with the property that any free operation $\Lambda\in \mathcal{F}(A\to XB)$ satisfies that $(\Delta_X\otimes I_B)\circ \Lambda=\Lambda$, with $\Delta_X(\cdot)\coloneqq \sum_{i=1}^{d} \ketbra{i}(\cdot)\ketbra{i}$ being the dephasing map. In other words, if a quantum system plays the role of a classical register, then no state apart from those that are diagonal in a fixed orthonormal basis $\{\ket{i}\}_{i=1,\ldots, d}$ is ever accessible with free operations.}
We will require that the following three consistency conditions are satisfied:
\begin{enumerate}[(i)]
\item for all $A\in \mathcal{S}$, the identity is a free operation on $A$, in formula $I_A \in \mathcal{F}(A\to A)$;
\item free operations are closed under sequential compositions, namely, if $A,B,C\in \mathcal{S}$ and $\Lambda\in\mathcal{F}(A\to B)$, $\Gamma\in \mathcal{F}(B\to C)$, then also $\Gamma\circ \Lambda\in \mathcal{F}(A\to C)$;
\item free operations are closed under parallel compositions, namely, if for $j=1,2$ one chooses $A_j,B_j\in \mathcal{S}$ and $\Lambda_j\in\mathcal{F}(A_j\to B_j)$, then also $\Lambda_1\otimes \Lambda_2 \in \mathcal{F}(A_1\otimes A_2\to B_1\otimes B_2)$.
\end{enumerate}
If every system in $\mathcal{S}$ is finite dimensional, we will say that $\mathcal{R}$ itself is \textbf{finite dimensional}.
\end{Def}

\subsection{Monotones}

Given a QRT $\mathcal{R}$ as above, one defines the set of free states on the system $A\in \mathcal{S}$ as
\bb
\mathcal{F}_S(A) \coloneqq \mathcal{F} (1\to A)\, .
\label{free_states}
\ee
Clearly, if partial traces are free, then $\Tr_A\left[\mathcal{F}_S(AB)\right]\subseteq \mathcal{F}_S(B)$. A central role in our paper is played by resource quantifiers, i.e., monotones. We define them as follows.  

\begin{Def} \label{monotones}
Let $\mathcal{R} = \left( \mathcal{S}, \mathcal{F} \right)$ be a resource theory. A mapping $G$ assigning to each $A\in \mathcal{S}$ a function $G_A:\D{A}\to [0,+\infty]$ on the set of states on $A$ that takes on values in the extended reals $[0,+\infty]$ is called a \textbf{resource monotone} --- or simply a \textbf{monotone} --- if
\begin{enumerate}[(i)]
\item $G_B \left( \Lambda(\rho)\right) \leq G_A (\rho)$ holds for all states $\rho$ on $A\in \mathcal{S}$ and for all free operations $\Lambda\in \mathcal{F}(A\to B)$, where $B\in \mathcal{S}$ is arbitrary;
\item $G_A(\sigma)=0$ for all $\sigma\in\mathcal{F}_S(A)$, with $\mathcal{F}_S(A)$ defined by~\eqref{free_states}.
\end{enumerate}
A monotone $G$ is said to be:
\begin{enumerate}[(a)]
\item \textbf{faithful}, if $G_A(\rho)=0$ implies that $\rho\in \mathcal{F}_S(A)$;
\item \textbf{convex}, if all functions $G_A$ are convex, i.e., $G_A\left( \sumno_i p_i \rho_i \right) \leq \sum_i p_i G_A(\rho_i)$ for all $A\in \mathcal{S}$ and all statistical ensembles $\{p_i, \rho_i\}$ on $A$;
\item \tcb{\textbf{asymptotically continuous}~\cite{Vidal2000, Donald2002, Horodecki2001, Synak2006} on some subsets of systems $\mathcal{S}'\subseteq \mathcal{S}$, if for all $A\in \mathcal{S}'$ we have that $\dim \HH_A<\infty$, and moreover there exist two continuous functions $f,g:[0,1]\to \R$ \tcb{independent of $\dim \HH_A$} such that $f(0) = g(0) = 0$ and
\bb
\left| G_A(\rho) - G_A(\sigma) \right| \leq f(\epsilon) \log(\dim \HH_A) + g(\epsilon)
\ee
for all $\rho,\sigma\in \D{A}$ at trace distance $\epsilon\coloneqq \frac12 \left\|\rho-\sigma\right\|_1$.}
\item \textbf{lower semicontinuous}, if $G_A$ is lower semicontinuous as a function on $\D{A}$ for all $A\in \mathcal{S}$, i.e., if $\lim_{n\to\infty}\left\|\rho_n-\rho\right\|_1=0$ for a sequence of states on $A$ implies that $\liminf_{n\to\infty} G_A(\rho_n)\geq G_A(\rho)$;
\item \textbf{strongly superadditive}, if $G_{AB}(\rho_{AB}) \geq G_A(\rho_A) + G_B(\rho_B)$ holds for all $A,B\in \mathcal{S}$ and for all states $\rho_{AB}\in \D{AB}$;
\item \textbf{superadditive}, if $G_{AB}(\rho_A \otimes \sigma_B) \geq G_A(\rho_A) + G_B(\sigma_B)$ for all $A,B\in \mathcal{S}$ and for all states $\rho_{A}\in \D{A}$ and $\sigma_B\in \D{B}$;
\item \textbf{weakly superadditive}, if $G_{A_1\ldots A_n}\left( \rho^{\otimes n}\right) \geq n\, G_A(\rho)$ for all $A\in \mathcal{S}$, \tcb{for all $n$} and all states $\rho\in \D{A}$, where $A_1\ldots A_n$ denotes the joint system formed by $n$ copies of $A$;
\item \textbf{additive}, if $G_{AB}(\rho_A \otimes \sigma_B) = G_A(\rho_A) + G_B(\sigma_B)$ for all $A,B\in \mathcal{S}$ and for all states $\rho_{A}\in \D{A}$ and $\sigma_B\in \D{B}$;
\item[(j)] \textbf{weakly additive}, if $G_{A_1\ldots A_n}\left( \rho^{\otimes n}\right) = n G_A(\rho)$ for all $A\in \mathcal{S}$ and all states $\rho\in \D{A}$, where $A_1\ldots A_n$ denotes the joint system formed by $n$ copies of $A$.
\end{enumerate} 
\end{Def}

\begin{rem}
In Definition~\ref{monotones}, $\text{(e)}\Rightarrow\text{(f)}\Rightarrow\text{(g)}$ and $\text{(h)}\Rightarrow\text{(j)}$.
\end{rem}

\begin{rem}
Any monotone is automatically invariant under free unitaries whose inverse is also free.
\end{rem}

\begin{rem}
The notions of \tcb{upper semicontinuous,} strongly subadditive, or subadditive monotone are obtained by reversing the inequalities and exchanging $\liminf$ with $\limsup$ in~(d),~(e), and~(f) of Definition~\ref{monotones}.
\end{rem}

\begin{note}
In what follows, with a slight abuse of notation we will often drop the subscript of $G$ specifying the system it refers to, and think of a monotone $G$ as a function defined directly on the collection of states on all possible systems $A\in \mathcal{S}$.
\end{note}

\subsection{Transformation rates}

We continue by recalling the definition of asymptotic transformation rate.

\begin{Def} \label{rate_def}
Let $(\mathcal{S},\mathcal{F})$ be a QRT. For any two systems $A,B\in\mathcal{S}$ and any two states $\rho_A\in\D{A}$ and $\sigma_B\in\D{B}$, the corresponding \textbf{(standard) asymptotic transformation rate} is given by
\bb
R(\rho_{A} \to \sigma_{B}) \coloneqq \sup\left\{r: \lim_{n\to\infty} \inf_{\Lambda_n\in\mathcal{F}\left( A^n\to B^{\floor{rn}}\right)} \left\|\Lambda_n\left( \rho_{A}^{\otimes n}\right) - \sigma_{B}^{\otimes \floor{r n}} \right\|_1 = 0 \right\} ,
\label{rate}
\ee
where $A^n$ denotes the system composed of $n$ copies of $A$. Any number $r>0$ in the set on the right-hand side of~\eqref{rate} is called a (standard) \textbf{achievable rate} for the transformation $\rho_{A} \to \sigma_{B}$.
\end{Def}

The above definition captures the intuitive notion of maximum yield of copies of the target state $\sigma_B$ that can be obtained per copy of the initial state $\rho_A$ by means of free operations and with asymptotically vanishing error. In Definition~\ref{rate_def}, we have measured the error using the global trace distance. However, it is possible and sometimes even reasonable to modify the error criterion. For instance, in a situation where the output copies are distributed to noninteracting parties, what is relevant is the maximum local error rather than the global one. This train of thought inspires the following definition.

\begin{Def} \label{max_rate_def}
Let $(\mathcal{S},\mathcal{F})$ be a QRT. For any two systems $A,B\in\mathcal{S}$ and any two states $\rho_A\in\D{A}$ and $\sigma_B\in\D{B}$, the corresponding \textbf{maximal asymptotic transformation rate} is given by
\bb
\widetilde{R}(\rho_{A} \to \sigma_{B})\coloneqq \sup\left\{ r : \lim_{n\to\infty} \inf_{\Lambda_n\in\mathcal{F}\left( A^n\to B^{\floor{rn}}\right)} \max_{j=1,\ldots, \floor{rn}} \left\| \left( \Lambda_n \left(\rho_{A}^{\otimes n}\right) \right)_{j} - \sigma_{B} \right\|_1 = 0 \right\} ,
\label{max_rate}
\ee
where for a state $\Omega\in \D{B^{k}}$ defined on $k$ copies of $B$ we defined $\Omega_j\coloneqq \Tr_{B^k\setminus B_j} \left[\Omega\right]\in \D{B_j}$ as the reduced state on the $j^{\text{th}}$ subsystem. Any number $r>0$ in the set on the right-hand side of~\eqref{rate} is called a \textbf{maximally achievable rate} for the transformation $\rho_{A} \to \sigma_{B}$.
\end{Def}

It is immediate to see that for any given pair of states the maximal rate always upper bounds the corresponding standard rate (Lemma~\ref{R_smaller_Rtilde_lemma}).

\subsection{Infinite-dimensional quantum resource theories}

The prime example of a quantum resource theory is naturally that of entanglement~\cite{Horodecki-review, Bennett-distillation, Bennett-distillation-mixed, quantum-resource-theories-review}. In spite of their central importance, very little is known about many fundamental operational questions in the infinite-dimensional case~\cite{introeisert, Eisert2002, Shirokov2016, Shirokov-sq}, with a partial exception being the theory of Gaussian entanglement~\cite{Werner01, Giedke01, giedkemode, adesso07, revisited}. We can formally define entanglement as a resource theory as follows.
\begin{Def}
The resource theory of bipartite \textbf{entanglement} is defined by setting:
\begin{itemize}
\item $\mathcal{S}$ to be the family of all (possibly infinite-dimensional) quantum systems $A\!:\!B$, where the colon indicates the bipartition in separate parties;
\item $\mathcal{F}_S(A\!:\!B) = \text{\emph{Sep}}(A\!:\!B) \coloneqq \overline{\co}\left\{ \ketbra{\phi_A}\otimes\ketbra{\psi_B}:\,\ket{\phi_A}\in\HH_A,\,\ket{\psi_B}\in\HH_B \right\}$ for any bipartite system $A\!:\!B\in\mathcal{S}$, \tcb{where $\overline{\co}$ denotes the closed (in trace norm topology) convex hull of a set};
\item $\mathcal{F}(A\!:\!B\to A'\!:\!B')=\text{\emph{LOCC}}(A\!:\!B\to A'\!:\!B')$ is the set of all \emph{LOCC} protocols from $A\!:\!B$ to $A'\!:\!B'$.
\end{itemize}
\end{Def}

Another important example is the resource theory of quantum thermodynamics~\cite{Brandao-thermo, CVThermo1, CVThermo2}. Just as that of entanglement, it can be constructed for finite-dimensional systems as well. However, in accordance with the spirit of this work, we will focus on the continuous variable case from now on.

\begin{Def}
The resource theory of \textbf{quantum thermodynamics} is defined by setting:
\begin{itemize}
\item $\mathcal{S}$ to be the family of all (possibly infinite-dimensional) quantum systems $A$ equipped with a Hamiltonian $H_A$ satisfying the Gibbs hypothesis, i.e., $\Tr\left[\exp^{-\beta H_A}\right]<\infty$ for any inverse temperature $\beta>0$; we also assume $H_{AB}=H_A+H_B$ for any systems $A,B\in\mathcal{S}$;
\item once an inverse temperature $\beta$ has been fixed for all systems, $\mathcal{F}_S(A) = \{\gamma_A\}$, with
\bbb
\gamma_A\coloneqq\frac{\exp^{-\beta H_A}}{\Tr\left[ \exp^{-\beta H_A}\right]}
\eee
being the \textbf{thermal state}, for any system $A\in\mathcal{S}$ with Hamiltonian $H_A$;
\item $\mathcal{F}(A\to B)$ to encompass all quantum channels $\Lambda:\T{A}\to\T{B}$ such that $\Lambda(\gamma_A)=\gamma_B$ for systems $A,B\in\mathcal{S}$ with thermal states $\gamma_A$ and $\gamma_B$ respectively (\textbf{Gibbs-preserving operations}).
\end{itemize}
\end{Def}

In the case where the family $\mathcal{S}$ contains only continuous variable quantum systems, other specific resource theories emerge naturally, as a result of operational or technological constraints. For example, the resource theory of optical nonclassicality~\cite{Sperling2015, Tan2017, Yadin2018, NC-review, taming-PRA, taming-PRL} is based on the premise that statistical mixtures of coherent states
are easy to synthesize, hence free, and ``classical'', as they most closely approximate classical electromagnetic waves. On the other hand, operationally, nonclassical states, such as Fock states~\cite{KLM, single-photon}, squeezed states~\cite{Kennard1927, Walls1983, Slusher1985, Andersen2016, Schnabel2017}, cat states~\cite{Dodonov1974, Ralph2003, Lund2008, Joo2011, Facon2016, Brask2010, Lee2013, Sychev2017}, or NOON states~\cite{Sanders1989, Dowling2008}, play an increasingly central role in applications. A formal definition of this resource theory is as follows.

\begin{Def}
The resource theory of (\textbf{optical}) \textbf{nonclassicality} is defined by setting:
\begin{itemize}
\item $\mathcal{S}$ to be the family of all continuous variable quantum systems;
\item $\mathcal{F}_S(A) = \CC_m \coloneqq \overline{\co}\left\{ \ketbra{\Alpha}:\, \Alpha\in \C^m\right\}$ for any $m$-mode system $A\in\mathcal{S}$;
\item $\mathcal{F}(A\to B)$, with $A,B\in\mathcal{S}$ being $m$ and $m'$-mode systems respectively, to encompass all quantum channels $\Lambda:\T{m}\to\T{m'}$ such that $\Lambda(\CC_m)\subseteq \CC_{m'}$ (\textbf{classical operations}).
\end{itemize}
\end{Def}

The so-called classical operations comprise, but are possibly not limited to, channels that can be obtained through passive linear optics, destructive measurements, and feed-forward of measurement outcomes~\cite{Tan2017, Yadin2018}. Note that the set of classical states and that of classical channels are both convex.

\section{Main results} \label{main results section}

In the present section we state the main results of our work. The starting point is the general bound on asymptotic transformation rates in Theorem~\ref{general_bound_rates_thm}. We then explore its consequences for the resource theories of entanglement and quantum thermodynamics in Corollaries~\ref{bound_rates_Esq_cor} and~\ref{bound_rates_free_energy_cor}, respectively. To apply it to the resource theory of nonclassicality, instead, we need to introduce and study two new monotones, the measured relative entropy of nonclassicality $\NCM$ and its regularized version (Definition~\ref{NC_measures_def}). One of our main technical contributions is the proof of a powerful variational expression for $\NCM$ (Theorem~\ref{main result thm}), from which we deduce the lower semicontinuity and --- most importantly --- the strong superadditivity of the measured relative entropy of nonclassicality $\NCM$ (Theorem~\ref{bound_rates_NCMi_thm}). 


\subsection{General bound on asymptotic rates}\label{bound on rates subsection}

Our first result is a general bound on (maximal) transformation rates that works for all quantum resource theories, \tcb{including infinite-dimensional ones. Formally, it is a generalization of a well-known bound holding for many monotones in finite dimensions~\cite[Section~VI.A.5]{quantum-resource-theories-review}. In this simpler context, it is possible to prove that if a resource monotone $G$ is weakly additive and moreover asymptotically continuous (Definition~\ref{monotones}, (j)~and~(c)) then the asymptotic transformation rate $R(\rho\!\to\!\sigma)$ between any two states $\rho,\sigma$ (Definition~\ref{rate_def}) satisfies that
\bb
R(\rho\!\to\!\sigma)\leq \frac{G(\rho)}{G(\sigma)}\, ,
\label{rate_bound_finite_dim}
\ee
\tcr{provided that the right-hand side is well defined. This bound has been proved in~\cite[Theorem~4]{Horodecki2001} (see also~\cite[Section~VI.A.5]{quantum-resource-theories-review}) using techniques developed in~\cite{Horodecki2000} and~\cite{Donald2002} --- see especially~\cite[Propositions~19,~20, and~22]{Donald2002} for a thorough discussion of many possible variations of the underlying hypotheses. In fact, the proof is so simple and enlightening that it is worth summarising it here. For any achievable rate $r$, i.e., for any element of the set in~\eqref{rate}, calling $A,B$ the systems to which $\rho,\sigma$ pertain, we can construct the sequence of maps $\Lambda_n\in\mathcal{F}\left( A^n\to B^{\floor{rn}}\right)$ such that $\epsilon_n \coloneqq \frac12 \left\|\Lambda_n\left( \rho^{\otimes n}\right) - \sigma^{\otimes \floor{r n}} \right\|_1$ satisfies $\lim_{n\to\infty} \epsilon_n = 0$. Leveraging properties~(j) and~(c) in Definition~\ref{monotones}, one then obtains that
\bb
n\, G(\rho) \texteq{(j)} G\left(\rho^{\otimes n}\right) \geq G\left(\Lambda_n\left( \rho^{\otimes n}\right)\right) \textgeq{(c)} G\big(\sigma^{\otimes \floor{rn}}\big) - f(\epsilon_n) \log \big( d^{\floor{rn}} \big) - g(\epsilon_n) \texteq{(j)} \floor{rn} \left(G(\sigma) - f(\epsilon_n) \log d \right) - g(\epsilon_n)\, ,
\label{finite_dim_proof}
\ee
where we called $d\coloneqq \dim \HH_B$. Dividing by $n$ and taking the limit $n\to\infty$ yields $r\leq G(\rho)/G(\sigma)$, and in turn~\eqref{rate_bound_finite_dim} once one takes the supremum over $r$.

Inequality~\eqref{rate_bound_finite_dim}, applied to different weakly additive and asymptotically continuous monotones, yields many of the commonly employed bounds on rates as far as finite-dimensional resources are concerned.}\footnote{\tcr{Notable exceptions are monotones based on the partial transpose in entanglement theory~\cite{negativity, plenioprl, irreversibility}.}} Unfortunately, for infinite-dimensional systems the notion of asymptotic continuity 
\tcr{becomes empty, and consequently the argument in~\eqref{finite_dim_proof} breaks down at the step marked as~(c). In fact, prior to our work \tcb{there seemed to be a lack of} technical tools} to address the approximation error allowed in the transformation~\eqref{rate} \tcr{in the infinite-dimensional case}. 
Due to this ``asymptotic continuity catastrophe'', in the existing literature prior to our work we could not locate \emph{any} upper bound on asymptotic transformation rates that holds for infinite-dimensional resource theories. The theorem below, whose proof is remarkably simple \tcr{(and yet very different from that in~\eqref{finite_dim_proof})} but whose applicability is surprisingly wide, remedies this regrettable state of affairs at least in the case where the employed monotone $G$ is strongly superadditive. An alternative but ultimately less satisfactory approach \tcr{to circumvent the asymptotic continuity catastrophe is sketched out} in Appendix~\ref{restr_asymp_cont_sec}, to which we refer the reader interested in more details on this point.}


\begin{thm} \label{general_bound_rates_thm}
For a given QRT, not necessarily finite dimensional, let $G$ be a monotone that is strongly superadditive, weakly additive, and lower semicontinuous. Then, for all states $\rho_A,\sigma_B$, it holds that
\bb
\label{bound on rates equation}
R(\rho_A\!\to\!\sigma_B)\leq \widetilde{R}(\rho_A\!\to\!\sigma_B) \leq \frac{G(\rho_A)}{G(\sigma_B)} \,,
\ee
whenever the rightmost side is well defined.
\end{thm}

With the above result at hand, we can now establish rigorous bounds on transformation rates in operationally important examples of infinite-dimensional quantum resource theories.


\subsection{First consequences: resource theories of entanglement and quantum thermodynamics}

We start with the QRT of entanglement. To the extent of our knowledge, there is no available technique to derive upper bounds on the transformation rate $R(\rho_{AB}\to \sigma_{AB})$ in terms of known monotones. Even the energy-constrained version of asymptotic continuity established by Shirokov~\cite{Shirokov-AFW-1, Shirokov-AFW-2, Shirokov-AFW-3} for many entanglement monotones does not suffice to this purpose. This is because we need continuity estimates on the output system, and --- while the input, consisting of many copies of a known state, is naturally energy constrained --- the output, being produced by a general unconstrained free channel, is not. We could of course impose such an energy constraint artificially, by enforcing the parties to operate with LOCCs that are \emph{uniformly} energy-constrained; however, the operational motivation behind this assumption is somewhat dubious; this is especially so if energy is much cheaper than entanglement, which is often the case in experimental practice. We refer the reader to Appendix~\ref{restr_asymp_cont_sec} for a more in-depth discussion of these points.

To apply Theorem~\ref{general_bound_rates_thm} to the case at hand, we need an entanglement monotone that obeys strong superadditivity. The \emph{squashed entanglement}, denoted $E_{sq}$, is a natural candidate~\cite{tucci1999, squashed, faithful, rel-ent-sq}. Shirokov~\cite{Shirokov-sq} has shown how to extend its definition to infinite-dimensional systems~\cite[Eq.~(17)]{Shirokov-sq}. \tcb{We report the definition of squashed entanglement later in Section~\ref{section squashed entanglement}}. Applying Theorem~\ref{general_bound_rates_thm} to it we deduce the following corollary.

\begin{cor} \label{bound_rates_Esq_cor}
Let $\rho_{AB}$ and $\sigma_{A'B'}$ be two bipartite states such that
\bb
\min\left\{ S(\rho_A),\, S(\rho_B),\, S(\rho_{AB})\right\} <\infty\, ,\qquad \min\left\{ S(\sigma_{A'}),\, S(\sigma_{B'}),\, S(\sigma_{A'B'})\right\} <\infty\, .
\label{strange_sets}
\ee
Then, in the QRT of entanglement it holds that
\bb
R(\rho_{AB}\!\to\!\sigma_{A'B'})\leq \widetilde{R}(\rho_{AB}\!\to\!\sigma_{A'B'}) \leq \frac{E_{sq}(\rho_{AB})}{E_{sq}(\sigma_{A'B'})}\, .
\label{bound_rates_Esq}
\ee
\end{cor}

Another possible application of Theorem~\ref{general_bound_rates_thm} is to the QRT of thermodynamics. The quantity $G(\rho_A)\coloneqq \frac{1}{\beta} D(\rho_A\|\gamma_A)$, which coincides with the free energy difference between $\rho_A$ and $\gamma_A$ when $\Tr\left[\rho_A H_A\right]<\infty$, can be seen to be strongly superadditive, additive and lower semicontinuous. We deduce the following.

\begin{cor} \label{bound_rates_free_energy_cor}
In the QRT of thermodynamics, for all states $\rho_{A}, \sigma_{B}$ it holds that
\bb
R(\rho_{A}\!\to\!\sigma_{B})\leq \widetilde{R}(\rho_{A}\!\to\!\sigma_{B}) \leq \frac{D(\rho_A\|\gamma_A)}{D(\sigma_B\|\gamma_B)} \,.
\label{bound_rates_free_energy}
\ee
\end{cor}

Let us stress that Corollary~\ref{bound_rates_free_energy_cor} extends the results of Brand\~{a}o et al.~\cite{Brandao-thermo}, which are valid in finite-dimensional systems, to \emph{all} quantum systems where a QRT of thermodynamics can be constructed. This is clearly a crucial improvement because of the ubiquity of harmonic oscillators in physical applications.

\subsection{Further consequences: resource theory of nonclassicality}

Over the past decades, there have been proposals to quantify the nonclassicality of quantum states of light, e.g., by their distance from the set of classical states~\cite{Hillery1987, Dodonov2000, Marian2004, BrandaoPlenio1}, by the amount of noise needed in order to make them classical~\cite{Lee1991, Luetkenhaus1995}, by their potential for entanglement generation~\cite{Asboth2005, Vogel2014, Killoran2016} or for metrological advantage~\cite{Kwon2019}, by the negativity~\cite{Kenfack2004, Tan2020}, the variances~\cite{Yadin2018} or other features~\cite{Vogel2000, Vogel2000-comment, Vogel2000-reply, Richter2002, idel2016, Bohmann2020} of their phase-space distributions, or by the minimum number of superposed coherent states needed to reproduce the target state~\cite{Gehrke2012}. Unfortunately, none of these monotones appears to yield bounds on asymptotic transformation rates, for they fail to satisfy asymptotic continuity. In fact, to the extent of our knowledge, no rigorous bounds on those rates are known for the resource theory of optical nonclassicality. Indeed, the transformations considered in Yadin et al.~\cite[Theorems~2 and~3]{Yadin2018} are probabilistic but exact, and moreover single-shot rather than asymptotic. One could argue that especially their zero-error nature somewhat limits their operational relevance in applications.

We therefore pursue a different approach. In analogy to what was previously done for entanglement~\cite{Vedral1997,vedral1998entanglement}, we use the relative entropies introduced in Section~\ref{entropies subsection} to construct nonclassicality measures.

\begin{Def} \label{NC_measures_def}
Let $\rho\in \D{m}$ be an $m$-mode state. The \textbf{relative entropy of nonclassicality} and the \textbf{measured relative entropy of nonclassicality} of $\rho$ are defined respectively as:
\bb
\NC(\rho) \coloneqq \inf_{\sigma\in\CC_m}D(\rho\|\sigma)\, ,\quad \NCM(\rho) \coloneqq \inf_{\sigma\in\CC_m}\DM(\rho\|\sigma)\, .
\label{NC_and_NCM}
\ee
\end{Def}

Note that our definition of $\NC$ differs from that of Marian et al.~\cite{Marian2004}, in that $\sigma$ is allowed to be an arbitrary classical state, not necessarily Gaussian. It is not difficult to see that $\NC$ and $\NCM$ are faithful and convex nonclassicality monotones (Lemma~\ref{NC_subadd_lemma}). Since $\NC$ is also subadditive (again, Lemma~\ref{NC_subadd_lemma}), its regularization $\NCi(\rho) \coloneqq \lim_{n\to \infty}\! \frac{\NC(\rho^{\otimes n})}{n}$ is well defined by Fekete's lemma~\cite{Fekete1923} and also subadditive (Corollary~\ref{NCi_monotone_cor}). We will show that both $\NC$ and $\NCi$ are always finite on bounded-energy states (Proposition~\ref{bounded energy prop}), but that there exist infinite-energy states $\rho$ such that $\NC(\rho)=\NCi(\rho)=\infty$ (Proposition~\ref{NC infinite prop}). Explicit computations or tight estimates for the measured relative entropy of nonclassicality for Fock-diagonal states, squeezed states, and cat states are reported in Sections~\ref{Fock-diagonal subsec}--\ref{cat states subsec}.


It might not be clear at this point why to introduce $\NCM$ alongside with $\NC$, given that the former quantity involves one more nested optimization than the latter. However, we now show that its computation can be notably simplified. 

\begin{thm} \label{main result thm}
For all $m$-mode finite-entropy states $\rho$, it holds that
\bb
\NCM(\rho) = \sup_{L>0} \left\{ \Tr \left[\rho \log_2 L\right] - \log_2 \sup_{\Alpha\in \C^m} \braket{\Alpha| L | \Alpha} \right\}\,, \label{second expression for ncm}
\ee
where $L$ ranges over all positive trace class operators on $\HH_m$ (equivalently, on all positive normalized states).
\end{thm}

The proof of Theorem~\ref{main result thm} involves two main ingredients. This first one is a generalization of the variational program for $\DM$ put forth by Berta et al.~\cite[Lemma~1]{Berta2017} to the infinite-dimensional case, which may be of independent interest.

\begin{lemma} \label{Berta variational lemma}
Let $\rho\in \D{}$ be a density operators on a (possibly infinite-dimensional) Hilbert space $\HH$, and let $\sigma\in \Tp{}$ be positive semidefinite and nonzero. Then
\begin{align}
	\DM (\rho\|\sigma) &= \sup_{h\in \B{}} \left\{ \Tr \left[\rho h\right] - \log_2 \Tr\left[ \sigma 2^h\right] \right\} \label{Berta variational 1} \\
	&= \sup_{h\in \B{}} \left\{ \Tr \left[\rho h\right] + \log_2 (e) \left(1 - \Tr \left[\sigma 2^h\right]\right) \right\} \label{Berta variational 2} \\
	&= \sup_{0<\delta \id<L\in \B{}} \left\{ \Tr \left[\rho \log_2 L\right] - \log_2 \Tr \left[\sigma L\right] \right\} \label{Berta variational 3bis} \\
	&= \sup_{0<\delta \id<L\in \B{}} \left\{ \Tr \left[\rho \log_2 L\right] + \log_2 (e) \left( 1 - \Tr \left[\sigma L\right]\right) \right\} \label{Berta variational 4bis} \\
	&= \sup_{0<L\in \B{}} \left\{ \Tr \left[\rho \log_2 L\right] - \log_2 \Tr \left[\sigma L\right] \right\} \label{Berta variational 3} \\
	&= \sup_{0<L\in \B{}} \left\{ \Tr \left[\rho \log_2 L\right] + \log_2 (e) \left(1 - \Tr \left[\sigma L\right]\right) \right\} . \label{Berta variational 4} 
\end{align}
\tcb{The notation $L>\delta \mathbb{1}$ in the supremum in equations~\eqref{Berta variational 3bis} and~\eqref{Berta variational 4bis} means that $L$ is required to have eigenvalues bounded from below by a positive quantity, i.e., to satisfy $L>\delta \mathbb{1}$ for some $\delta>0$, depending on $L$.}
\end{lemma}

\begin{rem}
For the case where the measurements in~\eqref{measured re} are restricted to be projective (i.e., $\M=\{E_x\}_{x\in \X}$ with $E_x$ a projector for all $x$, and $\sumno_x E_x=\id$), the expression in~\eqref{Berta variational 1} has been obtained already by Petz~\cite[Proposition~7.13]{PETZ-ENTROPY}).
\end{rem}

The above Lemma~\ref{Berta variational lemma} is proved in Section~\ref{variational expressions subsection}. By applying it to the program in~\eqref{NC_and_NCM} we are left with a nested optimization of the form $\sup \inf$. Then, the second critical ingredient that is needed to arrive at a proof of Theorem~\ref{main result thm} is an application of Sion's minimax theorem~\cite{sion} that allows us to exchange infimum and supremum in this resulting expression. This is technically challenging, as meeting the compactness hypothesis in Sion's theorem requires a careful choice of topology on the domain of optimization. The crucial technical contribution here is Lemma~\ref{Cm tilde compact lemma}, which establishes the compactness of the set of subnormalized classical states with respect to the weak*-topology (see Section~\ref{subsec_notation}). Along the way, we introduce and study an auxiliary quantity $\Gamma$ (Definition~\ref{Gamma_def} and Proposition~\ref{Gamma_properties_prop}).

An immediate consequence of Theorem~\ref{main result thm} is the superadditivity of $\NCM$ on finite-entropy states. This fact allows us to successfully construct the regularization $\NCMi$.

\begin{cor} \label{superadditivity_cor}
When computed on finite-entropy states, $\NCM$ is lower semicontinuous and strongly superadditive, meaning that
\bb
\NCM(\rho_{AB}) \geq \NCM(\rho_A) + \NCM(\rho_B)\qquad \forall\ \rho_{AB}:\ S(\rho_A), S(\rho_B) < \infty\, .
\ee 
Therefore, for any finite-entropy state $\rho$ its regularization
\bb
\NCMi(\rho)\!\coloneqq \!\lim_{n\to \infty} \!\frac{\NCM(\rho^{\otimes n})}{n}
\label{NCMi}
\ee
is a well defined nonclassicality monotone. It is lower semicontinuous, strongly superadditive, and weakly additive. Furthermore,
\bb
\NCM(\rho) \leq \NCMi(\rho) \leq \NCi(\rho) \leq \NC(\rho)
\label{hierarchy}
\ee
holds whenever $S(\rho)<\infty$. In particular, both $\NCi$ and $\NCMi$ are also faithful, at least on finite-entropy states.
\end{cor}


The variational expression in Theorem~\ref{main result thm} has many more consequences. For example, we use it to establish upper and lower bounds on $\NCM$ and its regularization $\NCMi$ based on the Wehrl entropy (Proposition~\ref{upper and lower bound for ncm prop}), which translate to tight estimates of these quantifiers for Gaussian states (Corollary~\ref{upper and lower bound for ncm for gaussian states cor}). The most important application is however the following.


\begin{thm} \label{bound_rates_NCMi_thm}
Let $\rho,\sigma$ be two CV states with finite entropy, i.e., such that $S(\rho),S(\sigma)<\infty$. Then the transformation rates in the resource theory of nonclassicality obey the inequalities
\bb
R(\rho\!\to\!\sigma)\leq \widetilde{R}(\rho\!\to\!\sigma) \leq \frac{\NCMi(\rho)}{\NCMi(\sigma)} \leq \frac{\NC(\rho)}{\NCM(\sigma)} \,,
\label{bound_rates_NCMi}
\ee
provided that the ratios on the right-hand sides are well defined.
\end{thm}


To the best of our knowledge,~\eqref{bound_rates_NCMi} is the first explicit bound on asymptotic transformation rates in the context of CV nonclassicality. However, it would amount to a rather futile theoretical statement if not complemented with a systematic way of upper bounding the ratio $\NC(\rho)/\NCM(\sigma)$. Note that $\NC$ can be estimated from above by simply making suitable ansatzes in~\eqref{NC_and_NCM}. The a priori less trivial task of lower bounding $\NCM$ can be carried out thanks to Theorem~\ref{main result thm}.

As an immediate application of Theorem~\ref{bound_rates_NCMi_thm}, we consider the paradigmatic example of (Schr\"{o}dinger) cat state manipulation~\cite{Lund2004, Laghout2013, Sychev2017, Sychev2018, Wang2018, Oh2018}. For $\alpha\in \C$, cat states are defined by~\cite{Dodonov1974}
\bb
\ket{\psi_{\alpha}^\pm}\coloneqq \frac{1}{\sqrt{2\left( 1 \pm e^{-2|\alpha|^2}\right)}} \left(\ket{\alpha}\pm \ket{-\alpha}\right) ,
\label{cat}
\ee
where $\ket{\pm \alpha}$ are coherent states~\eqref{coherent}. The transformations we look at are $\psi_\alpha^+\to \psi_{\sqrt2 \alpha}^+$ (amplification) and $\psi_{\sqrt2\alpha}^+\to \psi_{\alpha}^+\otimes \psi_{\alpha}^-$ (sign-randomized dilution). A protocol for amplification using linear optical elements and quadrature measurements has been designed by Lund et al.~\cite{Lund2004}. We present an ameliorated version of it (Proposition~\ref{Ferrari_prop}), together with a simple protocol for sign-randomized dilution (Proposition~\ref{sign_randomized_dilution_cat_prop}). The lower bounds on rates given by these explicit protocols are shown in Figure~\ref{protocols_fig}. The upper bound derived via Theorem~\ref{bound_rates_NCMi_thm} is asymptotically tight for the dilution task, but not in the case of amplification. This is due to the fact that our quantifiers all saturate to $1$ for cat states with $|\alpha|\to\infty$.

\begin{figure}[ht]
\centering
\includegraphics[scale=0.8]{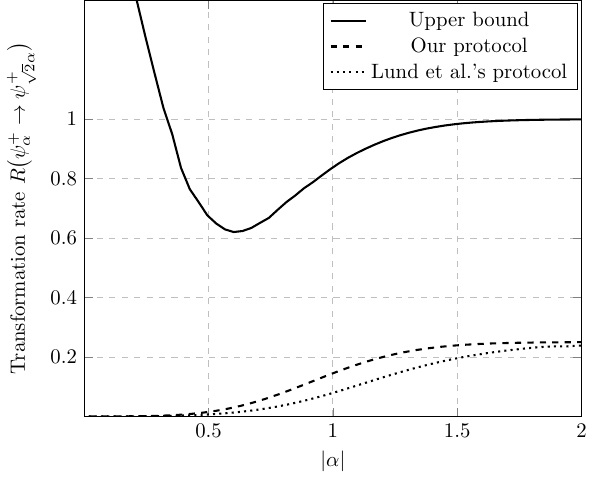}
\includegraphics[scale=0.8]{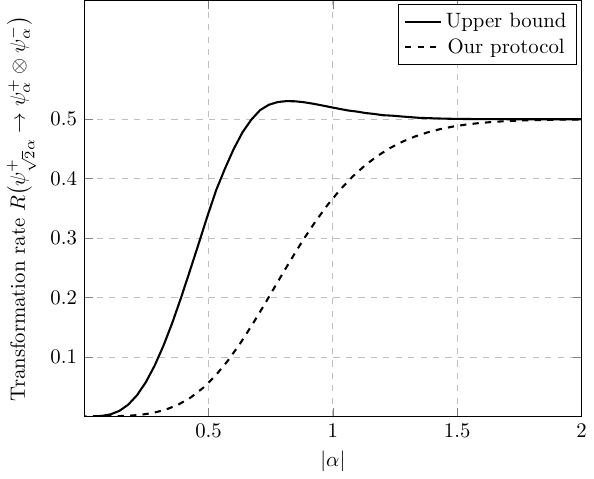}
\caption{Upper and lower bounds on asymptotic transformation rates of Schr{\"o}dinger cat states.} 
\label{protocols_fig}
\end{figure}


\section{Preliminary results} \label{preliminary results section}

Throughout this section we lay the ground for the proof of our main results, studying general properties of monotone regularization (Section~\ref{subsec_monotone_regularization}) and investigating in more detail nonclassicality monotones (Section~\ref{subsec_nonclassicality_monotones}).

\subsection{Generalities about monotone regularization} \label{subsec_monotone_regularization}

It turns out that any monotone $G$ can be made weakly additive by a procedure known as ``regularization''. 

\begin{Def} \label{regularization_def}
Let $\left( \mathcal{S}, \mathcal{F} \right)$ be a QRT equipped with a monotone $G$. Then the functions
\begin{align}
G^{\downarrow,\infty}(\rho) &\coloneqq \liminf_{n\to\infty} \frac1n\, G\left( \rho^{\otimes n}\right) , \label{regularization_down} \\
G^{\uparrow,\infty}(\rho) &\coloneqq \limsup_{n\to\infty} \frac1n\, G\left( \rho^{\otimes n}\right) \label{regularization_up} 
\end{align}
are called the \textbf{lower and upper regularizations} of $G$. On the domain of states $\rho$ such that $G^{\downarrow,\infty}(\rho) = G^{\uparrow,\infty}(\rho)\eqqcolon G^{\infty}(\rho)$ one can speak of a unique \textbf{regularization} $G^\infty$.
\end{Def}

The following result is immediate from the definition.

\begin{lemma} \label{regularization_lemma}
Let $\left( \mathcal{S}, \mathcal{F} \right)$ be a QRT equipped with a monotone $G$. Then the lower and upper regularizations $G^{\downarrow,\infty}$ and $G^{\uparrow,\infty}$ given by Definition~\ref{regularization_def} are also monotones. Moreover, $G^\infty$ is weakly additive on its domain, i.e., $G^{\downarrow,\infty}(\rho) = G^{\uparrow,\infty}(\rho)$ for a state $\rho$ implies that $G^{\infty}(\rho^{\otimes n})\equiv n\, G(\rho)$ for all $n\in \N_+$.
\end{lemma}

\begin{proof}
Let us start by showing that, e.g., $G^{\downarrow,\infty}$ is a monotone. Since parallel composition of free operations is free, for all $\rho\in \D{A}$ and for all $\Lambda\in \mathcal{F}\left(A\to B\right)$, with $A,B\in \mathcal{S}$, we obtain that
\bbb
G^{\downarrow,\infty}\left( \Lambda(\rho) \right) = \liminf_{n\to\infty} \frac1n\, G\left( \Lambda(\rho)^{\otimes n}\right) = \liminf_{n\to\infty} \frac1n\, G\left( \Lambda^{\otimes n} \left(\rho^{\otimes n}\right)\right) \leq \liminf_{n\to\infty} \frac1n\, G\left( \rho^{\otimes n} \right) = G^{\downarrow,\infty} (\rho)\, .
\eee
Moreover, if $\rho$ is free, also $\rho^{\otimes n}$ is so, and hence $G^{\downarrow, \infty}(\rho)=0$ as well. This proves the first claim.

Now, by definition $G^{\downarrow,\infty}(\rho) = G^{\uparrow,\infty}(\rho)$ implies that the sequence $\left(\tcb{\frac{1}{k}} G(\rho^{\otimes k}) \right)_{k\in \N_+}$ has a limit. If that is the case, then clearly $G^{\infty}(\rho^{\otimes n}) = \lim_{k\to \infty} \frac1k\, G\left( \rho^{\otimes kn}\right) = n \lim_{n\to\infty} \frac{1}{kn}\, G\left( \rho^{\otimes kn}\right) = n \, G^\infty(\rho)$ for all $n\in \N_+$.
\end{proof}

A useful fact that is slightly less obvious is as follows.

\begin{lemma} \label{regularization_superadd_lemma}
Let $\left( \mathcal{S}, \mathcal{F} \right)$ be a QRT equipped with a monotone $G$ that is weakly superadditive. Then:
\begin{enumerate}[(i)]
\item the regularization $G^\infty$ in Definition~\ref{regularization_def} exists for all states $\rho$, i.e., $G^{\downarrow,\infty}(\rho) = G^{\uparrow,\infty}(\rho)\eqqcolon G^\infty(\rho)$ for all $\rho\in \D{A}$ with $A\in\mathcal{S}$; it is also weakly additive and satisfies $G^\infty\tcb{\geq}G$;
\item If $G$ is also (strongly) superadditive, then $G^{\infty}$ is (strongly) superadditive as well;
\item If $G$ is lower semicontinuous, then so is $G^\infty$.
\end{enumerate}
\end{lemma}

\begin{rem}
The above result is still valid if we replace superadditivity with subadditivity, lower semicontinuity with upper semicontinuity, and reverse all inequalities.
\end{rem}

\begin{proof}[Proof of Lemma~\ref{regularization_superadd_lemma}]
Due to weak superadditivity, for all states $\rho$ the sequence $\left( a_n\right)_{n\in \N_+}$ defined by $a_n \coloneqq G(\rho^{\otimes n})$ is superadditive, meaning that $a_{n+m}\geq a_n+a_m$. Therefore, by Fekete's lemma~\cite{Fekete1923} $\lim_{n\to\infty}\frac{a_n}{n}$ exists, and it satisfies that $\lim_{n\to\infty}\frac{a_n}{n} = \sup_{n\in \N_+} \frac{a_n}{n}$. Therefore,
\bbb
G^\infty(\rho) = \lim_{n\to\infty} \frac1n\, G\left( \rho^{\otimes n} \right) = \sup_{n\in \N_+} \frac1n\, G\left( \rho^{\otimes n} \right)
\eee
is well defined for all $\rho$, and satisfies $G^\infty(\rho)\geq G(\rho)$. This proves (i).

Now we proceed to prove points (ii) and (iii). We already saw in Lemma~\ref{regularization_lemma} that \tcb{$G^\infty$} is a weakly additive monotone, so it suffices to show that it is (strongly) superadditive if $G$ was such. This is immediate to establish (we prove it only for strong superadditivity, as the superadditivity case is completely analogous):
\bbb
G^\infty(\rho_{AB}) = \lim_{n\to\infty} \frac1n\, G\left( \rho_{AB}^{\otimes n} \right) \geq \lim_{n\to\infty} \frac1n\, \left( G\left( \rho_{A}^{\otimes n} \right) + G\left( \rho_{B}^{\otimes n} \right) \right) = \lim_{n\to\infty} \frac1n\, G\left( \rho_{A}^{\otimes n} \right) + \lim_{n\to\infty} \frac1n\, G\left( \rho_{B}^{\otimes n} \right) .
\eee
To see that $G^\infty$ is lower semicontinuous if so was $G$, just notice that $G^\infty(\rho) = \sup_{n\in \N_+} \frac1n\, G\left( \rho^{\otimes n} \right)$ is the pointwise supremum of lower semicontinuous functions and thus must itself be lower semicontinuous.
\end{proof}

\subsection{Nonclassicality monotones} \label{subsec_nonclassicality_monotones}

If the reader is worried by the proliferation of regularized measures in Definition~\ref{NC_measures_def}, they should not be. In fact, we will show that the regularizations are unique in all physically interesting cases. We are able to readily prove the equality between $N_r^{\downarrow,\infty}$ and $N_r^{\uparrow,\infty}$, while a proof for $N_r^{M,\downarrow,\infty}$ and $N_r^{M,\uparrow,\infty}$ will be given at the end of Section~\ref{main results subsection}. The first step is to prove that the quantity we just defined are actually good resource monotones.

\begin{lemma} \label{NC_subadd_lemma}
The quantities $\NC$ and $\NCM$ are faithful and convex nonclassicality monotones. They obey the inequality $\NC\geq \NCM$. Moreover, $\NC$ is subadditive.
\end{lemma}

\begin{proof}
The argument is completely standard. The inequality $\NC\geq \NCM$ is obvious, and follows from the same relation between the relative entropy and its measured version. Since both $D(\cdot\|\cdot)$ and $\DM(\cdot\|\cdot)$ obey the data processing inequality, for every classical channel $\Lambda:\T{m}\to \T{m'}$ we obtain that
\bbb
\NC\left( \Lambda(\rho)\right) = \inf_{\sigma'\in \CC_{m'}} D\left( \Lambda(\rho) \big\| \sigma' \right) \leq \inf_{\sigma'\in \Lambda\left(\CC_{m'}\right)} D\left( \Lambda(\rho) \big\| \sigma' \right) = \inf_{\sigma\in \CC_m} D\left( \Lambda(\rho) \big\| \Lambda(\sigma) \right) \leq \inf_{\sigma\in \CC_m} D\left( \rho \| \sigma \right) = \NC(\rho)\, , 
\eee
and analogously for $\NCM$. This proves monotonicity.

Convexity descends from the fact that both $\NC$ and $\NCM$ are defined as the infimum of a jointly convex function on a convex domain. For example,
\begin{align*}
\NC\left(\sumno_i p_i \rho_i\right) &= \inf_{\sigma\in \CC_m} D\left( \sumno_i p_i \rho_i\, \big\|\, \sigma\right) \\
&= \inf_{\{\sigma_i\}_i\subseteq \CC_m} D\left( \sumno_i p_i \rho_i\, \big\|\, \sumno_i p_i \sigma_i \right) \\
&\leq \inf_{\{\sigma_i\}_i\in \CC_m} \sum_i p_i\, D\left(\rho_i \| \sigma_i \right) \\
&= \sum_i p_i\, \inf_{\sigma_i \in \CC_m}D\left(\rho_i \| \sigma_i \right) \\
&= \sum_i p_i\, \NC(\rho_i)\, .
\end{align*}
The proof for $\NCM$ is entirely analogous.

Faithfulness follows, e.g., from Pinsker's inequality $\DKL(p\|q)\geq \frac12 \log_2(e) \|p-q\|_1^2$~\cite{Pinsker}, which implies that
\bbb
\DM(\rho\|\sigma) = \sup_{\M} \DKL\left( P_\rho^\M \big\| P_\sigma^\M \right) \geq \frac12 \log_2(e) \sup_{\M} \left\|P_\rho^\M - P_\sigma^\M \right\|_1^2 = \frac12 \log_2(e)\, \|\rho-\sigma\|_1^2\, ,
\eee
where in the last line we used the elementary fact that the trace distance is achieved by the (binary) measurement $\{\Pi,\id-\Pi\}$, with $\Pi$ being the projector onto the positive subspace of $\rho-\sigma$.

To prove the
\tcb{subadditivity} of $\NC$, just notice that for all $(m+n)$-mode CV systems $AB$ it holds that
\begin{align*}
\NC(\rho_A\otimes \sigma_B) &= \inf_{\sigma_{AB} \in \CC_{m+n}} D\left(\rho_A\otimes \sigma_B \|\sigma_{AB}\right) \\
&\leq \inf_{\sigma_A\otimes \sigma_B\in \CC_{m+n}} D\left(\rho_A\otimes \sigma_B \|\sigma_A \otimes \sigma_B \right) \\
&= \inf_{\sigma_A\in \CC_m,\, \sigma_B\in \CC_n} \left\{ D\left(\rho_{A} \|\sigma_A \right) + D\left(\rho_{B} \| \sigma_B \right) \right\} \\
&= \NC(\rho_A) + \NC(\rho_B)\, ,
\end{align*}
where in the third line we used the identity~\cite[Eq.~(5.22)]{PETZ-ENTROPY}.
\end{proof}

\begin{cor} \label{NCi_monotone_cor}
The functions $N_r^{M,\downarrow,\infty}, N_r^{M,\uparrow,\infty}$ are nonclassicality monotones. The regularization $N_r^{\downarrow,\infty} = N_r^{\uparrow, \infty}\eqqcolon \NCi$ is unique and is a weakly additive nonclassicality monotone; it satisfies that $\NCi\leq \NC$.
\end{cor}

\begin{proof}
Follows directly from Lemmata~\ref{NC_subadd_lemma} and~\ref{regularization_superadd_lemma}.
\end{proof}

We now argue that the monotones $\NC, \NCM$ behave like useful resource quantifiers on states of physical interest. An essential basic feature is finiteness on bounded-energy states, where the energy is measured by the total photon number Hamiltonian.

\begin{prop}\label{bounded energy prop}
Let $\rho$ be an $m$-mode state with finite mean photon number $E\coloneqq \Tr \left[\rho \left( \sumno_{j=1}^m a_j^\dag a_j \right)\right] < \infty$. Then
\bb
\NCM(\rho) \leq \NC(\rho) \leq m\, g(E/m)\, ,
\ee
where $g(x)\coloneqq (x+1)\log_2 (x+1) - x\log_2 x$.
\end{prop}

\begin{proof}
It is well known that the entropy of an $m$-mode state with finite mean photon number $E$ is at most $m g(E/m)$, which indeed corresponds to the entropy of the thermal state with the same energy. Hence, $\rho$ has finite entropy, so that~\eqref{hierarchy} holds. Thus, we only have to show that $\NC(\rho)\leq m g(E/m)$. For an arbitrary $\nu\geq 0$, let
\bb
\tau_\nu \coloneqq \frac{1}{1+\nu} \sum_{n=0}^\infty \left( \frac{\nu}{1+\nu}\right)^n \ketbra{n} = \frac{1}{1+\nu} \left( \frac{\nu}{1+\nu}\right)^{a^\dag a}
\label{tau}
\ee
be the single-mode thermal state of mean photon number $\nu$. It is well known that $\tau_\nu\in \CC_1$, and hence $\tau_\nu^{\otimes m} \in \CC_m$, for all $\nu\in [0,\infty)$. 
Therefore,
\begin{align*}
\NC(\rho) &\leq \inf_{\nu\geq 0} D\left( \rho\, \big\|\, \tau_\nu^{\otimes m} \right) \\
&= \inf_{\nu\geq 0} \left\{ - S(\rho) + m \log_2 (1+\nu) - E \log_2 \left( \frac{\nu}{1+\nu} \right) \right\} \\
&= - S(\rho) + m\, g\left( E/m \right)\, ,
\end{align*}
where we used the variational representation
\bbb
g(x) = \inf_{\nu\geq 0} \left\{\log_2(1+\nu) - x \log_2 \left( \frac{\nu}{1+\nu}\right)\right\} ,
\eee
whose proof is elementary.
\end{proof}

Further results on our nonclassicality monotones will be given in Section~\ref{secondary res section}.

\section{Proof of Theorem~\ref{general_bound_rates_thm} and of Corollaries~\ref{bound_rates_Esq_cor} and~\ref{bound_rates_free_energy_cor}} \label{proofs 1 section}

\subsection{Proof of Theorem~\ref{general_bound_rates_thm}}

In this section we prove our first main result, Theorem~\ref{general_bound_rates_thm}. We start with a simple lemma, which justifies the name of maximal asymptotic transformation rate given to the quantity in Definition~\ref{max_rate_def} (cf.\ Definition~\ref{max_rate_def}).

\begin{lemma} \label{R_smaller_Rtilde_lemma}
Let $(\mathcal{S},\mathcal{F})$ be a QRT. For any two systems $A,B\in\mathcal{S}$ and any two states $\rho_A\in\D{A}$ and $\sigma_B\in\D{B}$, it holds that
\bb
R(\rho_{A} \to \sigma_{B}) \leq \widetilde{R}(\rho_{A} \to \sigma_{B})\, .
\ee
\end{lemma}

\begin{proof}
For all $n$ and all free operations $\Lambda_n \in\mathcal{F}\left( A^n\to B^{\floor{rn}}\right)$, the data processing inequality for the trace norm~\cite{Ruskai1994} implies that
\bbb
\max_{j=1,\ldots, \floor{rn}} \left\| \left(\Lambda_n \left(\rho_{A}^{\otimes n}\right) \right)_{j} - \sigma_{B} \right\|_1 \leq \left\|\Lambda_n\left( \rho_{A}^{\otimes n}\right) - \sigma_{B}^{\otimes \floor{r n}} \right\|_1\, .
\eee
Therefore, a sequence of protocols that achieves a rate $r$ in~\eqref{rate} (i.e., that makes the global error vanish) achieves the same rate in~\eqref{max_rate} (because the maximum local error will also vanish). The claim follows.
\end{proof}

We are now ready to present the proof of Theorem~\ref{general_bound_rates_thm}.

\begin{proof}[Proof of Theorem~\ref{general_bound_rates_thm}]
It suffices to show that $\widetilde{R}(\rho_A\!\to\!\sigma_B) \leq \frac{G(\rho_A)}{G(\sigma_B)}$. For any sequence of free operations $\Lambda_n\in\mathcal{F}\left( A^n \to B^{\floor{rn}}\right)$ satisfying
\tcb{, for all $j$, $\liminf_{n\to\infty}\left\| \left( \Lambda_n \left(\rho_{A}^{\otimes n}\right) \right)_{j} - \sigma_{B} \right\|_1=0$}
it holds that
\begin{align*}
G(\rho_A) &\texteq{1}  \liminf_{n\to\infty} \frac1n\, G\big( \rho_A^{\otimes n} \big) \\
&\textgeq{2} \liminf_{n\to\infty} \frac1n\, G\left(\Lambda_n\big(\rho_A^{\otimes n}\big)\right) \\
&\textgeq{3} \liminf_{n\to\infty} \frac1n\, \sum_{j=1}^{\floor{rn}} G\left(\left(\Lambda_n\big(\rho_A^{\otimes n}\big)\right)_j\right) \\
&\geq \liminf_{n\to\infty} \frac{\floor{rn}}{n} \min_j G\left(\left(\Lambda_n\big(\rho_A^{\otimes n}\big)\right)_j\right) \\
&\texteq{4} \liminf_{n\to\infty} \frac{\floor{rn}}{n} G\left(\left(\Lambda_n\big(\rho_A^{\otimes n}\big)\right)_{j_n}\right)\\
&= r \liminf_{n\to\infty} G\left(\left(\Lambda_n\big(\rho_A^{\otimes n}\big)\right)_{j_n}\right) \\
&\textgeq{5} r\, G(\sigma_B)\, .
\end{align*}
Here, 1~holds due to weak additivity, even without the lim inf and for every $n$; 2~comes from monotonicity; 3~from strong superadditivity; in~4 we constructed a sequence of indices $j_n$ achieving the minimum; finally, 5~descends from lower semicontinuity and the assumption on $\Lambda_n$. Then a supremum over $r$ yields the claim.
\end{proof}

Before moving on to the study of the applications, it is perhaps instructive to compare the above argument with \tcr{the one we saw in~\eqref{finite_dim_proof}, where the same bound on rates was proved (under different assumptions) in the finite-dimensional setting.} 
The main difference lies in step~3, in which we exploit strong superadditivity to move the error analysis to the single-copy level, where it is ultimately tackled by means of lower semicontinuity (step~5). \tcr{In~\eqref{finite_dim_proof}, instead, asymptotic continuity was leveraged to carry out an error analysis directly at the many-copy level.} 
This type of ideas had been previously exploited in~\cite[Theorem~4 and Remark~10]{GrandTour}.

\subsection{Proof of Corollary~\ref{bound_rates_Esq_cor}}
\label{section squashed entanglement}

We now apply Theorem~\ref{general_bound_rates_thm} to the resource theory of entanglement. Let us start by fixing some terminology. The \textbf{squashed entanglement} of a bipartite state $\rho_{AB}$ of a finite-dimensional bipartite system $AB$ is defined by~\cite{tucci1999, squashed, faithful, rel-ent-sq}
\bb
E_{sq}(\rho_{AB}) \coloneqq \frac12 \inf_{\rho_{ABE}} I(A:B|E)_\rho\, ,
\label{Esq}
\ee
where the infimum is over all extensions $\rho_{ABE}$ of the state $\rho_{AB}$, i.e., over all tripartite states $\rho_{ABE}$ satisfying that $\Tr_E \left[\rho_{ABE}\right]=\rho_{AB}$, and
\bb
I(A:B|E)_\rho \coloneqq S(\rho_{AE}) + S(\rho_{BE}) - S(\rho_E) - S(\rho_{ABE})
\label{conditional_mutual_information}
\ee
is the conditional mutual information. The problem with the above definition is that it cannot be extended directly to the infinite-dimensional case, because the right-hand side of~\eqref{conditional_mutual_information} may contain the undefined expression $\infty-\infty$~\cite{Shirokov-sq, Shirokov2016}.

Fortunately, Shirokov has found a way out of this impasse. The first step is to construct the conditional mutual information via an alternative expression to~\eqref{conditional_mutual_information}, namely,
\bb
I(A:B|E)_\rho = \sup_{\Pi_A} \left\{ I(A:BE)_{\Pi_A \rho \Pi_A} - I(A:E)_{\Pi_A \rho \Pi_A} \right\} ,
\label{conditional_mutual_information_alternative}
\ee
where the supremum is over all finite-dimensional projectors $\Pi_A$ on $A$. An equivalent expression is obtained by exchanging $A$ and $B$ in~\eqref{conditional_mutual_information_alternative}. Clearly,~\eqref{conditional_mutual_information_alternative} reduces to~\eqref{conditional_mutual_information} when $A$ is finite dimensional.

With~\eqref{conditional_mutual_information_alternative} at hand,~\eqref{Esq} can be extended without difficulty to the infinite-dimensional case~\cite[Eq.~(17)]{Shirokov-sq}. In order for this to work, we have to keep in mind that the system $E$ could and in general will be infinite-dimensional.

An alternative strategy to generalize the squashed entanglement to infinite-dimensional systems could be that of truncating the state directly by means of local finite-dimensional projectors. This results in a different function $\textit{\^E}_{sq}$, defined by~\cite[Eq.~(37)]{Shirokov-sq}
\bb
\textit{\^E}_{sq}(\rho_{AB}) \coloneqq \sup_{\Pi_A, \Pi_B} E_{sq}\left( (\Pi_A\otimes \Pi_B)\, \rho_{AB}\, (\Pi_A\otimes \Pi_B) \right) ,
\label{Esq_hat}
\ee
where the infimum runs over all finite-dimensional projectors $\Pi_A$ and $\Pi_B$. The nested optimizations hidden in~\eqref{Esq_hat} make $\textit{\^E}_{sq}$ a slightly less desirable quantity than $E_{sq}$. Nevertheless, we will find it useful in intermediate computations. 

The main properties of the two functions $E_{sq}$ and $\textit{\^E}_{sq}$ that we will use are as follows:
\begin{enumerate}[(a)]
\item both $E_{sq}$ and $\textit{\^E}_{sq}$ are strongly superadditive~\cite[Propositions~2B and~3B]{Shirokov-sq};
\item both $E_{sq}$ and $\textit{\^E}_{sq}$ are additive, and hence also weakly additive~\cite[Propositions~2B and~3B]{Shirokov-sq};
\item $\textit{\^E}_{sq}$ is lower semicontinuous everywhere~\cite[Proposition~3A]{Shirokov-sq};
\item $E_{sq}(\rho_{AB}) \equiv \textit{\^E}_{sq}(\rho_{AB})$ on all states with $\min\left\{ S(\rho_A),\, S(\rho_B),\, S(\rho_{AB}) \right\}<\infty$~\cite[Proposition~3C]{Shirokov-sq}. 
\end{enumerate}

\begin{proof}[Proof of Corollary~\ref{bound_rates_Esq_cor}]
It suffices to write that
\bbb
R(\rho_{AB}\!\to\!\sigma_{A'B'})\leq \widetilde{R} (\rho_{AB}\!\to\!\sigma_{A'B'}) \textleq{1} \frac{\textit{\^E}_{sq}(\rho_{AB})}{\textit{\^E}_{sq}(\sigma_{A'B'})} \texteq{2} \frac{E_{sq}(\rho_{AB})}{E_{sq}(\sigma_{A'B'})}\, ,
\eee
where 1~is just an application of Theorem~\ref{general_bound_rates_thm}, made possible by properties~(a),~(b), and~(c) above, while 2~follows from~\eqref{strange_sets} and property~(d).
\end{proof}

\subsection{Proof of Corollary~\ref{bound_rates_free_energy_cor}}

We now move on to the case of quantum thermodynamics at some (fixed) inverse temperature $\beta>0$. The monotone~\cite{Brandao-thermo}
\bb
G(\rho_A)\coloneqq \frac1\beta D(\rho_A \|\gamma_A)
\ee
is easily seen to be:
\begin{enumerate}[(a)]
\item strongly superadditive, because
\bbb
G(\rho_{AB}) = \frac1\beta D(\rho_{AB} \|\gamma_{AB}) =  \frac1\beta D(\rho_{AB} \|\gamma_A \otimes \gamma_B) \geq \frac1\beta D(\rho_A \|\gamma_A) + \frac1\beta D(\rho_B \| \gamma_B)\, ,
\eee
where the first identity is a consequence of the fact that $H_{AB}=H_A+H_B$, while the inequality follows from~\cite[Corollary~5.21]{PETZ-ENTROPY};
\item additive and hence weakly additive, since
\bbb
G(\rho_A \otimes \sigma_B) = \frac1\beta D\left(\rho_A \otimes \sigma_B \| \gamma_{AB} \right) = \frac1\beta D\left(\rho_A \otimes \sigma_B \| \gamma_A \otimes \gamma_B \right) = \frac1\beta D\left(\rho_A \| \gamma_A \right) + \frac1\beta D\left( \sigma_B \| \gamma_B \right) ;
\eee
finally,
\item lower semicontinuous, as follows, e.g., from~\cite[Proposition~5.23]{PETZ-ENTROPY}.
\end{enumerate}

\begin{proof}[Proof of Corollary~\ref{bound_rates_free_energy_cor}]
Thanks to properties~(a),~(b), and~(c) above, the claim follows directly from Theorem~\ref{general_bound_rates_thm}.
\end{proof}

\section{The long march towards Theorems~\ref{main result thm} and~\ref{bound_rates_NCMi_thm}} \label{proofs 2 section}

Throughout this section, we introduce all the necessary technical tools to arrive at a proof of Theorems~\ref{main result thm} and~\ref{bound_rates_NCMi_thm}. Along the way, we prove also Lemma~\ref{Berta variational lemma} (Section~\ref{variational expressions subsection}) and Corollary~\ref{superadditivity_cor} (Section~\ref{main results subsection})

\subsection{Proof of the variational expression for the measured relative entropy (Lemma~\ref{Berta variational lemma})} \label{variational expressions subsection}

The main goal of this subsection is to prove Lemma~\ref{Berta variational lemma}, which extends to the infinite-dimensional case the variational expressions for the measured relative entropy introduced in~\cite{Berta2017}.

Let us start by highlighting the main differences and similarities between the six variational expressions reported in Lemma~\ref{Berta variational lemma}\tcb{, reported here for the reader's convenience:
\begin{align}
	\DM (\rho\|\sigma) &= \sup_{h\in \B{}} \left\{ \Tr \left[\rho h\right] - \log_2 \Tr\left[ \sigma 2^h\right] \right\} \tag{\ref{Berta variational 1}} \\
	&= \sup_{h\in \B{}} \left\{ \Tr \left[\rho h\right] + \log_2 (e) \left(1 - \Tr \left[\sigma 2^h\right]\right) \right\} \tag{\ref{Berta variational 2}} \\
	&= \sup_{0<\delta \id<L\in \B{}} \left\{ \Tr \left[\rho \log_2 L\right] - \log_2 \Tr \left[\sigma L\right] \right\} \tag{\ref{Berta variational 3bis}} \\
	&= \sup_{0<\delta \id<L\in \B{}} \left\{ \Tr \left[\rho \log_2 L\right] + \log_2 (e) \left( 1 - \Tr \left[\sigma L\right]\right) \right\} \tag{\ref{Berta variational 4bis}} \\
	&= \sup_{0<L\in \B{}} \left\{ \Tr \left[\rho \log_2 L\right] - \log_2 \Tr \left[\sigma L\right] \right\} \tag{\ref{Berta variational 3}} \\
	&= \sup_{0<L\in \B{}} \left\{ \Tr \left[\rho \log_2 L\right] + \log_2 (e) \left(1 - \Tr \left[\sigma L\right]\right) \right\} . \tag{\ref{Berta variational 4}}
\end{align}
}
\vspace{-3ex}
\begin{itemize}
\item We see immediately that they can be grouped in pairs:~\eqref{Berta variational 1} and~\eqref{Berta variational 2};~\eqref{Berta variational 3bis} and~\eqref{Berta variational 4bis}; finally,~\eqref{Berta variational 3} and~\eqref{Berta variational 4}. The two expressions in each pair involve an optimization over exactly the same set, and differ only by the objective function, which contains a $-\log_2 x$ in~\eqref{Berta variational 1},~\eqref{Berta variational 3bis}, and~\eqref{Berta variational 3}, and its linearized version $\log_2(e) (1-x)$ in~\eqref{Berta variational 2},~\eqref{Berta variational 4bis}, and~\eqref{Berta variational 4}.
    \item The programs in~\eqref{Berta variational 3bis} and~\eqref{Berta variational 4bis} contain an optimization over all bounded operators $L$ that are also bounded away from $0$, i.e., such that $L\geq \delta \id$ for some $\delta>0$, where $\id$ is the identity on $\HH$.
    \item In the programs~\eqref{Berta variational 3} and~\eqref{Berta variational 4} we instead removed this latter constraint, and optimized only on positive operators $L>0$. Of course, this is a priori not the same: in infinite dimensions, it can happen --- e.g., for any strictly positive density operator --- that $L>0$ but there is no uniform bound $L\geq \delta \id>0$.
    \item Since in~\eqref{Berta variational 3} and~\eqref{Berta variational 4} the operator $\log_2 L$ is possibly unbounded from below, it may happen that $\Tr[\rho \log_2 L]=-\infty$. This is not a problem, because we always have that $\Tr[\sigma L]>0$ and hence $-\log_2 \Tr[\sigma L]<+\infty$; therefore, the first addend is the only one that may diverge, and no uncertainties of the form $-\infty+\infty$ can arise in the objective function. 
\end{itemize}

\begin{proof}[Proof of Lemma~\ref{Berta variational lemma}]
Following the above observations, we divide the proof in several smaller steps.
\begin{enumerate}

    \item Let us start by showing that~\eqref{Berta variational 1} is equivalent to~\eqref{Berta variational 2},~\eqref{Berta variational 3bis} to~\eqref{Berta variational 4bis}, and~\eqref{Berta variational 3} to~\eqref{Berta variational 4}. We only present the argument for the equivalence between~\eqref{Berta variational 1} and~\eqref{Berta variational 2}, as the others are entirely analogous. First, from the inequality $\log_2 x\leq \log_2(e) (x-1)$ we see that
    \bbb
    \Tr[\rho h]-\log_2\Tr \left[\sigma 2^h\right]\geq \Tr[\rho h] + \log_2(e) \left(1-\Tr\left[\sigma 2^h\right]\right)
    \eee
    for any $h$. At the same time, the expression~\eqref{Berta variational 1} is manifestly invariant under transformations of the type $h\mapsto h+\lambda I$ for any $\lambda\in\R$. So, we can always choose a $\lambda$ in both expressions such that $\Tr\left[\sigma 2^h\right]=1$, thus saturating the aforementioned inequality.

    \item Now, observe that~\eqref{Berta variational 1} is equivalent to~\eqref{Berta variational 3bis}, upon a change in         parametrization $h=\log_2 L$. In fact, $\log_2 L$ is bounded if and only if $L$ itself is bounded and moreover $L\geq \delta \id>0$. This implies that the variational expressions in~\eqref{Berta variational 1},~\eqref{Berta variational 2},~\eqref{Berta variational 3bis}, and~\eqref{Berta variational 4bis} all coincide.

    \item We now show that they also coincide with those in~\eqref{Berta variational 3} and~\eqref{Berta variational 4}. Clearly, since the optimization in~\eqref{Berta variational 3} is over a larger set than that in~\eqref{Berta variational 3bis}, its value cannot decrease. Therefore, to prove equality we only have to prove that
    \bbb
    \sup_{0<\delta \id<L\in \B{}} \left\{ \Tr\left[ \rho \log_2 L\right] - \log_2 \Tr\left[ \sigma L \right]\right\} \geq \sup_{0<L\in \B{}} \left\{ \Tr \left[\rho \log_2 L\right] - \log_2 \Tr [\sigma L] \right\} .
    \eee
    To this end, pick a bounded $L>0$, and let us show how to construct a family of bounded $L_\delta\geq \delta \id>0$ such that
    \begin{equation} \label{limit in delta for L_delta}
	\lim_{\delta\to 0^+}\left\{\Tr\left[\rho\log_2 L_\delta\right] - \log_2\Tr[\sigma L_\delta] \right\} = \Tr[\rho\log_2 L] - \log_2\Tr[\sigma L]\, .
	\end{equation}
    Since the expression $\Tr[\rho\log_2 L]-\log_2\Tr[\sigma L]$ is clearly scale-invariant in $L$, i.e., it takes the same value for $L$ and $\lambda L$, for all $\lambda>0$, we can assume without loss of generality that $L\leq \id/2$. For $0<\delta\leq 1/2$, set $L_\delta\coloneqq L + \delta \id\geq \delta \id$.
    
    Using the spectral theorem for bounded operators~\cite[Theorem~7.12]{HALL}, we can find a projection-valued measure $\mu$ on $[0,1/2]$ such that $L=\int_0^{1/2} \lambda d\mu(\lambda)$ and therefore $L_\delta = \int_0^{1/2} (\lambda+\delta) d\mu(\lambda)$. Defining the real-valued measure $\mu_\rho$ on $[0,1/2]$ such that $\mu_\rho(X) = \Tr[\rho \mu(X)]$ for all measurable sets $X\subseteq [0,1/2]$, we have that
    \bbb
    \Tr\left[\rho\left(-\log_2 L\right)\right] = \int_0^{1/2} (-\log_2 \lambda) d\mu_\rho(\lambda)\, ,\qquad \Tr\left[\rho\left(-\log_2 L_\delta\right)\right] = \int_0^{1/2} \left(-\log_2 (\lambda+\delta)\right) d\mu_\rho(\lambda)\, .
    \eee
    Since the functions $\lambda \mapsto -\log_2 (\lambda+\delta)$ are pointwise monotonically decreasing in $\delta$, converge pointwise to $\lambda \mapsto -\log_2 \lambda$, and all the functions involved are nonnegative, we can apply Beppo Levi's monotone convergence theorem~\cite{BeppoLevi} (see also~\cite[Theorem~11.28]{RUDIN-PRINCIPLES}) and conclude that
    \bbb
    \lim_{\delta\to 0^+} \Tr\left[\rho\left(-\log_2 L_\delta\right)\right] = \lim_{\delta\to 0^+} \int_0^{1/2} \left(-\log_2 (\lambda+\delta)\right) d\mu_\rho(\lambda) = \int_0^{1/2} \left(-\log_2 \lambda\right) d\mu_\rho(\lambda) = \Tr\left[\rho \left(-\log_2 L\right)\right] .
    \eee
    On the other hand, clearly $\Tr[\sigma L_\delta] = \Tr[\sigma L] + \delta$ converges to $\Tr[\sigma L]>0$ as $\delta\to 0^+$. This proves~\eqref{limit in delta for L_delta}, and thus allows us to conclude that the optimizations in~\eqref{Berta variational 1}--\eqref{Berta variational 4} all coincide.


    \item We now show that the variational program in~\eqref{Berta variational 4bis} actually yields the measured relative entropy $\DM(\rho\|\sigma)$. To begin, we prove that in~\eqref{Berta variational 4bis} we can restrict $L$ to be of the form $L=I+R$, with $\rk R<\infty$, without changing the value of the supremum. To this end, pick $L$ such that $1/m\leq L\leq m$ for some $m>0$, and consider an arbitrary $\epsilon>0$. Construct a finite-dimensional projector $P$ such that $\left\|\rho - P\rho P\right\|_1,\, \left\|\sigma - P \sigma P\right\|_1 \leq \epsilon$. Then,
	\begin{align*}
	&\Tr \left[\rho \log_2 L\right] + \log_2(e) \left(1 - \Tr \left[\sigma L\right]\right) \\
	&\qquad \textleq{1} \Tr \left[P\rho P \log_2 L\right] + \log_2(e) \left(1 - \Tr \left[P \sigma P L\right]\right) + \epsilon (\log_2 m + m \log_2(e)) \\
	&\qquad \textleq{2} \Tr \left[\rho \log_2 (PLP + \id - P)\right] + \log_2(e) \left(1 - \Tr\left[ P \sigma P L\right]\right) + \epsilon (\log_2 m + m \log_2(e)) \\
	&\qquad \textleq{3} \Tr\left[ \rho \log_2 (PLP + \id - P)\right] + \log_2(e) \left(1 - \Tr\left[ \sigma (PLP + \id - P) \right]\right) + \epsilon (\log_2 m + (m+1) \log_2(e))\, . 
	\end{align*}
    Here, 1~follows because $\|\log_2 L\|_\infty\leq \log_2 m$ and $\|L\|_\infty\leq m$ (where $\|\cdot\|_\infty$ is the operator norm), in~2 we applied the operator Jensen inequality~\cite{Hansen2003} to the operator-concave function $\log_2$, and 3~is an application of the estimate $\Tr[\sigma (\id - P)] = \Tr[\sigma - P\sigma P]\leq \left\| \sigma - P \sigma P\right\|_1\leq \epsilon$. We see that up to introducing an arbitrarily small error we can substitute $L\mapsto PLP + \id - P = \id +R$, where $\rk R \leq \rk P<\infty$.
    
	Now, let $R$ be of finite rank, and denote with $R=\sum_{n=1}^N \lambda_n P_n$ its spectral decomposition. Then $L= \id + R = \sum_{n=0}^N (1+\lambda_n) P_n$, where $P_0 \coloneqq \id - \sum_{n=1}^N P_n$ and $\lambda_0=0$, and consequently
	\begin{align*}
	&\Tr[\rho \log_2 L] + \log_2(e) \left(1 - \Tr[\sigma L]\right) \\
	&\qquad\qquad = \log_2(e) (1-\Tr[\sigma]) + \sum_{n=0}^N \left(\log_2(1+\lambda_n)\Tr[\rho P_n] - \log_2(e) \lambda_n\Tr[\sigma P_n]\right) \\
	&\qquad\qquad \textleq{4} \log_2(e) (1-\Tr[\sigma]) + \sum_{n=1}^N \left( \Tr[\rho P_n] \log_2 \frac{\Tr[\rho P_n]}{\Tr[\sigma P_n]} - \log_2(e)\left(\Tr[\rho P_n] - \Tr[\sigma P_n]\right) \right) \\
	&\qquad\qquad \textleq{5} \log_2(e) (1-\Tr[\sigma]) + \sum_{n=0}^N \left( \Tr[\rho P_n] \log_2 \frac{\Tr[\rho P_n]}{\Tr[\sigma P_n]} - \log_2(e)\left(\Tr[\rho P_n] - \Tr[\sigma P_n]\right) \right) \\
	&\qquad\qquad = \sum_{n=0}^N \Tr[\rho P_n] \log_2 \frac{\Tr[\rho P_n]}{\Tr[\sigma P_n]} \\
	&\qquad\qquad \texteq{6} \DKL\left( P^\M_\rho \big\| P^\M_\sigma\right) \\
	&\qquad\qquad \leq \DM (\rho\|\sigma) .
	\end{align*}
	Here, the inequality in~4 comes from the estimate $a \log_2 (1+x) - \log_2(e) b x \leq a \log_2 \frac{a}{b} - \log_2(e) (a - b)$, (which can be proven simply by maximisation in $x$), while 5~is a consequence of the fact that $a \log_2 \frac{a}{b} - \log_2(e) (a - b)\geq 0$ for all $a,b\geq 0$.
	In~6, we introduced the measurement $\M\coloneqq \{P_x\}_{x\in \{0,\ldots, N\}}$.

	The converse is proved with exactly the same argument put forth by Berta et al.\ in the proof of~\cite[Lemma~1]{Berta2017}. Namely, let $\M=\{E_x\}_{x\in \X}$ be a quantum measurement. If there exists $x\in \X$ such that $\Tr[\sigma E_x]=0< \Tr[\rho E_x]$, then on the one hand clearly $\DM(\rho\|\sigma) \geq \DKL\left( P^\M_\rho \big\| P^\M_\sigma\right) = +\infty$. On the other, we see that the kernels of $\rho$ and $\sigma$ obey $\ker(\sigma)\nsubseteq \ker(\rho)$, i.e., there exists a pure state $\ket{\psi}\in \ker(\sigma)\setminus \ker(\rho)$. Setting $L=\lambda\psi + \id-\psi$ and letting $\lambda\to +\infty$ proves that the variational program in~\eqref{Berta variational 4bis} is unbounded from above, as it should be.

	We now consider the case where $\Tr[\sigma E_x]=0$ only when also $\Tr[\rho E_x]=0$. Introduce the set
	\bbb
	\widetilde{\X}\coloneqq \{x\in \X: \Tr[\rho E_x]\Tr [\sigma E_x]>0 \}\,,
	\eee
	and write:
	\begin{align*}
	\DKL\left( P^\M_\rho \big\| P^\M_\sigma\right) =& \sum_{x\in \widetilde{\X}} \Tr[\rho E_x] \left( \log_2 \Tr[\rho E_x] - \log_2 \Tr[\sigma E_x] \right) \\
	=& \Tr\left[ \rho \sumno_{x\in \widetilde{\X}} \sqrt{E_x}\, \log_2 \left( \frac{\Tr[\rho E_x]}{\Tr[\sigma E_x]}\cdot \id\right) \sqrt{E_x} \right] \\
	\textleq{7}& \Tr\left[ \rho \log_2 \left( \sumno_{x\in \widetilde{\X}} \frac{\Tr[\rho E_x]}{\Tr[\sigma E_x]}\, E_x\right) \right] \\
	\texteq{8}& \Tr\left[ \rho \log_2 L \right] + \log_2(e)\left(1 - \Tr[\sigma L]\right) ,
	\end{align*}
	where 7~is again an application of the operator Jensen inequality~\cite{Hansen2003} to the operator-concave function $\log_2$, and in~8 we defined $L\coloneqq \sum_x \frac{\Tr[\rho E_x]}{\Tr[\sigma E_x]}\, E_x$, so that $\Tr[\sigma L]=1$.
\end{enumerate}
\end{proof}

\begin{rem} \label{vanishing 2nd state measured rem}
The programs~\eqref{Berta variational 2},~\eqref{Berta variational 4bis}, and~\eqref{Berta variational 4} are all well defined also for $\sigma=0$. They yield $\DM(\rho\|0)=+\infty$, as it should be.
\end{rem}

\subsection{The monotone $\boldsymbol{\Gamma}$}\label{the long march subsection}

In order to arrive at a proof of Theorem~\ref{main result thm}, we first formalize the definition of the quantity that appears on the right-hand side of~\eqref{second expression for ncm}.

\begin{Def} \label{Gamma_def}
For an arbitrary $m$-mode state $\rho$, let us construct the quantity
\begin{align}
\Gamma(\rho) \coloneqq&\ \sup_{h\in \B{m}} \left\{ \Tr[\rho h] - \log_2 \sup_{\alpha\in \C^m} \braket{\alpha|2^h|\alpha}\right\} \label{Gamma_1} \\
=& \sup_{h\in \B{m}} \left\{ \Tr[\rho h] + \log_2(e) \left( 1- \sup_{\alpha\in \C^m} \braket{\alpha|2^h|\alpha}\right) \right\} \label{Gamma_2}
\end{align}
\end{Def}

Note that since $2^h>0$, there must exist some $\alpha\in \C^m$ such that $\braket{\alpha|2^h|\alpha}>0$.
Moreover, the two programs in~\eqref{Gamma_1} and~\eqref{Gamma_2} are equivalent, as can be verified by following the same strategy as in step~1 of the proof of Lemma~\ref{Berta variational lemma}. This ensures that $\Gamma$ is indeed well defined. Let us now establish some of its basic properties.

\begin{lemma} \label{variational_Gamma_lemma}
For an $m$-mode state $\rho$, we have that
\begin{align}
	\Gamma(\rho) &= \sup_{0<\delta \id<L\in \B{m}} \left\{ \Tr \left[\rho \log_2 L\right] - \log_2 \sup_{\alpha\in \C^m} \braket{\alpha|L|\alpha} \right\} \label{Gamma variational 3bis} \\
	&= \sup_{0<\delta \id<L\in \B{m}} \left\{ \Tr \left[\rho \log_2 L\right] + \log_2 (e) \left( 1 - \sup_{\alpha\in \C^m} \braket{\alpha|L|\alpha} \right) \right\} \label{Gamma variational 4bis} \\
	&= \sup_{0<L\in \B{m}} \left\{ \Tr \left[\rho \log_2 L\right] - \log_2 \sup_{\alpha\in \C^m} \braket{\alpha|L|\alpha} \right\} \label{Gamma variational 3} \\
	&= \sup_{0<L\in \B{m}} \left\{ \Tr \left[\rho \log_2 L\right] + \log_2 (e) \left(1 - \sup_{\alpha\in \C^m} \braket{\alpha|L|\alpha} \right) \right\} . \label{Gamma variational 4} 
\end{align}
\end{lemma}

\begin{proof}
The argument proceeds exactly as in steps~1--3 of the proof of Lemma~\ref{Berta variational lemma}.
\end{proof}

We deduce the following elementary but important properties of the function $\Gamma$.

\begin{prop} \label{Gamma_properties_prop}
The function $\Gamma$ in Definition~\ref{Gamma_def} is a convex, lower semicontinuous, strongly superadditive nonclassicality monotone. It holds that $\Gamma (\rho) \leq \NCM(\rho)$ for all states $\rho$.
\end{prop}

\begin{proof}
First of all, $\Gamma$ is convex and lower semicontinuous because it is the pointwise supremum of convex-linear and lower semicontinuous functions $\rho \mapsto \Tr[\rho h] - \log_2 \sup_{\alpha\in \C^m} \braket{\alpha|2^h|\alpha}$ (cf.\ Definition~\ref{Gamma_def}). To see that it is a nonclassicality monotone, consider $\rho\in \D{m}$ and a classical channel $\Lambda:\T{m}\to \T{m'}$, and write
\begin{align}
    \Gamma\left( \Lambda(\rho)\right) &= \sup_{0<L'\in \B{m'}} \left\{ \Tr \left[\Lambda(\rho) \log_2 L'\right] - \log_2 \sup_{\alpha\in \C^{m'}} \braket{\alpha|L'|\alpha} \right\} \nonumber \\
    &\texteq{1} \sup_{0<L'\in \B{m'}} \left\{ \Tr \left[\rho\, \Lambda^\dag \left(\log_2 L'\right) \right] - \log_2 \sup_{\sigma'\in \CC_{m'}} \Tr[\sigma' L'] \right\} \nonumber \\
    &\textleq{2} \sup_{0<L'\in \B{m'}} \left\{ \Tr \left[\rho \log_2 \Lambda^\dag \left(L'\right) \right] - \log_2 \sup_{\sigma'\in \CC_{m'}} \Tr[\sigma' L'] \right\} \nonumber \\
    &\textleq{3} \sup_{0<L'\in \B{m'}} \left\{ \Tr \left[\rho \log_2 \Lambda^\dag \left(L'\right) \right] - \log_2 \sup_{\sigma \in \CC_{m}} \Tr[\Lambda(\sigma) L'] \right\} \label{Gamma_properties_eq1} \\
    &= \sup_{0<L'\in \B{m'}} \left\{ \Tr \left[\rho \log_2 \Lambda^\dag \left(L'\right) \right] - \log_2 \sup_{\sigma \in \CC_{m}} \Tr[\sigma \Lambda^\dag(L')] \right\} \nonumber \\
    &\textleq{4} \sup_{0<L\in \B{m}} \left\{ \Tr \left[\rho \log_2 L \right] - \log_2 \sup_{\sigma \in \CC_{m}} \Tr[\sigma L] \right\} \nonumber \\
    &= \Gamma(\rho)\, . \nonumber
\end{align}
The justification of the above derivation is as follows. 1: We used the definition of adjoint map, and observed that since $L'$ is bounded and $\CC_{m'}=\overline{\co}\left\{\ketbra{\alpha}:\, \alpha\in \C^{m'}\right\}$, it holds that $\sup_{\sigma'\in \CC_{m'}} \Tr[\sigma' L'] = \sup_{\alpha\in \C^{m'}} \braket{\alpha|L'|\alpha}$. 2: We applied the operator Jensen inequality~\cite{Hansen2003} to the operator-concave function $\log_2$. 3: We restricted the inner supremum over $\sigma'$ to classical states of the form $\sigma'=\Lambda(\sigma)$, with $\sigma\in \CC_m$. 4: We observed that if $0<L'\in \B{m'}$ then also $0<\Lambda^\dag(L')\in \B{m}$, which can be seen by noticing that $\Tr\left[\omega\, \Lambda^\dag(L')\right]=\Tr\left[\Lambda(\omega) L'\right]>0$ for all states $\omega\in \D{m}$.

We now prove that $\Gamma$ is strongly superadditive. To this end, we take an arbitrary $(m+n)$-mode state $\rho_{AB}$ and write
\begin{align*}
    \Gamma\left(\rho_{AB}\right) &= \sup_{0<L_{AB}\in \B{m+n}} \left\{ \Tr \left[\rho_{AB} \log_2 L_{AB}\right] - \log_2 \sup_{\alpha_A \in \C^m\!\!,\; \alpha_B \in \C^n 
    } \left(\bra{\alpha_A}\otimes \bra{\alpha_B}\right)\tcb{L_{AB}}\left( \ket{\alpha_A}\otimes \ket{\alpha_B} \right) \right\} \\
    &\textgeq{5} \sup_{\substack{0<L_A\in \B{m},\\0<L_B\in \B{n}}} \left\{ \Tr \left[\rho_{AB} \log_2 (L_A \otimes L_B)\right] - \log_2 \sup_{\alpha_A \in \C^m\!\!,\; \alpha_B \in \C^n} \left(\bra{\alpha_A}\otimes \bra{\alpha_B}\right)\tcb{(L_A\otimes L_B)}\left( \ket{\alpha_A}\otimes \ket{\alpha_B} \right) \right\} \\
    &= \sup_{\substack{0<L_A\in \B{m},\\0<L_B\in \B{n}}} \left\{ \Tr \left[\rho_A \log_2 L_A\right] + \Tr\left[\rho_B \log_2 L_B\right] - \log_2 \left[ \left( \sup_{\alpha_A\in \C^{m}} \braket{\alpha_A|L_A|\alpha_A} \right) \left(\sup_{\alpha_B\in \C^{n}} \braket{\alpha_B|L_B|\alpha_B}\right)\right] \right\} \\
    &= \sup_{0<L_{A}\in \B{m}} \left\{ \Tr \left[\rho_{A} \log_2 L_{A}\right] - \log_2 \sup_{\alpha_{\!A}\in \C^{m}} \braket{\alpha_{A}|L_{A}|\alpha_{A}} \right\} \\
    &\qquad \qquad + \sup_{0<L_{B}\in \B{n}} \left\{ \Tr \left[\rho_{B} \log_2 L_{B}\right] - \log_2 \sup_{\alpha_{B}\in \C^{n}} \braket{\alpha_{B}|L_{B}|\alpha_{B}} \right\} \\
    &= \Gamma(\rho_A)+\Gamma(\rho_B)\, ,
\end{align*}
where in~5 we restricted the supremum to product operators $L_{AB}=L_A \otimes L_B$. It remains to establish the inequality $\Gamma\leq \NCM$. This is done as follows:
\begin{align*}
\NCM(\rho) &= \inf_{\sigma\in \CC_m} \DM(\rho\|\sigma) \\
&\texteq{6} \inf_{\sigma\in \CC_m} \sup_{0<L\in \B{m}} \left\{ \Tr \left[\rho \log_2 L\right] - \log_2 \Tr \left[\sigma L\right] \right\} \\
&\textgeq{7} \sup_{0<L\in \B{m}} \inf_{\sigma\in \CC_m} \left\{ \Tr [\rho \log_2 L] - \log_2 \Tr [\sigma L] \right\} \\
&= \sup_{0<L\in \B{m}} \left\{ \Tr [\rho \log_2 L] - \log_2 \sup_{\sigma\in \CC_m} \Tr [\sigma L] \right\} \\
&\texteq{8} \sup_{0<L\in \B{m}} \left\{ \Tr [\rho \log_2 L] - \log_2 \sup_{\alpha\in \C^m} \braket{\alpha|L|\alpha} \right\}\,.
\end{align*}
Here, in~6 we employed the variational representation~\eqref{Berta variational 3} for the measured relative entropy, in~7 we remembered that
\bbb
\inf_{x\in X} \sup_{y\in Y} f(x,y) \geq \sup_{y\in Y} \inf_{x\in X} f(x,y)
\eee
holds for an arbitrary function $f:X\times Y\to \R$ on any product set $X\times Y$, and finally in~8 we noted that since $\CC_m=\overline{\co}\left\{\ketbra{\alpha}:\, \alpha\in \C^m\right\}$ and the function $\sigma\mapsto \Tr[\sigma L]$ is linear and trace-norm continuous (because $L$ is bounded), it achieves the maximum on the extreme points of $\CC_m$, i.e., on coherent states.
\end{proof}

\begin{cor} \label{Gamma_regularization_cor}
The regularization $\Gamma^\infty(\rho)\coloneqq \lim_{n\to\infty} \frac1n\, \Gamma(\rho^{\otimes n})$ exists and is unique for all states $\rho$. It is a lower semicontinuous, weakly additive, and strongly superadditive nonclassicality monotone, and it satisfies that $\Gamma^\infty\geq \Gamma$.
\end{cor}

\begin{proof}
Follows by combining Lemma~\ref{regularization_superadd_lemma} and Proposition~\ref{Gamma_properties_prop}.
\end{proof}

\subsection{Two more technical lemmata}

In order to prove Theorem~\ref{main result thm}, and from there deduce Theorem~\ref{bound_rates_NCMi_thm}, we need two more technical lemmata. The first one tells us that provided a state $\rho$ has finite entropy, which will most definitely be the case in all situations of physical interest, we can take the operator $L$ in the variational program for $\NCM$ to be not only bounded but also trace class.

\begin{lemma}\label{restrictions on L}
On an $m$-mode system, let
\bb
\widetilde{\CC}_m \coloneqq \co\left( \CC_m \cup\{0\}\right)
\label{Cm tilde}
\ee
denote the set of subnormalized classical states. Then, the measured relative entropy of nonclassicality admits the variational expressions
\begin{align}
\NCM(\rho) &= \inf_{\sigma\in\CC_m}\sup_{\substack{0<L\in\B{m}}}\left\{\Tr[\rho\log_2  L] + \log_2(e) \left( 1 - \Tr[\sigma L] \right) \right\} \label{rough_variational_NCM} \\
&=\inf_{\sigma\in\widetilde{\CC}_m}\sup_{\substack{0<L\in\B{m}}}\left\{\Tr[\rho\log_2 L] + \log_2(e) \left( 1 - \Tr[\sigma L] \right) \right\} \label{rough_variational_NCM_alt}
\end{align}
for all $m$-mode states $\rho$. Moreover, in both~\eqref{rough_variational_NCM} and~\eqref{rough_variational_NCM_alt}:
\begin{enumerate}[(i)]
\item if $S(\rho)<\infty$, we can assume that $L\in\T{m}$ is \tcb{of trace class, and that $-\Tr \rho \log_2 L <\infty$};
\item if $\rk\rho<\infty$, we can assume that \tcb{$\supp L = \supp \rho$ and hence $\rk L<\infty$, with the convention that $-\Tr \rho \log_2 L$ is computed on the common support of $\rho$ and $L$}.
\end{enumerate}
\end{lemma}

\begin{proof}
As we have already seen, the expression~\eqref{rough_variational_NCM} is obtained by plugging~\eqref{Berta variational 4} into the definition~\eqref{NC_and_NCM} of measured relative entropy of nonclassicality. To see that also~\eqref{rough_variational_NCM_alt} holds, just notice that
\begin{align*}
\NCM(\rho) &= \inf_{\sigma\in \CC_m} \DM(\rho\|\sigma) \\
&= \inf_{\sigma\in \CC_m,\, \lambda \in [0,1]} \left\{\DM(\rho\|\sigma) - \log_2 \lambda \right\} \\
&= \inf_{\sigma\in \CC_m,\, \lambda \in [0,1]} \DM(\rho\|\lambda\sigma) \\
&= \inf_{\sigma\in \widetilde{\CC}_m} \DM(\rho\| \sigma) \\
&= \inf_{\sigma\in \widetilde{\CC}_m} \sup_{0<L\in \B{m}} \left\{ \Tr [\rho \log_2 L] + \log_2(e) \left( 1 - \Tr[\sigma L] \right) \right\} ,
\end{align*}
where the last step is once again~\eqref{Berta variational 4}. We now prove claims~(i) and~(ii) for~\eqref{rough_variational_NCM}.

We start by observing that restricting the set of operators $L$ over which we optimize can only decrease the final value of the program. Thus, it suffices to establish the opposite inequality. We start from claim~(i). Let $\rho$ be a finite-entropy $m$-mode state with spectral decomposition $\rho=\sum_{k=0}^\infty p_k \ketbra{e_k}$. We can assume without loss of generality that $\overline{\Span}\{e_k\}_{k\in \N}=\HH_m$, i.e., that $\{e_k\}_{k\in \N}$ forms a basis of the entire Hilbert space. Pick a bounded but not necessarily trace class operator $L$ that can enter the expression~\eqref{rough_variational_NCM}. Without loss of generality, we can assume that
\bb
-\infty < \Tr[\rho \log_2 L] = \sum_{k=0}^\infty p_k \braket{e_k|\log_2 L|e_k} < +\infty\, .
\label{restrictions on L eq1}
\ee
\tcb{In fact, if this is not the case the objective function evaluates to $-\infty$.}

For a certain $n\in \N$, construct the completely positive unital map $\Pi_n:\B{m}\to \B{m}$ given by $\Pi(X)\coloneqq P_nXP_n+Q_nXQ_n$, where $P_n\coloneqq \sum_{k=0}^{n-1}\ketbra{e_k}$ is the projector onto the the linear span $\Span\{\ket{e_k}\}_{k=0,\ldots, n-1}$ of the first $n$ eigenvectors of $\rho$, and $Q_n\coloneqq \id-P_n = \sum_{k=n}^\infty \ketbra{e_k}$. Set $\rho = \rho_n + \delta_n$, with $\rho_n\coloneqq P_n \rho P_n$ and $\delta_n \coloneqq Q_n \rho Q_n$, and define the trace class operator $L_n\coloneqq P_n L P_n + \delta_n$. Then, we have that
\begin{align}
    \Tr [\rho \log_2 L] &= \Tr [\rho_n \log_2 L] + \Tr [\delta_n \log_2 L] \nonumber \\
    &\texteq{1} \Tr [\Pi(\rho_n) \log_2 L] + \Tr [\delta_n \log_2 L] \nonumber \\
    &\texteq{2} \Tr [\rho_n\, \Pi \left( \log_2 L\right)] + \Tr [\delta_n \log_2 L] \nonumber \\
    &\textleq{3} \Tr [\rho_n \log_2 \Pi(L)] + \Tr [\delta_n \log_2 L] \label{restrictions on L eq2} \\
    &\texteq{4} \Tr [\rho_n \log_2 (P_n L P_n + \delta_n)] + \Tr [\delta_n \log_2 L] \nonumber \\
    &= \Tr [\rho \log_2 L_n] - \Tr [\delta_n \log_2 \delta_n] + \Tr [\delta_n \log_2 L]\, . \nonumber
\end{align}
Here, in~1 we observed that $\rho_n=\Pi(\rho_n)$, in~2 we used the easily verified fact that $\Pi=\Pi^\dag$, in~3 we applied the operator Jensen inequality~\cite{Hansen2003}, and finally in~4 we changed the \tcb{component of the argument of the first logarithm} on the subspace $\supp Q_n$, which is irrelevant because the trace is against $\rho_n$\tcb{, whose support is orthogonal to that of $Q_n$}.
Now, since $S(\rho) = -\Tr [\rho\log_2 \rho] = \sum_{k=0}^\infty p_k \log_2\frac{1}{p_k} < \infty$, we see that
\bb
\lim_{n\to\infty} \left( - \Tr [\delta_n \log_2 \delta_n] \right) = \lim_{n\to\infty} \sum_{k=n}^\infty p_k \log_2\frac{1}{p_k} = 0\, .
\label{restrictions on L eq3}
\ee
Moreover,~\eqref{restrictions on L eq1} implies that
\bb
\lim_{n\to\infty} \Tr [\delta_n \log_2 L] = \lim_{n\to\infty} \sum_{k=n}^\infty p_k \braket{e_k|\log_2 L|e_k} = 0\, .
\label{restrictions on L eq4}
\ee
Putting~\eqref{restrictions on L eq2}--\eqref{restrictions on L eq4} together, we see that
\bb
\liminf_{n\to\infty} \Tr [\rho \log_2 L_n] \geq \Tr [\rho \log_2 L]\, .
\label{restrictions on L eq5}
\ee

On the other hand, since $\overline{\Span}\{\ket{e_k}\}_{k\in \N} = \HH_m$, we have that $\lim_{n\to\infty} \Tr [\sigma P_n] = 1$ and therefore, by the gentle measurement lemma~\cite{Davies1969, VV1999} (see also~\cite[Lemma~9.4.2]{Wilde2}), 
\bb
\lim_{n\to\infty} \left\|\sigma - P_n \sigma P_n\right\|_1 = 0\, .
\label{restrictions on L eq6}
\ee
This immediately implies that
\begin{align}
\liminf_{n\to\infty} \left( 1 - \Tr [\sigma L_n] \right) &= \liminf_{n\to\infty} \left( 1 - \Tr [P_n\sigma P_n L] - \Tr [\sigma \delta_n] \right) \nonumber \\
&\textgeq{5} \liminf_{n\to\infty} \left( 1 - \Tr [\sigma L] - \left\|\sigma - P_n\sigma P_n\right\|_1 \|L\|_\infty - \Tr [\delta_n] \right) \label{restrictions on L eq7}
\\
&\texteq{6} 1 - \Tr [\sigma L]\, . \nonumber
\end{align}
Here, 5~comes from the fact that $L$ is bounded and also that $\sigma\leq \id$, while 6~descends from~\eqref{restrictions on L eq6} and from the elementary observation that since $\Tr[\rho] = \sum_{k=0}^\infty p_k = 1$ it follows that $\lim_{n\to\infty} \Tr [\delta_n] = \lim_{n\to\infty} \sum_{k=n}^\infty p_k = 0$.

Finally, combining~\eqref{restrictions on L eq5} and~\eqref{restrictions on L eq7} we deduce that
\begin{align*}
\liminf_{n\to\infty} \left( \Tr [\rho \log_2 L_n] + \log_2(e) \left( 1 - \Tr [\sigma L_n] \right) \right) &\geq \liminf_{n\to\infty} \Tr [\rho \log_2 L_n] + \log_2(e) \liminf_{n\to\infty} \left(1 - \Tr [\sigma L_n] \right) \\
&\geq \Tr [\rho \log_2 L] + \log_2(e) \left(1- \Tr [\sigma L]\right) .
\end{align*}
Remembering that $L_n$ is a trace class operator, this in turn implies that
\bbb
\sup_{\substack{0<L\in\B{m}}}\left\{\Tr[\rho\log_2 L] + \log_2(e) \left( 1- \Tr[\sigma L]\right) \right\} \leq \sup_{\substack{0<L\in\T{m}}}\left\{\Tr[\rho\log_2  L] + \log_2(e) \left( 1- \Tr[\sigma L]\right) \right\} ,
\eee
thus showing that in fact equality holds. The proof of claim~(i) is now complete.

As for claim~(ii), it suffices to repeat the above reasoning and observe that if $\rk \rho < \infty$ then $\delta_n=0$ for sufficiently large $n$, thus entailing that $\rk L_n < \infty$.
\end{proof}

Our second preliminary lemma presents a technical result whose topological content will be indispensable for a careful application of Sion's minimax theorem to the variational program~\eqref{rough_variational_NCM}.

\begin{lemma} \label{Cm tilde compact lemma}
The cone
\bb
\CC_m^+ \coloneqq \left\{\lambda\sigma: \lambda\geq 0,\, \sigma \in \CC_m\right\} \subset \Tp{m}
\ee
generated by the set of classical states is closed with respect to the weak* topology on $\T{m}$. Therefore, the set $\widetilde{\CC}_m = \co \left( \CC_m \cup \{0\} \right)$ of subnormalized classical states, defined in~\eqref{Cm tilde}, is weak*-compact.
\end{lemma}

\begin{proof}
Remember by Remark~\ref{w* rem} that we can think of $\T{m}$ as the dual space to $\K{m}$, the set of compact operators on $\HH_m$. We now show that $\CC_m^+$ is in fact the dual of a set $\pazocal{S}\subseteq \K{m}$ of compact operators, i.e.,
\bbb
\CC_m^+=\pazocal{S}^*\coloneqq \left\{T\in \T{m}:\, \Tr[TK]\geq 0\ \forall\ K\in \pazocal{S}\right\} .
\eee
Dual sets turn out to be automatically weak*-closed. This can be seen, e.g., in the case of $\pazocal{S}^*$, by noting that it can be written as the intersection 
\bbb
\pazocal{S}^* = \bigcap_{K\in \pazocal{S}} \left\{T\in \T{m}:\, \Tr[TK]\geq 0\right\} = \bigcap_{K\in \pazocal{S}} \varphi_K^{-1}\left( [0,\infty) \right) ,
\eee
where $\varphi_K:\T{m}\to \R$ is defined by $\varphi_K(T)\coloneqq \Tr[TK]$. Since the maps $\varphi_K$ are weak*-continuous by definition, each set $\varphi_K^{-1}\left( [0,\infty) \right)$ is weak*-closed, and therefore so is their intersection $\pazocal{S}^*$. 

From now on, for the sake of readability we write everything for single-mode systems only. Set
\bbb
\pazocal{S} \coloneqq \left\{ \sum_{\mu,\nu=1}^n\! \psi_\mu^* \psi_\nu\, e^{\frac12 |\alpha_\mu-\alpha_\nu|^2}\, \lambda^{a^\dag\! a} \disp(\alpha_\mu - \alpha_\nu) \lambda^{a^\dag\! a}\! :\, n\in \N_+,\, \psi\in \C^n,\, \alpha\in \C^n,\, \lambda \in [0,1) \right\} ,
\eee
where $\disp$ is the displacement operator~\eqref{disp}. Note that every operator in $\pazocal{S}$ is a finite linear combination of operators of the form $\lambda^{a^\dag\! a} \disp(\alpha_\mu - \alpha_\nu) \lambda^{a^\dag\! a}$, which are clearly compact (in fact, even trace class) as long as $\lambda \in [0,1)$. It is also elementary to see that $\ketbra{\beta}\in \pazocal{S}^*$ for every $\beta\in \C$, because
\begin{align*}
\braket{\beta| \left( \sum_{\mu,\nu=1}^n\! \psi_\mu^* \psi_\nu\, e^{\frac12 |\alpha_\mu-\alpha_\nu|^2}\, \lambda^{a^\dag\! a} \disp(\alpha_\mu - \alpha_\nu) \lambda^{a^\dag\! a}\right) |\beta} &= \sum_{\mu,\nu=1}^n\! \psi_\mu^* \psi_\nu\, e^{\frac12 |\alpha_\mu-\alpha_\nu|^2} \braket{\beta| \lambda^{a^\dag\! a} \disp(\alpha_\mu - \alpha_\nu) \lambda^{a^\dag\! a} |\beta} \\
&\texteq{1} \sum_{\mu,\nu=1}^n\! \psi_\mu^* \psi_\nu\, e^{\frac12 |\alpha_\mu-\alpha_\nu|^2} e^{-(1-\lambda^2)|\beta|^2} \braket{\lambda \beta| \disp(\alpha_\mu - \alpha_\nu) |\lambda \beta} \\
&\texteq{2} e^{-(1-\lambda^2)|\beta|^2} \sum_{\mu,\nu=1}^n\! \psi_\mu^* \psi_\nu\, e^{\lambda \left( (\alpha_\mu-\alpha_\nu)\beta^* - (\alpha_\mu-\alpha_\nu)^* \beta\right)} \\
&= e^{-(1-\lambda^2)|\beta|^2} \sum_{\mu,\nu=1}^n\! \psi_\mu^*\, e^{\lambda (\alpha_\mu \beta^* - \alpha_\mu^* \beta)}\, \psi_\nu\, e^{\lambda (\alpha_\nu^* \beta - \alpha_\nu \beta^*)} \\
&= e^{-(1-\lambda^2)|\beta|^2} \left|\sumno_{\mu=1}^n \psi_\mu^*\, e^{\lambda (\alpha_\mu \beta^* - \alpha_\mu^* \beta)}\right|^2 \\
&\geq 0\, ,
\end{align*}
where in~1 we used~\eqref{disp_vacuum} and in~2 the Weyl form~\eqref{Weyl} of the canonical commutation relations multiple times. Since $\pazocal{S}^*$ is convex and weak*-closed, and hence in particular closed with respect to the trace norm topology, we see that $\CC_1 = \overline{\co}\{\ketbra{\beta}:\beta\in \C\}\subseteq \pazocal{S}^*$. Noting that $\pazocal{S}^*$ is a cone, i.e., it is closed under multiplication by nonnegative scalars, we conclude that in fact $\CC_1^+ \subseteq \pazocal{S}^*$.

Let us now prove the opposite inclusion, again in the single-mode case. Pick $T\in\T{1}$ such that $\Tr[TK]\geq 0$ for all $K\in \pazocal{S}$; then
\begin{align*}
0 &\leq \liminf_{\lambda\to 1^-} \sum_{\mu,\nu=1}^n\! \psi_\mu^* \psi_\nu\, e^{\frac12 |\alpha_\mu-\alpha_\nu|^2}\, \Tr\left[T\, \lambda^{a^\dag\! a} \disp(\alpha_\mu - \alpha_\nu)\lambda^{a^\dag\! a}\right] \\
&\leq \sum_{\mu,\nu=1}^n\! \psi_\mu^* \psi_\nu\, e^{\frac12 |\alpha_\mu-\alpha_\nu|^2}\, \lim_{\lambda\to 1^-} \Tr\left[T\, \lambda^{a^\dag\! a} \disp(\alpha_\mu - \alpha_\nu)\lambda^{a^\dag\! a}\right] \\
&\texteq{3} \sum_{\mu,\nu=1}^n\! \psi_\mu^* \psi_\nu\, e^{\frac12 |\alpha_\mu-\alpha_\nu|^2}\, \Tr\left[T\, \disp(\alpha_\mu - \alpha_\nu)\right] \\
&= \sum_{\mu,\nu=1}^n\! \psi_\mu^* \psi_\nu\, e^{\frac12 |\alpha_\mu-\alpha_\nu|^2}\, \chi_T(\alpha_\mu - \alpha_\nu)
\end{align*}
for all $\alpha\in \C^n$ and $\psi\in \C^n$, where the function $\chi_T:\C\to \C$ defined by $\chi_T(\alpha)= \Tr[T \disp(\alpha)]$ is the characteristic function~\eqref{chi} of $T$. To prove~3, since $\disp(\alpha_\mu-\alpha_\nu)$ is bounded (actually, unitary) it suffices to show that $\lim_{\lambda\to 1^-} \left\|\lambda^{a^\dag\! a} T \lambda^{a^\dag\! a} - T\right\|_1=0$ for all trace class $T$. To see this, we decompose $T=T_+ - T_-$ into its positive and negative parts $T_\pm\geq 0$, which are also trace class operators. Note that 
\bbb
\lim_{\lambda\to 1^-} \Tr \left[\lambda^{2 a^\dag\! a} T_\pm\right] = \lim_{\lambda\to 1^-} \sum_{n=0}^\infty \lambda^{2n} \braket{n|T_\pm|n} = \sum_{n=0}^\infty \braket{n|T_\pm|n} = \Tr [T_\pm]
\eee
thanks to Abel's theorem, and therefore, by the gentle measurement lemma~\cite{Davies1969, VV1999} (see also~\cite[Lemma~9.4.2]{Wilde2}),
\bbb
\lim_{\lambda\to 1^-} \left\|T_\pm - \lambda^{a^\dag\! a}T_\pm \lambda^{a^\dag\! a}\right\|_1=0\, ,
\eee
in turn implying that
\bbb
\lim_{\lambda\to 1^-} \left\|\lambda^{a^\dag\! a} T \lambda^{a^\dag\! a} - T\right\|_1 \leq \lim_{\lambda\to 1^-} \left\|\lambda^{a^\dag\! a} T_+ \lambda^{a^\dag\! a} - T_+\right\|_1 + \lim_{\lambda\to 1^-} \left\|\lambda^{a^\dag\! a} T_- \lambda^{a^\dag\! a} - T_-\right\|_1 = 0\, .
\eee

We have just established that, for all $\alpha\in \C^n$, the matrix $\left( e^{\frac12 |\alpha_\mu-\alpha_\nu|^2}\, \chi_T(\alpha_\mu - \alpha_\nu)\right)_{\mu,\nu=1,\ldots, n}$ is positive semidefinite. This is known~\cite{Richter2002} to imply that $T=\lambda\sigma$ for some $\lambda\geq 0$ and some classical state $\sigma$, i.e., $T\in \CC_m^+$.

This latter claim can be also verified as follows. Applying the classical Bochner theorem, we see that the function $\C\ni \alpha \mapsto \varphi_T(\alpha)\coloneqq \chi_T(\alpha)\, e^{\frac12 |\alpha|^2}$ is the Fourier transform of a positive measure. Since $\varphi_T$ is well-known to be the Fourier transform of the $P$-function~\cite[Lemma~1]{Bach1986}, we conclude that the $P$-function of $T$ is non-negative, i.e., $T$ is a non-negative multiple of a classical state.

We conclude that $\CC_1^+=\pazocal{S}^*$, and hence that $\CC_1^+$ is weak*-closed. The exact same argument in fact shows that $\CC_m^+$ is weak*-closed for any finite number of modes $m$. Since the unit ball $B_m\coloneqq \left\{T\in \T{m}: \|T\|_1\leq 1\right\}$ of $\T{m}=\K{m}^*$ is weak*-compact by the Banach--Alaoglu theorem~\cite[Thm.~2.6.18]{MEGGINSON},
\bbb
\widetilde{\CC}_m = \co \left( \CC_m \cup \{0\} \right) = \CC_m^+ \cap B_m
\eee
is the intersection of a weak*-closed and a weak*-compact set, and hence it is itself weak*-compact.
\end{proof}

\subsection{Proof of Theorems~\ref{main result thm} and~\ref{bound_rates_NCMi_thm}} \label{main results subsection}

We are finally ready to present our main result about the measured relative entropy of nonclassicality.


\begin{proof}[Proof of Theorem~\ref{main result thm}]
Let us \tcb{use Lemma~\ref{restrictions on L}(i) to write an improved form of}~\eqref{rough_variational_NCM_alt} as
\tcb{
\begin{align*}
\NCM(\rho) =& \inf_{\sigma\in \widetilde{\CC}_m} \sup_{\substack{0<L\in \T{m},\\[.2ex] -\Tr \rho\log_2 L <\infty}} F_\rho(\sigma, L)\, , \\
F_\rho(\sigma, L) \coloneqq&\ \Tr [\rho \log_2 L] + \log_2(e) \left( 1 - \Tr [\sigma L] \right) .
\end{align*}
}
Now:
\begin{enumerate}[(i)]
\item $\widetilde{\CC}_m$ is weak*-compact by Lemma~\ref{Cm tilde compact lemma}, and manifestly convex;
\item $\{L\in \T{m}:\, L>0\tcb{,\ -\Tr \rho\log_2 L<\infty}\}$ is convex \tcb{thanks to the operator concavity of the logarithm};
\item $F_\rho (\cdot,L)$ is a convex (actually, convex-linear) function on $\widetilde{\CC}_m$ for every fixed $L>0$; by definition of weak* topology it is also weak*-continuous (because $L$ is also compact);
\item $F_\rho (\sigma,\cdot)$ is a concave function on $\{L\in \T{m}:\, L>0\tcb{,\ -\Tr \rho\log_2 L<\infty}\}$ for all $\sigma\in \widetilde{\CC}_m$, because $\log_2$ is operator concave; it is also upper semicontinuous with respect to the trace norm topology, because $\Tr[\rho\log_2 L]=-S(\rho)-D(\rho\|L)$, and $L\mapsto D(\rho\|L)$ is lower semicontinuous with respect to the weak topology~\cite[Corollary~5.12(i)]{PETZ-ENTROPY} and hence (Corollary~\ref{swot cor}) with respect to the trace norm topology, too.
\end{enumerate}
Since all assumptions of Sion's minimax theorem~\cite{sion} are satisfied, we can exchange infimum and supremum, and write
\begin{align*}
\NCM(\rho) &\texteq{1} \sup_{\tcb{\substack{0<L\in \T{m},\\[.2ex] -\Tr \rho\log_2 L <\infty}}} \inf_{\sigma\in \widetilde{\CC}_m} F_\rho(\sigma, L) \\
&\tcb{\ \textleq{2} \sup_{0<L\in \T{m}} \inf_{\sigma\in \widetilde{\CC}_m} F_\rho(\sigma, L)} \\
&= \sup_{0<L\in \T{m}} \inf_{\sigma\in \widetilde{\CC}_m} \left\{ \Tr [\rho \log_2 L] + \log_2(e) \left( 1 - \Tr [\sigma L ]\right)\right\} \\
&= \sup_{0<L\in \T{m}} \left\{ \Tr [\rho \log_2 L] + \log_2(e) \left( 1 - \sup_{\sigma \in \widetilde{\CC}_m} \Tr [\sigma L]\right) \right\} \\
&\texteq{\tcb{3}} \sup_{0<L\in \T{m}} \left\{ \Tr [\rho \log_2 L] + \log_2(e) \left( 1 - \max\left\{ \sup_{\alpha \in \C^m} \braket{\alpha| L|\alpha},\, 0\right\}\right) \right\} \\
&\texteq{\tcb{4}} \sup_{0<L\in \T{m}} \left\{ \Tr [\rho \log_2 L] + \log_2(e) \left( 1 - \sup_{\alpha \in \C^m} \braket{\alpha| L|\alpha} \right) \right\} \\
&\texteq{\tcb{5}} \sup_{0<L\in \T{m}} \left\{ \Tr [\rho \log_2 L] - \log_2 \sup_{\alpha \in \C^m} \braket{\alpha|L|\alpha} \right\} \\
&\textleq{\tcb{6}} \sup_{0<L\in \B{m}} \left\{ \Tr [\rho \log_2 L] - \log_2 \sup_{\alpha \in \C^m} \braket{\alpha|L|\alpha} \right\} \\
&\texteq{\tcb{7}} \Gamma(\rho)\, .
\end{align*}
Here, 1~is Sion's theorem~\cite{sion}, \tcb{in~2 we simply extended the supremum, 3}~comes from the fact that the extreme points of $\widetilde{\CC}_m$ are either coherent states or $0$, as it follows from~\eqref{Cm tilde}, \tcb{4}~holds because $L>0$, \tcb{5}~is proved by scale invariance of the expression on the sixth line exactly as in step~1 of the proof of Lemma~\ref{Berta variational lemma}, in~\tcb{6} we extended the supremum to all $0<L\in \B{m}$, and finally \tcb{7}~holds thanks to Lemma~\ref{variational_Gamma_lemma}. Since Proposition~\ref{Gamma_properties_prop} establishes that $\Gamma \leq \NCM$ on all states, we have actually proved that 
\bb
\NCM(\rho) = \sup_{0<L\in \T{m}} \left\{ \Tr [\rho \log_2 L] - \log_2 \sup_{\alpha \in \C^m} \braket{\alpha| L|\alpha} \right\} = \Gamma(\rho)\, .
\label{actually_proved}
\ee
The fact that $L$ can be taken to be a state follows by scale invariance.
\end{proof}

\begin{proof}[Proof of Corollary~\ref{superadditivity_cor}]
Thanks to Theorem~\ref{main result thm}, the function $\NCM$ inherits all properties of $\Gamma$, as established in Proposition~\ref{Gamma_properties_prop} and Corollary~\ref{Gamma_regularization_cor}, on the whole set of finite-entropy states. Given such a state $\rho$, the same Corollary~\ref{Gamma_regularization_cor} also shows that $\NCMi(\rho)\geq \NCM(\rho)$. On the other hand, regularizing the inequality $\NCM(\rho) \leq \NC(\rho)$ (Lemma~\ref{NC_subadd_lemma}) we see that $\NCMi(\rho)\leq \NCi(\rho)$. Remembering that $\NCi(\rho) \leq \NC(\rho)$ by Corollary~\ref{NCi_monotone_cor} concludes the proof of~\eqref{hierarchy}. Faithfulness of $\NCMi$ and hence of $\NCi$ on finite-entropy states follows from the fact that $\NCM$ itself is faithful (Lemma~\ref{NC_subadd_lemma}).
\end{proof}





We conclude this section with the proof of Theorem~\ref{bound_rates_NCMi_thm}.

\begin{proof}[Proof of Theorem~\ref{bound_rates_NCMi_thm}]
To establish~\eqref{bound_rates_NCMi}, we apply Theorem~\ref{general_bound_rates_thm} to the lower semicontinuous, weakly additive, and strongly superadditive nonclassicality monotone $\Gamma^\infty$ (Corollary~\ref{Gamma_regularization_cor}):
\bbb
R(\rho \to \sigma)\leq \widetilde{R}(\rho \to \sigma) \leq \frac{\Gamma^\infty(\rho)}{\Gamma^\infty(\sigma)} = \frac{\NCMi(\rho)}{\NCMi(\sigma)}\, ,
\eee
where the last equality is just~\eqref{actually_proved}, which is applicable because $S(\rho), S(\sigma)<\infty$. Finally, the last estimate in~\eqref{bound_rates_NCMi} is a simple application of~\eqref{hierarchy}.
\end{proof}

\section{Further properties of our nonclassicality monotones} \label{secondary res section}

We now present some additional results which can be useful in actual applications of Theorem~\ref{general_bound_rates_thm}. In particular, we present two different and independent bounds on $\NC$ and $\NCM$, and a technique for approximating them in the case of infinite rank states, where analytical methods or even numerical simulations might not be enough.

\subsection{Bounds on nonclassicality monotones}

We start however with a little --- hopefully instructive --- detour. In light of Proposition~\ref{bounded energy prop}, one may wonder whether $\NC$ and $\NCM$ can take the value $+\infty$ at all. We now set out to show that this may indeed be the case. Clearly, Proposition~\ref{bounded energy prop} implies that any state with this property must have infinite mean photon number.

\begin{prop} \label{NC infinite prop}
There exists a single-mode (infinite-energy) state $\rho\in \D{1}$ such that $\NCM(\rho)=\NC(\rho)=+\infty$, i.e., $D(\rho\|\sigma)= \DM(\rho\|\sigma) = +\infty$ for all classical states $\sigma \in \CC_m$ --- including those of infinite energy!
\end{prop}

\begin{proof}
Let
\bb
\rho \coloneqq \frac{6}{\pi^2} \sum_n \frac{1}{(n+1)^2}\, \ketbra{2^n}\, 
\ee
be a modified ``Basel-type state''\tcb{, where the $\ket{2^n}$ are Fock states}. \tcb{It is easy to see that $\rho$ has finite entropy.} Then, because of Theorem~\ref{main result thm} and  Lemma~\ref{NC_subadd_lemma}, we see that
\bbb
\NC(\rho) \geq \NCM(\rho) \tcb{=} \Gamma(\rho) = \sup_{h\in \B{m}} \left\{ \Tr[\rho h] - \log_2 \sup_{\alpha\in \C} \braket{\alpha | 2^h | \alpha} \right\} .
\eee
Now, set $h_N\coloneqq \frac{1}{3} \sum_{n=0}^N n \ketbra{2^n}$. Observe that
\bbb
\Tr[\rho\, h_N] = \frac{2}{\pi^2} \sum_{n=0}^N \frac{n}{(n+1)^2} \tends{}{N\to\infty} +\infty\, ,
\eee
while
\begin{align*}
\sup_{\alpha\in \C} \braket{\alpha | 2^{h_N} | \alpha} &= \sup_{\alpha\in \C} \braket{\alpha |\left( \sumno_{n=0}^N 2^{n/3} \ketbra{2^n} \right) | \alpha} \\
&= \sup_{\alpha\in \C} \sumno_{n=0}^N 2^{n/3} \frac{|\alpha|^{2^{n+1}} e^{-|\alpha|^2}}{(2^n)!} \\
&\leq \sumno_{n=0}^N 2^{n/3} \sup_{\alpha\in \C} \frac{|\alpha|^{2^{n+1}} e^{-|\alpha|^2}}{(2^n)!} \\
&= \sumno_{n=0}^N 2^{n/3} \sup_{t\geq 0} \frac{t^{2^n}e^{-t}}{(2^n)!} \\
&= \sumno_{n=0}^N 2^{n/3} \frac{2^{n 2^n}}{e^{2^n} (2^n)!} \\
&\tends{}{N\to \infty} \text{const} < \infty\, ,
\end{align*}
where the evaluation of the limit is made possible by the fact that 
\bbb
2^{n/3} \frac{2^{n 2^n}}{e^{2^n} (2^n)!} \sim \frac{2^{n/3}}{\sqrt{2\pi}\, 2^{n/2}} = \frac{1}{\sqrt{2\pi}\, 2^{n/6}}
\eee
by Stirling's formula, in the sense that the ratio between the left-hand and the right-hand sides tends to $1$ as $n\to \infty$. We conclude that
\bbb
\NC(\rho) \geq \NCM(\rho) \geq \Gamma(\rho) \geq \lim_{N\to \infty} \left\{ \Tr[\rho \, h_N] - \log_2 \sup_{\alpha\in \C} \braket{\alpha | 2^{h_N} | \alpha} \right\} = +\infty\, ,
\eee
as claimed. \tcb{Clearly, this construction is easily generalized to the multi-mode case, where it leads to the same conclusion.}
\end{proof}

\subsection{Estimates based on the Wehrl entropy} \label{Wehrl subsection}

We now go back to the problem of estimating our nonclassicality monotones $\NCM$ and $\NCMi$, already tackled in Proposition~\ref{bounded energy prop}. The next result gives another independent upper bound for $\NCM$ and $\NCMi$ in terms of the Wehrl entropy~\eqref{Wehrl}.

\begin{prop}\label{upper and lower bound for ncm prop}
For any \tcb{finite-entropy} $m$-mode state $\rho$, it holds that
\bb
-\log_2\left\|Q_\rho\right\|_\infty - S(\rho) - m \log_2\pi \leq \NCM(\rho) \leq \NCMi(\rho) \leq S_W(\rho) - S(\rho)\, ,
\label{bound Wehrl}
\ee
\tcb{where $\|Q_\rho\|_\infty\coloneqq \sup_{\alpha\in\C^m} |Q_\rho(\alpha)|$ is the sup norm of $Q_\rho$. If instead $S(\rho)=+\infty$, then $N_r^M(\rho)=N_r^{M,\infty}(\rho)=+\infty$ as well.}
\end{prop}

\begin{proof}
Let us start by proving that $\NCM(\rho) \leq S_W(\rho) - S(\rho)$ whenever $S(\rho)<\infty$. We are in the situation of Theorem~\ref{main result thm}, so that we can write
\bb
\begin{aligned}
\NCM(\rho) &\texteq{1} \sup_{\omega\in\D{m}}\left\{\Tr[\rho\log_2\omega] - \log_2 \sup_{\alpha\in \C^m} \braket{\alpha|\omega|\alpha} \right\} \\
&= \sup_{\omega\in\D{m}}\left\{ - S(\rho) - D(\rho\|\omega) - \log_2 \sup_{\alpha\in \C^m} \braket{\alpha|\omega|\alpha} \right\} \\
&= \sup_{\omega\in\D{m}}\left\{ - S(\rho) - D(\rho\|\omega) - \log_2 \left( \pi^m \|Q_\omega\|_\infty \right) \right\} \\
&\textleq{2} \sup_{\omega\in\D{m}}\left\{ - S(\rho) - \DKL\left(Q_\rho\|Q_\omega\right) - \log_2 \left( \pi^m \|Q_\omega\|_\infty \right) \right\} \\
&= \sup_{\omega\in\D{m}}\left\{ - S(\rho) + S_W(\rho) + \int d^{2m}\alpha\, Q_\rho(\alpha) \log_2 Q_\omega(\alpha) - \log_2 \|Q_\omega\|_\infty \right\} \\
&\textleq{3} S_W(\rho) - S(\rho)\, .
\end{aligned}
\label{upper and lower bound for ncm eq1}
\ee
Here, 1~is just Theorem~\ref{main result thm}, in~2 we applied the data processing inequality~\cite{lieb73a, lieb73b, lieb73c, Lindblad-monotonicity} (see also~\cite[Proposition~5.23(iv)]{PETZ-ENTROPY}) to the quantum-to-classical channel $\rho\mapsto Q_\rho$, which physically corresponds to a heterodyne detection~\cite[5.4.2]{BUCCO} (for an independent proof, see Lemma~\ref{stupid lemma}), and finally in~3 we noted that $Q_\omega(\alpha)\leq \|Q_\omega\|_\infty$ and remembered that $Q_\rho$ is a probability density function.

Since $\NCM(\omega)\leq S_W(\omega) - S(\omega)$ whenever $\omega$ has finite entropy, setting $\omega=\rho^{\otimes n}$ yields
\bbb
\NCMi(\rho) = \lim_{n\to\infty} \frac1n \NCM(\rho^{\otimes n}) \leq \lim_{n\to\infty} \frac1n \left( S_W(\rho^{\otimes n}) - S(\rho^{\otimes n})\right) = S_W(\rho) - S(\rho)\, ,
\eee
where in the last step we used the additivity of both the von Neumann entropy and the Wehrl entropy.


\tcb{To prove the lower bound on $\NCM$, we use Proposition~\ref{Gamma_properties_prop} together with the expression~\eqref{Gamma variational 3} for $\Gamma$. Start by denoting with $\Pi$ the orthogonal projector onto the kernel of $\rho$. Then for all $\epsilon>0$ we have that $\rho+\epsilon\Pi>0$ and moreover $\Tr\left[ \rho \log_2 (\rho+\epsilon\Pi) \right] = \Tr \left[ \rho \log_2 \rho\right]$, and hence
\begin{align*}
\NCM(\rho) &\geq \sup_{0<L\in \B{m}} \left\{ \Tr \left[\rho \log_2 L\right] - \log_2 \sup_{\alpha\in \C^m} \braket{\alpha|L|\alpha} \right\} \\
&\geq \limsup_{\epsilon\to 0^+} \left\{\Tr\left[ \rho \log_2 (\rho+\epsilon\tcb{\Pi}) \right] - \log_2 \sup_{\alpha\in \C^m} \braket{\alpha|(\rho+\epsilon\Pi)|\alpha} \right\} \\
&\geq \limsup_{\epsilon\to 0^+} \left\{ \Tr\left[ \rho \log_2 \rho\right] - \log_2 \left(\sup_{\alpha\in \C^m} \braket{\alpha|\rho|\alpha} + \epsilon \right) \right\} \\
&= \Tr\left[ \rho \log_2 \rho\right] - \log_2 \sup_{\alpha\in \C^m} \braket{\alpha|\rho|\alpha} \\
&= - S(\rho) - m \log_2\pi - \log_2\|Q_\rho\|_\infty
\end{align*}
Since it relies only on Proposition~\ref{Gamma_properties_prop}, this} \tcb{lower bound holds even if $S(\rho)=+\infty$, in which case it implies that $N_r^M(\rho)=N_r^{M,\infty}(\rho)=+\infty$.}
This completes the proof.
\end{proof}

We can immediately draw some interesting consequences concerning Gaussian states. Following the conventions of the excellent monograph by Serafini~\cite{BUCCO}, for an $m$-mode state $\rho$ we set $s_j\coloneqq \Tr [\rho R_j]$, with $j=1,\ldots, 2m$ and $R\coloneqq (x_1,p_1\ldots, x_m,p_m)^\intercal$, and define the quantum covariance matrix by $V_{jk}\coloneqq \Tr \left[\rho \left\{R_j,R_k\right\}\right] - 2s_j s_k$. Gaussian states are those whose characteristic function~\eqref{chi} is a multivariate Gaussian, and are uniquely characterized by the vector $s$ and the quantum covariance matrix $V$.

\begin{cor} \label{upper and lower bound for ncm for gaussian states cor}
Let $\rho$ be an arbitrary $m$-mode Gaussian state with quantum covariance matrix $V$. Then
\begin{equation*}
\frac12 \log_2\det (V+\id) - S(\rho) - m \leq \NCM(\rho) \leq \NCMi(\rho) \leq \frac12 \log_2 \det (V+\id) - S(\rho) + m \log_2 (e)\, .
\end{equation*}
\end{cor}

\begin{proof}
One just needs to remember that the Husimi function $Q_\rho$ of a Gaussian state $\rho$ with quantum covariance matrix $V$ is a Gaussian with (classical) covariance matrix $(V+\id)/2$ (to see this, just set $\sigma=V$ and $\sigma_m=\id$ in~\cite[Eq.~(5.139)]{BUCCO}). This implies immediately that $\|Q_\rho \|_\infty = \pi^{-m} \left( \det \left(\frac{V+\id}{2}\right)\right)^{-1/2}$, and that the Wehrl entropy of $\rho$ satisfies
\bb
S_W(\rho) = - \int d^{2m}\alpha\, Q_\rho(\alpha) \log_2 \left( \pi^m Q_\rho(\alpha)\right) = \frac12 \log_2 \det (V+\id) + m \log_2 (e)\, .
\label{Wehrl entropy Gaussian}
\ee
This concludes the proof.
\end{proof}

\subsection{Symmetries} \label{symmetries subsection}

A notion that we will often exploit is that of symmetry. Its implications for the variational program in Theorem~\ref{main result thm} are as follows. 

\begin{prop} \label{symmetric states ncm prop}
Let $\Lambda:\T{m}\to \T{m}$ be a classical operation on an $m$-mode system, and let $\rho\in \D{m}$ be an invariant state, in formula $\Lambda(\rho)=\rho$. Then we have that
\bb
\NCM(\rho) = \inf_{\sigma\in\Lambda(\CC_m)}\DM(\rho\|\sigma)\,. 
\label{symmetry restriction inf}
\ee
If $S(\rho)<\infty$, then it also holds that
\bb
\NCM(\rho) = \sup_{0<L \in \Lambda^\dag (\B{m})}\left\{ \Tr [\rho \log_2 L] - \log_2 \sup_{\alpha\in \C^m} \braket{\alpha|L|\alpha} \right\}\,.
\label{symmetry restriction sup}
\ee
\end{prop}

\begin{proof}
We start with~\eqref{symmetry restriction inf}, which follows from general and well-known arguments. We have that
\begin{align*}
\NCM(\rho) &= \inf_{\sigma\in\CCM}\DM(\rho\|\sigma) \textgeq{1} \inf_{\sigma\in\CCM} \DM\left(\Lambda(\rho)\|\Lambda(\sigma)\right) = \inf_{\sigma\in\CCM} \DM\left(\rho\|\Lambda(\sigma)\right) = \inf_{\sigma\in\Lambda(\CCM)} \DM(\rho\|\sigma)\,,
\end{align*}
where 1~holds because of the monotonicity under channels of $D^M$. Clearly, since restricting the infimum can only increase the value of the program, it also holds that $\NCM(\rho)\leq\inf_{\sigma\in\Lambda(\CCM)} \DM(\rho\|\sigma)$. This proves~\eqref{symmetry restriction inf}

To prove~\eqref{symmetry restriction sup}, we go back to~\eqref{Gamma_properties_eq1}. Assuming that $\Gamma((\Lambda(\rho)) = \Gamma(\rho) = \NCM(\rho)$, as implied by Theorem~\ref{main result thm}, the derivation in~\eqref{Gamma_properties_eq1} also shows that we can in fact restrict $L$ to belong to $\Lambda^\dag(\B{m})$.
\end{proof}

The above result is particularly useful when the state $\rho$ under examination is invariant under a group action.

\begin{cor} \label{group symmetric states cor}
Let $U:G\to \B{m}$ be a unitary representation of a compact group $G$ on the Hilbert space $\HH_m$. Assume that $U(g)$ maps coherent states to coherent states for all $g\in G$. Let $\rho\in \D{m}$ be a finite-entropy state such that is invariant under $G$, i.e., such that $U(g) \rho U(g)^\dag \equiv \rho$ for all $g\in G$. Then
\begin{align}
\NCM(\rho) &= \inf_{\sigma\in \CC_m^G} \DM(\rho\|\sigma) \label{group symmetry restriction inf} \\
&= \sup_{0<L \in \pazocal{B}_{\mathrm{sa}}^{G}(\HH_m)} \left\{ \Tr [\rho \log_2 L] - \log_2 \sup_{\alpha\in \C^m} \braket{\alpha|L|\alpha} \right\} , \label{group symmetry restriction sup}
\end{align}
where a superscript $G$ denotes that we restrict to $G$-invariant operators.
\end{cor}

\begin{proof}
It suffices to apply Proposition~\ref{symmetric states ncm prop} to the totally symmetrizing map
\begin{equation}
\Lambda_G:T\longmapsto \Lambda_G \left(T\right) \coloneqq \int_G U(g) T U^\dag (g)\, d\mu(g)\, ,
\label{totally symmetrizing}
\end{equation}
where $\mu$ denotes the \tcb{left} Haar measure on $G$, and the integral on the right-hand side is to be underdstood in the Bochner sense. Note that $\Lambda_G$ is a classical channel, because each $U(g)$ maps coherent states to coherent states, and the set of classical channels is convex.
\end{proof}




\section{Applications} \label{applications section}

To get a feeling of how tight the estimates in Theorem~\ref{bound_rates_NCMi_thm} for asymptotic transformation rates in the QRT of nonclassicality really are, we need to design distillation protocols that can provide lower bounds on those rates. To do so, we have to first compute or bound the resource content of the states we work with. Before going on, we fix some notation. Consider a two-mode CV quantum system with annihilation operators $a,b$, and pick $\lambda\in [0,1]$. The \textbf{beam splitter} with transmissivity $\lambda$ is represented by the unitary
\bb
U_\lambda \coloneqq e^{\arccos\sqrt{\lambda}\, ( a^\dag b - ab^\dag )}\, .
\ee
Its action on operators and vectors is given by
\begin{align}
U_\lambda \smatrix{ a \\ b } U_\lambda^\dag &= \smatrix{ \sqrt{\lambda} & - \sqrt{1-\lambda} \\ \sqrt{1-\lambda} & \sqrt{\lambda} } \smatrix{ a \\ b } , \\
U_\lambda\, \disp\left(\!\smatrix{ \alpha \\ \beta }\!\right) U_\lambda^\dag &= \disp\left(\!\smatrix{ \sqrt{\lambda} & \sqrt{1-\lambda} \\ -\sqrt{1-\lambda} & \sqrt{\lambda} }\! \smatrix{ \alpha \\ \beta }\right) .
\end{align}
Therefore, thanks to~\eqref{disp_vacuum} we see that
\bb
U_\lambda \ket{\alpha}\ket{\beta} = \ket{\sqrt\lambda \alpha + \sqrt{1-\lambda}\beta}\ket{-\sqrt{1-\lambda}\alpha + \sqrt\lambda \beta}\, .
\label{beam splitter coherent states}
\ee

\subsection{Fock diagonal states} \label{Fock-diagonal subsec}

We now compute or estimate our nonclassicality monotones for some multimode Fock-diagonal states. Denoting with $\{\ket{n}\}_{n\in \N^m}$ the Fock basis, as usual, define the \textbf{totally dephasing map} $\Delta$ by
\bb
\Delta(\rho) \coloneqq \sum_{n\in \N^m} \ketbra{n}\rho\ketbra{n}\, .
\label{Delta}
\ee
This is a classical channel because it is of the form~\eqref{totally symmetrizing}, for $G=(S^1)^{\times m}\simeq [0,2\pi)^m$ and $U(\varphi) = e^{i \sum_j \varphi_j a_j^\dag\! a_j}$. In other words,
\bbb
\Delta(\rho)=\frac{1}{2\pi}\int_0^{2\pi} d^m \varphi\, e^{i \sum_j \varphi_j a_j^\dag\! a_j} \rho e^{-i \sum_j \varphi_j a_j^\dag\! a_j}\, .
\eee
Clearly, the unitary $e^{i \sum_j \varphi_j a_j^\dag\! a_j}$, which is nothing but a phase space rotation, sends coherent states to coherent states.

Applying Corollary~\ref{group symmetric states cor} to any finite-entropy Fock-diagonal state $\rho\in \CCFD$ then yields
\begin{equation}
\NCM(\rho) = \NCMi(\rho) = \NCi(\rho) = \NC(\rho) = \inf_{\sigma\in\CCMFD} D(\rho\|\sigma) = \sup_{0<L \in \Delta (\B{1})}\left\{ \Tr [\rho \log_2 L] - \log_2 \sup_{\alpha\in \C^m} \braket{\alpha|L|\alpha} \right\} ,
\label{symmetry inf&sup FD}
\end{equation}
where the equalities $\NCM(\rho) = \NCMi(\rho) = \NCi(\rho) = \NC(\rho)$ comes from the fact that the optimal state $\sigma$, being Fock-diagonal, commutes with $\rho$, and $\DM(\rho\|\sigma)=D(\rho\|\sigma)$ whenever $[\rho,\sigma]=0$; thus, $\NCM(\rho) = \NC(\rho)$, which in turn makes the whole hierarchy~\eqref{hierarchy} collapse.

We now look at single-mode Fock-diagonal states with finite rank, since these will commonly be encountered in experimental applications.

\begin{prop} \label{finiterankfdstates}
Let $\rho$ be a single-mode Fock-diagonal state with finite rank. Let $M\coloneqq\max\{n:\braket{n|\rho|n} \neq 0\}$. Then in~\eqref{symmetry inf&sup FD} we can also take $L$ to have the same support as that of $\rho$ (and to be positive there only). In formula,
\begin{equation}
\NCM(\rho) = \sup_{L\in \widetilde{B}_{\mathrm{sa}}^{\mathrm{FD}}(\HH_1)} \left\{\Tr[\rho\log_2 L] - \log_2 \sup_{\alpha\in\left[0,\sqrt{M}\right]}\braket{\alpha|L|\alpha} \right\}\,,
\label{finiterankfdstates equation}
\end{equation}
where $\widetilde{B}_{\mathrm{sa}}^{\mathrm{FD}}(\HH_1) \coloneqq \left\{L\in\B{1}:\ L=\Delta(L),\ \supp L = \supp \rho,\ P_\rho L P_\rho > 0\right\}$, and $P_\rho:\HH_1\to \supp\rho$ is the projector onto the finite-dimensional space $\supp\rho$.
\end{prop}

\begin{proof}
We have that
\begin{align*}
\NCM(\rho) &= \sup_{0 < L\in\Delta(\B{1})} \left\{\Tr[\rho\log_2 L] - \log_2 \sup_{\alpha\in\C} \braket{\alpha|L|\alpha}\right\} \\
&\texteq{1} \sup_{0 < L\in\Delta(\B{1})} \left\{ \Tr\left[\rho\log_2 \left(P_\rho L P_\rho\right)\right] - \log_2 \sup_{\alpha\in\C} \braket{\alpha|L|\alpha} \right\} \\
&\textleq{2} \sup_{0 < L\in\Delta(\B{1})} \left\{ \Tr\left[\rho\log_2 \left(P_\rho L P_\rho\right)\right] - \log_2 \sup_{\alpha\in\C} \braket{\alpha|P_\rho L P_\rho|\alpha} \right\} \\
&= \sup_{L\in \widetilde{B}_{\mathrm{sa}}^{\mathrm{FD}}(\HH_1)} \left\{\Tr[\rho\log_2 L] - \log_2 \sup_{\alpha\in \C} \braket{\alpha|L|\alpha} \right\} \\
&\texteq{3} \sup_{L\in \widetilde{B}_{\mathrm{sa}}^{\mathrm{FD}}(\HH_1)} \left\{\Tr[\rho\log_2 L] - \log_2 \sup_{\alpha\in \left[0,\sqrt{M}\right]} \braket{\alpha|L|\alpha} \right\} .
\end{align*}
Here: 1~follows because
\tcb{$[\rho,L]=0$ and hence $\Tr[\rho \log_2 L]=\Tr[\rho P_\rho(\log_2 L) P_\rho]=\Tr[\rho \log_2 (P_\rho L P_\rho)]$} (with a slight abuse of notation, we thought of $P_\rho$ as having the entire $\HH_1$ as codomain); 2~holds thanks to the fact that $L\geq P_\rho L P_\rho$ as both $L$ and $P_\rho$ are Fock-diagonal and hence commute; finally, in~3 we noticed that for $|\alpha|^2>M$ and for $L=\sum_{n=0}^M \ell_n \ketbra{n}$ the function
\bbb
\braket{\alpha|L|\alpha} = e^{-|\alpha|^2} \sum_{n=0}^M \frac{|\alpha|^{2n} \ell_n}{n!}
\eee
becomes monotonically decreasing in $|\alpha|$, essentially because it is a sum of monotonically decreasing functions. 
\end{proof}

\begin{rem} \label{rem truncation fdstates}
From Corollary~\ref{approximation ncm cor} we know that
\begin{equation}
\NCM(\rho_n) \tendsn{} \NCM(\rho)\,,
\end{equation}
where $\rho_n$ is the spectral truncation of the Fock-diagonal state $\rho$. Therefore, in principle we can use Proposition~\ref{finiterankfdstates} to approximate numerically $\NCM(\rho)$ for any Fock-diagonal state $\rho$ with arbitrary precision. Explicit estimates of the error associated with each truncation can be deduced from Corollary~\ref{approximation ncm cor}.
\end{rem}

The simplest example of Fock diagonal states is naturally given by Fock states themselves.\footnote{LL acknowledges useful discussions with Andreas Winter and Krishna Kumar Sabapathy on the problem of calculating $\NC(\ketbra{n})$.}

\begin{lemma} \label{NCM Fock lemma}
For a Fock state $\ket{n}$ we have that
\bb
\NCM(\ketbra{n}) = \NCMi(\ketbra{n}) = \NCi(\ketbra{n}) = \NC(\ketbra{n}) = \log_2 \left(\frac{n!e^n}{n^n}\right) = \frac12 \log_2 (2\pi n) + O(n^{-1})\, .
\label{NCM Fock}
\ee
\end{lemma}

\begin{proof}
The optimization in~\eqref{finiterankfdstates equation} involves a single parameter and is thus elementary. To deduce the asymptotic expansion on the righmost side, it suffices to apply Stirling's formula.
\end{proof}




Another example of Fock diagonal state is a noisy Fock state, e.g., a Fock state mixed with a certain amount of thermal noise. These states, herafter called \textbf{noisy Fock states}, are defined by
\begin{equation}
\rho_{n,\nu}(p)\coloneqq p\ketbra{n}+(1-p)\tau_\nu\, ,
\label{noisyfock}
\end{equation}
where the thermal state $\tau_\nu$ is given in~\eqref{tau}. In principle, we can approximate the exact value of $\NCM(\rho_{n,\nu}(p))$ with arbitrary precision for any $n$ and $\nu$, as pointed out in Remark~\ref{rem truncation fdstates}. Let us first consider the simpler case $\nu=0$, which is a good approximation in certain regimes, e.g., optical frequencies at room temperature. The state then becomes $\rho_{n,0}(p) = p\ketbra{n}+(1-p)\ketbra{0}$, and thanks to Proposition~\ref{finiterankfdstates} we can assume $L$ to be in the form $L = \ell \ketbra{n} + \ketbra{0}$ (we already exploited the scale invariance). Now we have to perform just two nested optimizations over one real parameter each, that is,
\begin{equation}
\NCM(\rho_{n,0}(p)) = \sup_{\ell>0} \left\{p\log_2 \ell - \log_2 \max_{\alpha\in\left[0,\sqrt{n}\right]} e^{-\alpha^2}\left(1+\frac{\ell\alpha^{2n}}{n!}\right)\right\} .
\end{equation}
For $n\leq 4$ the above program can even be solved analytically, since the inner maximization reduces to solving a $n$-th order algebraic equation. For example, for $n=1$ one simply finds $\beta=\sqrt{p}$, $\ell=1/(1-p)$ and $\NCM(\rho_{1,0}(p))=p+(1-p)\log_2 (1-p)$. 
The case of a nonzero temperature can be tackled by considering truncations of $\rho$ and performing numerical optimizations until some tolerance threshold is achieved. The results for different values of $\nu$ and $n$ are reported in Figures~\ref{noisyfockstateplotfixedn} and~\ref{noisyfockstateplotfixedT}.

\begin{figure}
\centering
\subfigure[~1 photon]{\includegraphics[width=81mm]{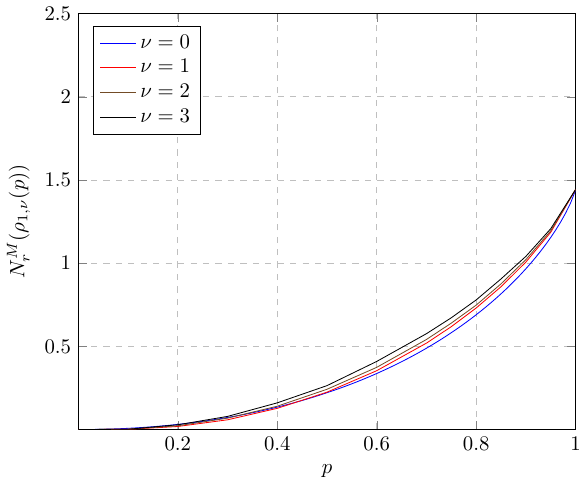}}
\subfigure[~2 photons]{\includegraphics[width=81mm]{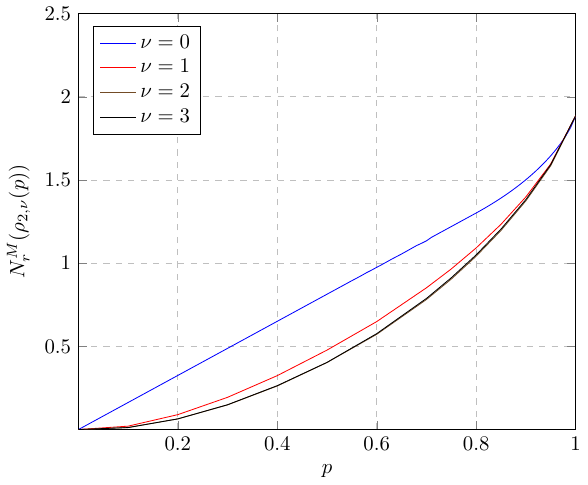}}
\subfigure[~3 photons]{\includegraphics[width=81mm]{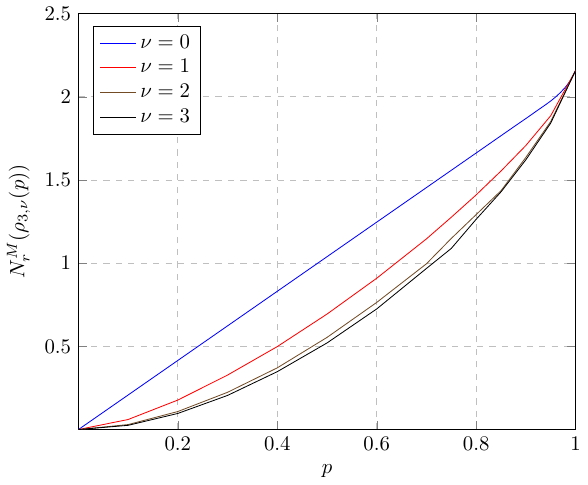}}
\subfigure[~4 photons]{\includegraphics[width=81mm]{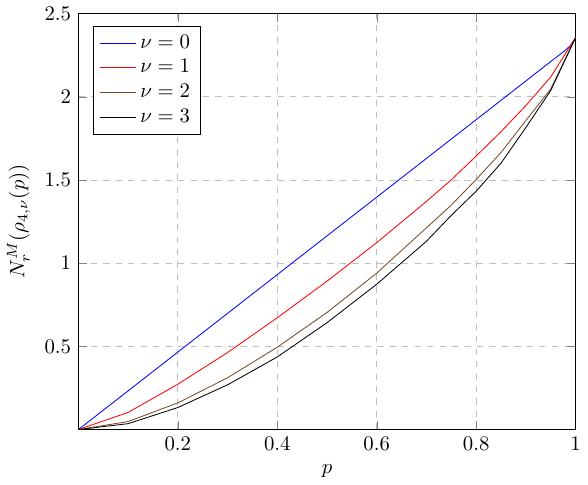}}
\caption{Nonclassicality for noisy Fock states: varying $\nu$ at fixed $n$.}
\label{noisyfockstateplotfixedn}
\end{figure}

\begin{figure}
\centering
\subfigure[~$\nu$=0]{\includegraphics[width=81mm]{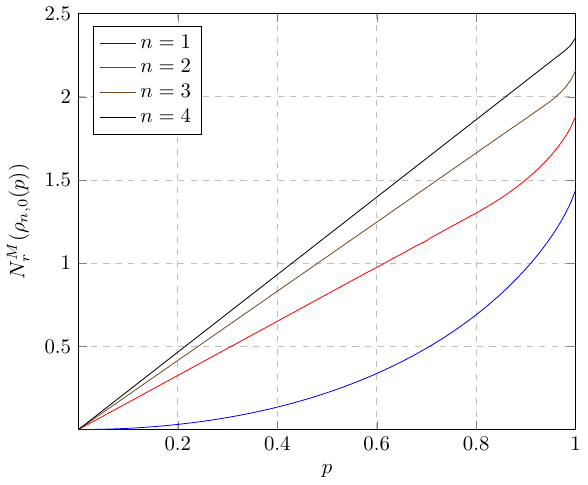}}
\subfigure[~$\nu$=1]{\includegraphics[width=81mm]{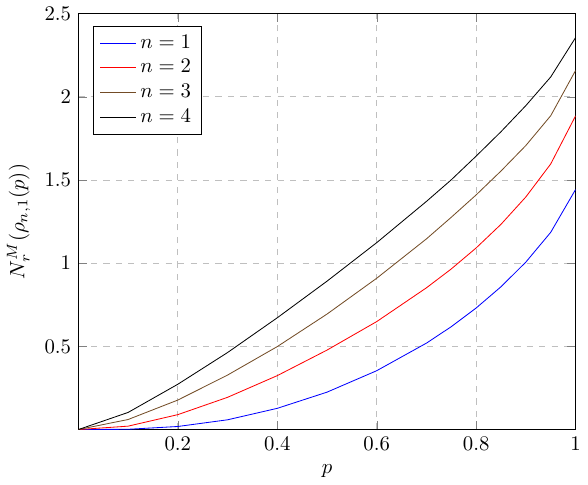}}
\subfigure[~$\nu$=2]{\includegraphics[width=81mm]{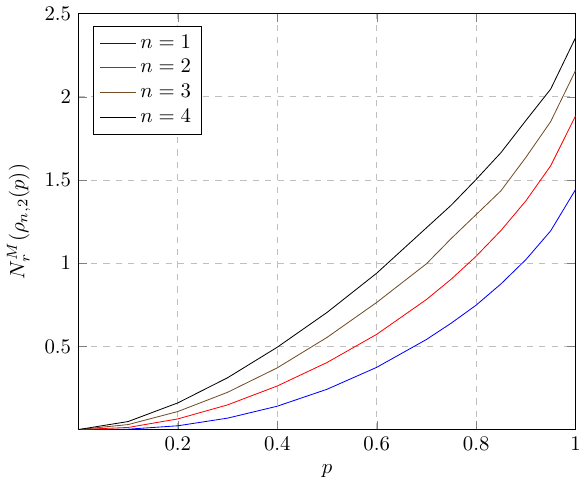}}
\subfigure[~$\nu$=3]{\includegraphics[width=81mm]{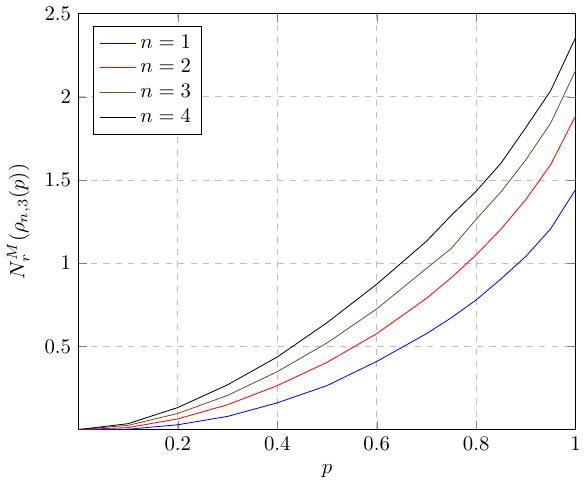}}
\caption{Nonclassicality for noisy Fock states: varying $n$ at fixed $\nu$.}
\label{noisyfockstateplotfixedT}
\end{figure}

\subsection{Schr\"{o}dinger cat states} \label{cat states subsec}

For $\alpha\in \C$, the associated \textbf{Schr\"{o}dinger cat states} (or simply \textbf{cat state}) is defined by~\cite{Dodonov1974}
\bb
\ket{\psi_{\alpha}^\pm} \coloneqq \frac{1}{\sqrt{2\left( 1\pm e^{-2|\alpha|^2}\right)}} \left(\ket{\alpha}\pm \ket{-\alpha}\right) .
\ee
It is a nonclassical state for all $\alpha\neq 0$. Since a phase space rotation acts as $e^{i\varphi a^\dag a}\ket{\psi_\alpha^\pm} = \ket{\psi_{e^{i\varphi}\alpha}^\pm}$, and all of our nonclassicality monotones are left invariant by such transformations, in what follows we can without loss of generality assume that $\alpha\in \R$.
Now, for a cat state with real $\alpha$, we can consider the group $G=\mathds{Z}^2$ and its representation $U:G\to \B{1}$ given by the reflection with respect to the real and/or imaginary axis. Applying Corollary~\ref{group symmetric states cor} to this setting (with $m=1$) shows immediately that~\eqref{group symmetry restriction inf}--\eqref{group symmetry restriction sup} hold with $\CC_1^G$ and $\pazocal{B}_{\mathrm{sa}}^{G}(\HH_1)$ being the sets of classical states and bounded operators that are invariant under reflections with respect to the real and/or imaginary axis.
A lower bound for $\NCM(\psi_\alpha^\pm)$ can be easily computed by setting a maximum rank for $L$ in the second line of~\eqref{group symmetry restriction sup} and then optimizing numerically. When $\rk L\leq 3$, in order to preserve the symmetry, $L$ must be supported on the subspace $V=\Span(\ket{\alpha},\ket{-\alpha},\ket{0})$. Analogously, an upper bound for $\NC(\psi_\alpha^\pm)$ can be found with a classical $\sigma$ belonging to $V$. In Figure~\ref{catstateplot} we report these two bounds for the even cat state $\psi_\alpha^+$, and an analogous lower bound for $N^M_r(\ketbra{\psi^-_\alpha})$.

\begin{figure}
\centering
\subfigure[~Even cat state]{\includegraphics[width=81mm]{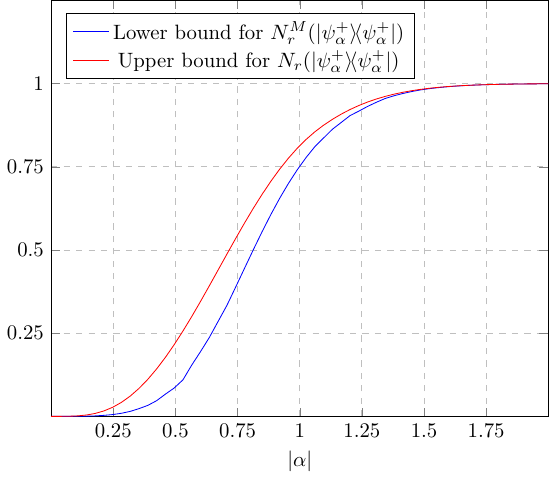}}
\subfigure[~Odd cat state]{\includegraphics[width=79.2mm]{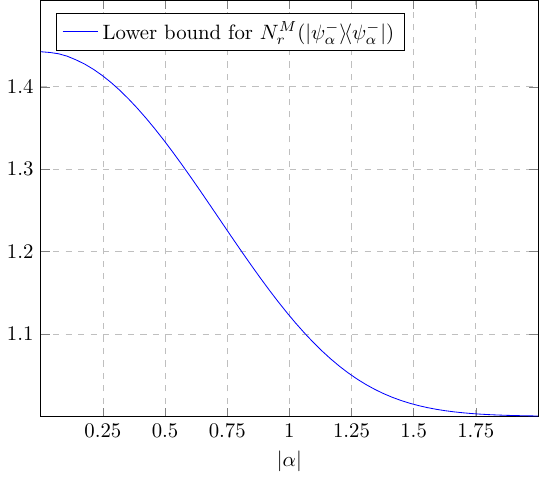}}
\caption{Bounds for the nonclassicality of a cat state, for different values of $|\alpha|$}
\label{catstateplot}
\end{figure}

\subsection{Squeezed states} \label{squeezed subsec}

A single-mode \textbf{squeezed vacuum state} is defined by~\cite[Eq.~(3.7.5)]{BARNETT-RADMORE}
\bb
\ket{\zeta_{r,\phi}}=\frac{1}{\sqrt{\cosh(r)}} \sum_{n=0}^\infty\sqrt{\binom{2n}{n}} \left(-\frac12\, e^{i\phi} \tanh(r)\right)^n \ket{2n}\,.
\ee
Since changing $\phi$ amounts to a simple rotation in phase space, and this cannot modify the value of any of our nonclassicality monotones, we will assume $\phi=0$ from now on. A squeezed state $\zeta_r\coloneqq \zeta_{r,0}$ has always finite energy $E(\psi_r)=\sinh^2(r)$, and hence we can use Proposition~\ref{bounded energy prop} to get the upper bound
\bbb
\NC(\zeta_r) \leq g(\sinh^2(r)) = 2\log_2 \cosh r - 2 \sinh^2(r) \log_2\tanh(r)\,.
\eee
A second upper bound on $\NC$ can be found by considering a (classical) squeezed thermal state
\bbb
\sigma_s = S(s)\tau_{N(s)}S^\dagger(s)=\sqrt{\frac{2}{\pi (e^{4s}-1)}} \int_{-\infty}^{+\infty} dt\, e^{-\frac{2t^2}{e^{4s}-1}}\, \ketbra{it}\, ,\quad s\geq 0\, ,
\eee
\tcb{where $S(s)$ is the usual squeezing unitary with real parameter $s$}, and plugging it in the infimum that defines $\NC$ (cf.~\eqref{NC_and_NCM}), i.e.,
\begin{align*}
\NC(\zeta_r)\leq&\ \inf_{s\geq 0} D(\zeta_r\| \sigma_s) = \inf_{s\geq 0} \left( \log_2 (1+N(s)) + 2\sinh^2(r-s) \log_2 \left(1+\frac{1}{N(s)}\right) \right),
\end{align*}
\tcb{where $N(s) \coloneqq \frac{e^{2s}-1}{2}$.} The 
\tcb{rightmost side of the above} expression can be easily optimized numerically. A lower bound on $\NCM$ can be found from Corollary~\ref{upper and lower bound for ncm for gaussian states cor}. All these estimates are plotted in Figure~\ref{squeezedstateplot}.

\begin{figure}
\centering
\includegraphics[width=100mm]{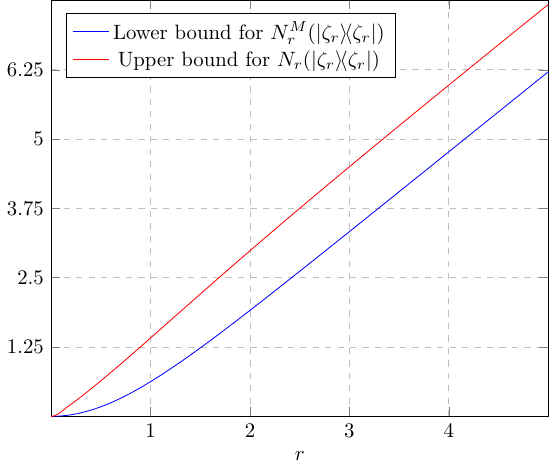}
\caption{Bounds for the nonclassicality of a squeezed state, for different values of $r$.}
\label{squeezedstateplot}	
\end{figure}

\subsection{Fock state dilution}

Now we are ready to report an example in which the bound in Theorem~\ref{bound_rates_NCMi_thm} is (asymptotically) tight.

\begin{prop}
Let $0<p\leq 1$ and $n\geq 2$ be fixed. Consider the transformation $\rho_{n,0}(p)\to \ketbra{n\!-\!1}$, where the noisy Fock state is defined in~\eqref{noisyfock}. It holds that
\bb
p \leq R\left(\rho_{n,0}(p)\to \ketbra{n\!-\!1}\right) \leq \frac{p \log_2\left( \frac{n! e^n}{n^n}\right)}{\log_2\left( \frac{(n-1)! e^{n-1}}{(n-1)^{n-1}}\right)} \tendsn{} p\, ,
\ee
with the upper bound being given by Theorem~\ref{bound_rates_NCMi_thm}.
\end{prop}

\begin{proof}
We start with the lower bound. Consider the following protocol, implemented with only linear optics, destructive measurements, and feed forward.
\begin{enumerate}[(1)]
    \item We send $\rho_{n,0}(p)$ into a beam splitter with transmissivity $\lambda$ whose second mode's initial state is the vacuum.
    \item We perform photon counting on the ancillary mode.
    \item If we measure 0 photons, the output state of the remaining mode is $\rho_{n,0}(p')$, with $p'\coloneqq \frac{p\lambda^n}{p\lambda^n + 1 - p}$. We restart with step (1).
    \item If we measure 1 photon, the output state of the remaining mode is $\ket{n-1}$, and we have succeeded.
    \item If we measure 2 or more photons, the protocol is aborted.
\end{enumerate}
Using the well-known formula~\cite{knight2002}
\bb
U_\lambda \ket{n,0} = \lambda^{\frac{n}{2}} \sum_{\ell=0}^{n} (-1)^\ell \sqrt{\binom{n}{\ell}} \left(\frac{1-\lambda}{\lambda}\right)^{\frac{\ell}{2}} \ket{n-\ell, \ell}\, ,
\ee
a lengthy but straightforward calculation shows that the global probability of success of this protocol is
\bb
P_{s}(n,p;\lambda) = \frac{pn (1-\lambda)\lambda^{2n-1}}{(1-\lambda^n)(p\lambda^n + 1-p)}\, .
\ee
Since we can take $\lambda$ arbitrarily close to $1$, we see that
\bbb
R\left(\rho_{n,0}(p)\to \ketbra{n\!-\!1}\right) \geq \lim_{\lambda\to 1^-} P_{s}(n,p;\lambda) = p\, ,
\eee
which proves the lower bound.

As for the upper bound, using~\eqref{symmetry inf&sup FD} together with convexity and~\eqref{NCM Fock}, we see that
\bbb
\NCMi(\rho_{n,0}(p)) = \NCMi(\rho_{n,0}(p)) \leq p \NCM(\ketbra{n}) = p \log_2\left( \frac{n! e^n}{n^n} \right) .
\eee
Leveraging once again~\eqref{NCM Fock}, this entails that
\bbb
R\left(\rho_{n,0}(p)\to \ketbra{n\!-\!1}\right) \leq \frac{p \log_2\left( \frac{n! e^n}{n^n}\right)}{\log_2\left( \frac{(n-1)! e^{n-1}}{(n-1)^{n-1}}\right)} \tendsn{} p\, .
\eee
This completes the proof.
\end{proof}


\subsection{Cat state manipulation}

Let us now discuss some protocols to transform cat states, enlarging or reducing their amplitude $\alpha$. Hereafter we take without loss of generality $\alpha$ to be real. The first transformation we consider is amplification: $\psi_\alpha^+\to \psi_{\sqrt2 \alpha}^+$. Lund et al.~\cite{Lund2004} have provided a protocol that achieves exact conversion of two copies of the initial state with probability
\bb
P_{\mathrm{Lund}}\left(\psi_\alpha^+ \otimes \psi_\alpha^+ \to \psi_{\sqrt2 \alpha}^+\right) = \frac{e^{-\alpha^2} \cosh(2\alpha^2) \sinh^2\left(\alpha^2/2\right)}{\cosh^2(\alpha^2)}\, .
\label{Lund prob}
\ee
Hence,
\bb
R\left(\psi_\alpha^+ \to \psi_{\sqrt2 \alpha}^+\right) \geq \frac12 P_{\mathrm{Lund}}\left(\psi_\alpha^+ \otimes \psi_\alpha^+ \to \psi_{\sqrt2 \alpha}^+\right) = \frac{e^{-\alpha^2} \cosh(2\alpha^2) \sinh^2\left(\alpha^2/2\right)}{2 \cosh^2(\alpha^2)}\, .
\label{Lund rate}
\ee
Mimicking the protocol of Lund et al.\ but employing slightly better (yet less realistic) measurements, we are able to obtain a better bound.

\begin{prop} \label{Ferrari_prop}
In the QRT of nonclassicality it is possible to achieve exact conversion $\psi_\alpha^+ \otimes \psi_\alpha^+ \to \psi_{\sqrt2 \alpha}^+$ with probability
\bb
P_{\mathrm{our}}\left(\psi_\alpha^+ \otimes \psi_\alpha^+ \to \psi_{\sqrt2 \alpha}^+\right) = \frac12 \tanh^2(\alpha^2)\, .
\label{Ferrari prob}
\ee
Therefore,
\bb
R\left(\psi_\alpha^+ \to \psi_{\sqrt2 \alpha}^+\right) \geq \frac14\tanh^2(\alpha^2)\, .
\label{Ferrari rate}
\ee
\end{prop}

\begin{proof}
Consider the following protocol. Apply a beam splitter with trasmissivity $1/2$ to the initial state $\ket{\psi^+_\alpha} \ket{\psi^+_\alpha}$. Using~\eqref{beam splitter coherent states}, we obtain that
\begin{align*}
    U_{1/2} \ket{\psi^+_\alpha} \ket{\psi^+_\alpha} &= \frac{1}{2\left( 1+ e^{-2\alpha^2}\right)} \left( \ket{0}\ket{\sqrt2 \alpha} + \ket{0}\ket{- \sqrt2 \alpha} + \ket{\sqrt2 \alpha}\ket{0} + \ket{-\sqrt2 \alpha}\ket{0} \right) \\
    &= \frac{\sqrt{\cosh(2\alpha^2)}}{2 \cosh(\alpha^2)} \left( \ket{0}\ket{\psi_{\sqrt2 \alpha}^+} + \ket{\psi_{\sqrt2 \alpha}^+} \ket{0} \right) .
\end{align*}
Carrying out on the second mode the measurement $\{\ketbra{\chi}, \id-\ketbra{\chi}\}$, with
\bbb
\ket{\chi} \coloneqq \frac{1}{\sqrt2 \sinh(\alpha^2)}\left( \sqrt{\cosh(2\alpha^2)} \ket{0} - \ket{\psi_{\sqrt2 \alpha}^+} \right) ,
\eee
yields
\bbb
\tensor[_2]{\braket{\chi|U_{1/2} |\psi_\alpha^+ \psi_\alpha^+}}{_{12}} = \frac{1}{\sqrt2}\tanh(\alpha^2)\ket{\psi_{\sqrt2 \alpha}^+}\, ,
\eee
where the subscripts identify different modes. Computing the norm of the above vector yields~\eqref{Ferrari prob} and this in turn~\eqref{Ferrari rate}.
\end{proof}

We now move on to cat state dilution. We consider the slightly simpler task of balanced dilution $\psi_{\sqrt2 \alpha}^+ \to \psi^+_\alpha \otimes \psi_\alpha^-$.

\begin{prop} \label{sign_randomized_dilution_cat_prop}
In the QRT of nonclassicality it holds that
\bb
R\left(\psi_{\sqrt2 \alpha}^+ \to \psi^+_\alpha \otimes \psi_\alpha^- \right) \geq \frac{\sinh^2(\alpha^2)}{2 \cosh(2\alpha^2)}\, .
\label{my rate}
\ee
\end{prop}

\begin{proof}
Consider the following protocol. Apply a beam splitter with trasmissivity $1/2$ to the initial state $\ket{\psi^+_{\sqrt2 \alpha}} \ket{0}$. Using one again~\eqref{beam splitter coherent states}, we obtain that
\begin{align*}
    U_{1/2} \ket{\psi^+_{\sqrt2\alpha}} \ket{0} &= \frac{1}{\sqrt{2\left(1 + e^{-4\alpha^2} \right)}} \left(\left(1+e^{-2\alpha^2}\right)\ket{\psi_\alpha^+}\ket{\psi_\alpha^+} + \left(1-e^{-2\alpha^2}\right) \ket{\psi_\alpha^-}\ket{\psi_\alpha^-} \right) .
\end{align*}
Therefore, measuring the second mode in the orthonormal basis whose first two elements are $\ket{\psi_\alpha^+}$ and $\ket{\psi_\alpha^-}$, we obtain that
\begin{align*}
    \tensor[_2]{\braket{\psi^+_\alpha | U_{1/2} |\psi^+_{\sqrt2\alpha}, 0}}{_{12}} &= \frac{\cosh(\alpha^2)}{\sqrt{\cosh(2\alpha^2)}}\, \ket{\psi^\pm_\alpha}\, , \\
    \tensor[_2]{\braket{\psi^-_\alpha | U_{1/2} |\psi^+_{\sqrt2\alpha}, 0}}{_{12}} &= \frac{\sinh(\alpha^2)}{\sqrt{\cosh(2\alpha^2)}}\, \ket{\psi^\mp_\alpha}\, .
\end{align*}
Computing the norms of the vectors on the right-hand side yields the estimates
\begin{align*}
P_{\mathrm{our}} \left(\psi_{\sqrt2 \alpha}^+ \to \psi_{\alpha}^+\right) &= \frac{\cosh^2(\alpha^2)}{\cosh(2\alpha^2)}\, , \\
P_{\mathrm{our}} \left(\psi_{\sqrt2 \alpha}^+ \to \psi_{\alpha}^-\right) &= \frac{\sinh^2(\alpha^2)}{\cosh(2\alpha^2)}\, .
\end{align*}

Applying the above protocol to $n$ copies of $\psi_{\sqrt2 \alpha}^+$ yields, in the limit of large $n$, at least $\frac{n \sinh^2(\alpha^2)}{2\cosh(2\alpha^2)}$ copies of $\psi_{\alpha}^+\otimes \psi_{\alpha}^-$. Hence,
\bbb
R\left(\psi_{\sqrt2 \alpha}^+ \to \psi^+_\alpha \otimes \psi_\alpha^- \right) \geq \frac{\sinh^2(\alpha^2)}{2 \cosh(2\alpha^2)}\, ,
\eee
which completes the proof.
\end{proof}

Finally, we upper bound the maximal asymptotic transformation rates of both amplification and dilution of cat states by means of the formula~\ref{bound_rates_NCMi}, and the numerical results reported in Figure~\ref{protocols_fig}.



\section{Acknowledgments}
\tcr{GF and LL dedicate this work to the memory of Piero Angela, whose career devoted to the popularization of science, spanning more than 50 years, was of immense inspiration for generations of Italian scientists~\cite{PA}.} GF and LL contributed equally to this paper. LL is grateful to Krishna Kumar Sabapathy and Andreas Winter for enlightening discussions on the computation of the relative entropy of nonclassicality for Fock states. LL and MBP are supported by the ERC Synergy Grant BioQ (grant no.\ 319130) and HyperQ (grant no. 856432). LL thanks also the Alexander von Humboldt Foundation. GF acknowledges the support received from the EU through the ERASMUS+ Traineeship program and from the Scuola Galileiana di Studi Superiori.

\medskip
\noindent\textbf{Competing interest} --- The authors declare no competing interests. \\

\medskip
\noindent\textbf{Data availability} --- No data sets were generated during this study.

\appendix
\renewcommand{\thethm}{A\arabic{thm}}
\renewcommand{\theequation}{A\arabic{equation}}

\section{Restricted asymptotic continuity in infinite dimensions} \label{restr_asymp_cont_sec}

Among the many properties that a monotone may have, one of the most desirable is, in many scenarios, some form of continuity, since it allows to relate the resource content of a state to that of a \enquote{good} approximation of it. The most widely used notion of continuity, when it comes to resource monotones in finite-dimensional QRT, is that of asymptotic continuity~\cite{Synak2006} \tcb{(Definition~\ref{monotones}(c)). Its popularity is due to its role in proving bounds on asymptotic rates. Of course, the notion of asymptotic continuity becomes empty in infinite dimensions, ultimately leading to the ``asymptotic continuity catastrophe'' discussed in the Introduction. In this \tcb{A}ppendix, we show how asymptotic continuity can be defined in a certain sense for infinite-dimensional systems as well, why this extension presents some substantial limitations, and in which sense Theorem~\ref{general_bound_rates_thm} overcomes such limitations.}



\subsection{Abstract approach}

The traditional approach to the problem of bounding asymptotic transformation rates in infinite-dimensional QRTs makes use of the notion of restricted asymptotic continuity~\cite{Eisert2002, Shirokov-sq, Shirokov-AFW-1, Shirokov-AFW-2, Shirokov-AFW-3}. Let us give a formal definition, taken from~\cite[Corollary~7]{Shirokov-sq}.

\begin{Def} \label{restr_asymp_cont_def}
Let $(\mathcal{S},\mathcal{F})$ be a QRT equipped with a monotone $G$. For some $B\in\mathcal{S}$, fix a family of states $\mathcal{T} = \left\{ \mathcal{T}_{B^n}\right\}_{n\in \N_+}$, with $\mathcal{T}_{B^n}\subseteq \D{B^n} = \mathcal{D}\left( \HH^{\otimes n}_B \right)$. We say that $G$ is \textbf{weakly asymptotically continuous on $\boldsymbol{\mathcal{T}}$} if for all sequence of states $(\rho_n)_{n\in \N_+}$ and $(\sigma_n)_{n\in \N_+}$ with $\rho_n,\sigma_n\in \mathcal{T}_{B^n}$ and $\lim_{n\to\infty} \left\|\rho_n-\sigma_n\right\|_1 = 0$ it holds that
\bb
\lim_{n\to\infty} \frac{\left| G(\rho_n) - G(\sigma_n) \right|}{n} = 0\, .
\ee
\end{Def}

Formally, from the above definition the following can be deduced.

\begin{prop} \label{bound_rates_restr_asymp_cont_prop}
Let $(\mathcal{S},\mathcal{F})$ be a QRT equipped with a weakly additive monotone $G$. Consider $A,B\in \mathcal{S}$, and assume that there exists a family of states $\mathcal{T} = \left\{ \mathcal{T}_{B^n}\right\}_{n\in \N_+}$ on which $G$ is weakly asymptotically continuous. Pick $\rho_A\in \D{A}$ and $\sigma_B\in \D{B}$ with $\sigma_B^{\otimes n}\in \mathcal{T}_{B^n}$ for all sufficiently large $n$, and consider the modified asymptotic transformation rate
\bb
R_{\mathcal{T}}(\rho_{A} \to \sigma_{B}) \coloneqq \sup\left\{r:\ \lim_{n\to\infty} \inf_{\substack{ \Lambda_n\in\mathcal{F}\left( A^n\to B^{\floor{rn}}\right) \\ \Lambda_n(\rho_A^{\otimes n})\in \mathcal{T}_{B^n}}} \left\|\Lambda_n\left( \rho_{A}^{\otimes n}\right) - \sigma_{B}^{\otimes \floor{r n}} \right\|_1 = 0 \right\} .
\label{modified_rate}
\ee
\tcb{Then it satisfies:}
\bb
R_{\mathcal{T}}(\rho_{A} \to \sigma_{B}) \leq \frac{G(\rho_A)}{G(\sigma_B)}\, ,
\ee
whenever the right-hand side is well defined.\footnote{This just means that it does not contain the fractions $\frac{\infty}{\infty}$ or $\frac{0}{0}$. We instead convene that $\frac{0}{\infty}=0$ and $\frac{\infty}{0}=\infty$.}
\end{prop}

\begin{proof}
For any sequence of free operations $\Lambda_n\in\mathcal{F}\left( A^n \to B^{\floor{rn}}\right), \Lambda_n(\rho_A^{\otimes n})\in \mathcal{T}_{B^n}$ satisfying
\bbb
\liminf_{n\to\infty}  \left\|\Lambda_n\!\left( \rho_{\!A}^{\otimes n}\right) - \sigma_{\!B}^{\otimes \floor{r n}} \right\|_1\!\! = 0\,,
\eee
it holds that
\begin{align*}
G(\rho_A) &\texteq{1} \liminf_{n\to\infty} \frac1n\,  G\left( \rho_A^{\otimes n} \right) \\
&\textgeq{2} \liminf_{n\to\infty} \frac1n\, G\left(\Lambda_n\left(\rho_A^{\otimes n}\right)\right) \\
&\geq \liminf_{n\to\infty} \left( \frac1n\, G\left( \sigma_B^{\otimes \floor{rn}} \right) - \frac1n \left| G\left(\Lambda_n\left(\rho_A^{\otimes n}\right)\right) - G\left( \sigma_B^{\otimes\floor{rn}} \right)\right| \right) \\
&\texteq{3} \liminf_{n\to\infty} \frac{\floor{rn}}{n}\, G\left( \sigma_B \right) \\
&= r\, G \left( \sigma_B\right) .
\end{align*}
Here, 1~follows from weak additivity, 2~from monotonicity, and 3~from weak asymptotic continuity and again weak additivity.
\end{proof}

The main problem with Proposition~\ref{bound_rates_restr_asymp_cont_prop} is that the rate in~\eqref{modified_rate} only takes into account a restricted set of possible free transformations. We will see what this means in a physically relevant setting below.

\subsection{A physically interesting case: energy-constrained asymptotic continuity}

The above definition may seem rather abstract. However, there is a physically very natural scenario where it can be applied. Let us assume that a certain system $B\in\mathcal{S}$ is equipped with a Hamiltonian (i.e., a self-adjoint operator) $H_B$. Let us assume that the Hamiltonians add up without interaction terms upon taking multiple copies of $B$, in formula $H_{B^n} = H_B \otimes \id_B^{\otimes (n-1)} + \id_B \otimes H_B \otimes \id_B^{\otimes (n-2)} + \ldots + \id_B^{\otimes (n-1)} \otimes H_B$. Now, for a real number $E$, set
\bb
\mathcal{T}_{B^n}^E \coloneqq \left\{ \rho\in \D{B^n}:\, \Tr\left[ \rho\, H_{B^n} \right] \leq nE \right\} .
\label{energy_constrained_set}
\ee
Basically, we are considering states whose energy increases at most linearly in the number of systems $n$. When the set~\eqref{energy_constrained_set} is chosen in Definition~\ref{restr_asymp_cont_def}, the corresponding notion of weakly asymptotic continuity becomes a physically relevant and indeed fruitful one.

\begin{Def}
Let $(\mathcal{S},\mathcal{F})$ be a QRT endowed with a monotone $G$. Let $B\in\mathcal{S}$ be equipped with a Hamiltonian $H_B$. If for all $E$ the monotone $G$ is weakly asymptotically continuous on the set $\mathcal{T}^E$ defined by~\eqref{energy_constrained_set}, then we say that it is \textbf{asymptotically continuous in the presence of an energy constraint}, or \textbf{EC asymptotically continuous} for short.
\end{Def}

In practice, the above definition just means that whenever $(\rho_n)_{n\in \N_+}$ and $(\sigma_n)_{n\in \N_+}$ are sequences of states with $\rho_n, \sigma_n \in \D{A^n}$, $\Tr[\rho_n H_{A^n}], \Tr[\sigma_n H_{A^n}]\leq nE$ (for a fixed but arbitrary real number $E$), and moreover $\lim_{n\to \infty} \left\|\rho_n - \sigma_n\right\|_1=0$, then
\bb
\lim_{n\to\infty} \frac{\left|G(\rho_n) - G(\sigma_n)\right|}{n} = 0\, .
\ee
This definition turns out to encompass a sufficiently wide set of monotones. For example, Shirokov has proved that many important entanglement monotones are EC asymptotically continuous, with respect to several physically relevant Hamiltonians~\cite{Shirokov-sq, Shirokov-AFW-1, Shirokov-AFW-2, Shirokov-AFW-3}. To deduce a useful result from Proposition~\ref{bound_rates_restr_asymp_cont_prop} we need to fix some terminology.

\begin{Def}[{\cite[p.~9]{VV-diamond}}]
Let $A,B$ be two quantum system equipped with Hamiltonians $H_A,H_B$. A quantum channel $\Lambda:\T{A}\to\T{B}$ from $A$ to $B$ is called \textbf{$\boldsymbol{(\kappa,\delta)}$-energy-limited} if $\Lambda^\dag (H_B) \leq \kappa H_A + \delta$, with $\Lambda^\dag:\B{B}\to \B{A}$ being the adjoint of $\Lambda$. The set of such channels will be denoted with $\mathrm{EL}_{\kappa,\delta}(A\to B)$, where the choice of the Hamiltonians is not made explicit and assumed to be clear from the context.
\end{Def}

In such a setting, directly from Proposition~\ref{bound_rates_restr_asymp_cont_prop} we deduce the following.

\begin{prop} \label{bound_rates_EC_asymp_cont_prop}
Let $(\mathcal{S},\mathcal{F})$ be a QRT. Let $H_A,H_B$ be two Hamiltonians on $A,B\in\mathcal{S}$, and let $G$ be a weakly additive monotone that is EC asymptotically continuous. Then for all $\rho_A\in \D{A}$ and $\sigma_B\in \D{B}$ with finite energy (i.e., such that $\Tr[\rho_A H_A], \Tr[\sigma_B H_B]<\infty$) the uniformly energy-constrained (UEC) asymptotic transformation rate defined by
\bb
R_{\text{UEC}}(\rho_{A} \to \sigma_{B}) \coloneqq \sup_{0<\kappa,\delta<\infty} \sup\left\{r:\ \lim_{n\to\infty} \inf_{ \Lambda_n\in\mathcal{F}\left( A^n\to B^{\floor{rn}}\right)\, \cap\, \mathrm{EL}_{\kappa,\delta}\left( A^n\to B^{\floor{rn}} \right)} \left\|\Lambda_n\left( \rho_{A}^{\otimes n}\right) - \sigma_{B}^{\otimes \floor{r n}} \right\|_1 = 0 \right\}
\label{UEC_rate}
\ee
satisfies that
\bb
R_{\text{UEC}}(\rho_{A} \to \sigma_{B}) \leq \frac{G(\rho_A)}{G(\sigma_B)}\, ,
\ee
whenever the right-hand side is well defined.
\end{prop}

\begin{proof}
Let $E\coloneqq \max\left\{\Tr[\rho_A H_A],\, \Tr[\sigma_B H_B]\right\}<\infty$. Fix arbitrary $0<\kappa,\delta<\infty$, and let the rate $r$ be achievable in~\eqref{UEC_rate} by means of a sequence of protocols $\Lambda_n\in\mathcal{F}\left( A^n\to B^{\floor{rn}}\right)\, \cap\, \mathrm{EL}_{\kappa,\delta}\left( A^n\to B^{\floor{rn}} \right)$. Then for sufficiently large $n$ it holds that
\begin{align*}
\Tr\left[\sigma_B^{\otimes n} H_{B^n}\right] &= n \Tr[\sigma_B H_B] \leq nE \leq n E' , \\
\Tr\left[\Lambda_n\left(\rho_A^{\otimes n}\right) H_{B^n}\right] &= \Tr\left[\rho_A^{\otimes n} \Lambda_n^\dag\left(H_{B^n}\right)\right] \leq \Tr\left[\rho_A^{\otimes n} (\kappa H_{A^n} + \delta)\right] = n\kappa \Tr[\rho_A H_A] + \delta \leq n\kappa E + \delta \leq n E' ,
\end{align*}
where we set $E'\coloneqq \max\{\kappa,1\} (E+1)$. Since $G$ is asymptotically continuous on the set $\mathcal{T}^{E'}$, and we have just shown that $\sigma_B^{\otimes n},\, \Lambda_n\left(\rho_A^{\otimes n}\right)\in \mathcal{T}^{E'}$, then we can apply Proposition~\ref{bound_rates_restr_asymp_cont_prop} and conclude the proof.
\end{proof}

The reason why Proposition~\ref{bound_rates_EC_asymp_cont_prop} \tcb{is ultimately not satisfactory, and why we on the contrary deem Theorem~\ref{general_bound_rates_thm} more compelling,} is twofold. First, Proposition~\ref{bound_rates_EC_asymp_cont_prop} only allows us to bound the standard asymptotic transformation rate (Definition~\ref{rate_def}), while we have seen that in certain settings the relevant quantity is the maximal asymptotic transformation rate (Definition~\ref{max_rate_def}). Secondly, \tcb{the above result only takes} into account sequences of protocols $(\Lambda_n)_{n}$ that are \emph{uniformly} energy-constrained, meaning that the output energy is at most $E_{\text{out}}\leq \kappa E_{\text{in}} + \delta$, with $\kappa$ and $\delta$ fixed for the whole sequence. If each $\Lambda_n$ is $(\kappa_n,\delta_n)$-energy-limited for each $n$, but $\limsup_{n\to\infty} \kappa_n = +\infty$ or $\limsup_{n\to\infty} \frac{\delta_n}{n}\tcb{=+\infty}$, the above method does not seem to tell us much about the corresponding rate, even when the initial and final states have a fixed (and finite) energy. Therefore, for instance, a sequence of free operations on CV systems where each $\Lambda_n$ involves either (a)~a squeezing whose intensity increases with $n$ and tends to $\infty$ in the limit $n\to\infty$; or (b)~a displacement unitary whose parameter $\alpha_n$ is superlinear in $n$, are excluded from the bound in Proposition~\ref{bound_rates_EC_asymp_cont_prop}. 

\tcb{As we have seen, our Theorem~\ref{general_bound_rates_thm}} eliminates the need for both of these requirements, and instead provides ultimate bounds on maximal (instead of standard) asymptotic transformation rates, in a setting where the free protocols employed are otherwise totally unconstrained.

\section{Approximation by spectral truncation}

Here we study the problem of approximating $\NCM$ by truncating the input state. The forthcoming Corollary~\ref{approximation ncm cor} has been used to check the numerical validity of the plots in Section~\ref{Fock-diagonal subsec} We state first a useful lemma, whose proof follows closely that of~\cite[Lemma~7]{tightuniform}, with some adaptations made to fit our infinite-dimensional case. In what follows, for a trace class operator $X\in \T{}$ with decomposition $X = X_+ - X_-$ into positive and negative parts, we denote with $|X|\coloneqq X_+ + X_-$ its absolute value.

\begin{lemma} \label{WAC lemma}
Let $\rho,\sigma\in \D{m}$ be two $m$-mode states, and set $\epsilon \coloneqq \frac12 \left\|\rho-\sigma\right\|_1$. Assume that the operator $|\rho-\sigma|$ has finite mean photon number $E\coloneqq \Tr \left[|\rho-\sigma| \left(\sumno_{j=1}^n a_j^\dag a_j\right)\right] < \infty$. Then, for $F=\NCM, \NCi, \NC$ it holds that
\bb
\left| F(\rho) - F(\sigma) \right| \leq m\, \epsilon\, g\left(\frac{E}{m\,\epsilon} \right) + g(\epsilon)\, ,
\label{WAC}
\ee
where $g(x) = (1+x) \log_2(1+x) - x \log_2 x$.
\end{lemma}

\begin{rem}
Recently, Shirokov~\cite{Shirokov-AFW-1} has put forward a more general technique that allows to obtain general continuity results for relative entropy distance measures in infinite dimensions, thus removing the need to make any assumption concerning the operator $|\rho-\sigma|$. While theoretically superior, his bounds are less tight and ultimately not suited for our practical purposes.
\end{rem}

\begin{proof}[Proof of Lemma~\ref{WAC lemma}]
We start with the case where $F=\NC$. Here we actually prove that
\bb
\left| \NC(\rho\otimes \tau) - \NC(\sigma\otimes \tau) \right| \leq m\, \epsilon\, g\left(\frac{E}{m\,\epsilon} \right) + g(\epsilon)\, ,
\label{better WAC}
\ee
for all $n$-mode auxiliary states $\tau$. Call $h_2(p)\coloneqq - p \log_2 p - (1-p) \log_2(1-p)$ the binary entropy function. Using the convexity of $\NC$ (Lemma~\ref{NC_subadd_lemma}) together with~\cite[Proposition~5.24]{PETZ-ENTROPY}, it is not difficult to observe, as done by Winter~\cite[Lemma~7]{tightuniform}, that
\bb
p \NC(\rho_1) + (1-p) \NC(\rho_2) - h_2(p) \leq \NC(p\rho_1 + (1-p) \rho_2) \leq p \NC(\rho_1) + (1-p) \NC(\rho_2)\, .
\label{convex_but_not_too_much}
\ee

We now construct two states $\delta, \delta'\in \D{m}$ such that 
\bbb
\rho - \sigma = \epsilon \left( \delta - \delta'\right) ,\quad |\rho - \sigma | = \epsilon \left(\delta + \delta'\right) .
\eee
In particular, the mean photon number of $\delta$ satisfies that 
\bbb
\Tr \left[\delta \left(\sumno_{j=1}^n a_j^\dag a_j\right)\right] \leq \frac{1}{\epsilon} \Tr \left[|\rho-\sigma| \left(\sumno_{j=1}^n a_j^\dag a_j\right)\right] \leq \frac{E}{\epsilon}\, .
\eee
Now, set
\bbb
\omega \coloneqq \frac{1}{1+\epsilon}\, \rho + \frac{\epsilon}{1+\epsilon}\, \delta' = \frac{1}{1+\epsilon}\, \sigma + \frac{\epsilon}{1+\epsilon}\, \delta\, .
\eee
Then, on the one hand
\begin{align*}
\NC(\omega \otimes \tau) &\textleq{1} \frac{1}{1+\epsilon}\, \NC(\sigma\otimes \tau) + \frac{\epsilon}{1+\epsilon}\, \NC(\delta\otimes \tau) \\
&\textleq{2} \frac{1}{1+\epsilon}\, \NC(\sigma \otimes \tau) + \frac{\epsilon}{1+\epsilon} \left( m\, g\left(\frac{E}{m\,\epsilon}\right) + \NC(\tau)\right) .
\end{align*}
Here, the estimate in~1 comes from convexity~\eqref{convex_but_not_too_much}, while that in~2 is an application of the subadditivity of $\NC$ (Lemma~\ref{NC_subadd_lemma}) together with Proposition~\ref{bounded energy prop}. On the other hand, we can write
\begin{align*}
\NC(\omega \otimes \tau) &= \NC\left( \frac{1}{1+\epsilon}\, \rho \otimes \tau + \frac{\epsilon}{1+\epsilon}\, \delta' \otimes \tau \right) \\
&\textgeq{3} \frac{1}{1+\epsilon}\, \NC(\rho \otimes \tau) + \frac{\epsilon}{1+\epsilon}\, \NC(\delta' \otimes \tau) - h_2\left( \frac{\epsilon}{1+\epsilon} \right) \\
&\textgeq{4} \frac{1}{1+\epsilon}\, \NC(\rho \otimes \tau) + \frac{\epsilon}{1+\epsilon}\, \NC(\tau) - h_2\left( \frac{\epsilon}{1+\epsilon} \right) ,
\end{align*}
where the inequality in~3 is the lower bound in~\eqref{convex_but_not_too_much}, and that in~4 holds because of the monotonicity of $\NC$ under the classical operation of tracing away subsystems. Putting all together we see that
\bbb
\NC(\rho \otimes \tau) - \NC(\sigma \otimes \tau) \leq m\, \epsilon\, g\left(\frac{E}{m\, \epsilon} \right) + (1+\epsilon) h_2\left( \frac{\epsilon}{1+\epsilon} \right) = m\, \epsilon\, g\left(\frac{E}{m\, \epsilon} \right) + g(\epsilon)\, .
\eee
Together with the corresponding inequality with $\rho$ and $\sigma$ exchanged, this yields~\eqref{better WAC}, and in particular proves~\eqref{WAC} for $F=\NC$.

Now, again borrowing a telescopic argument from~\cite{tightuniform}, for all $n\in \N_+$ we have that
\begin{align*}
\left| \NC( \rho^{\otimes n}) - \NC(\sigma^{\otimes n}) \right| &= \left| \sumno_{k=0}^{n-1} \left( \NC\left(\rho^{\otimes (n-k)}\otimes \sigma^{\otimes k}\right) - \NC\left(\rho^{\otimes (n-k-1)}\otimes \sigma^{\otimes (k+1)}\right) \right) \right| \\
&\leq \sumno_{k=0}^{n-1} \left| \NC\left(\rho^{\otimes (n-k)}\otimes \sigma^{\otimes k}\right) - \NC\left(\rho^{\otimes (n-k-1)}\otimes \sigma^{\otimes (k+1)}\right) \right| \\
&= \sumno_{k=0}^{n-1} \left| \NC\left(\rho \otimes \tau_k \right) - \NC\left(\sigma \otimes \tau_k\right) \right| \\
&\leq k\, \left( m\, \epsilon\, g\left(\frac{E}{m\,\epsilon} \right) + g(\epsilon) \right) ,
\end{align*}
where in the last line we applied~\eqref{better WAC}. Diving by $k$ and taking the limit for $k\to\infty$ we see that
\bbb
\left| \NCi(\rho) - \NCi(\sigma) \right| \leq m\, \epsilon\, g\left(\frac{E}{m\,\epsilon} \right) + g(\epsilon)\, ,
\eee
which proves~\eqref{WAC} also when $F=\NCi$.

The case of $F=\NCM$ can be tackled with exactly the same techniques, because $\NCM$ obeys an inequality analogous to~\eqref{convex_but_not_too_much}. In turn, this is a consequence of the fact that the classical Kullback--Leibler divergence satisfies the estimates in~\cite[Proposition~5.24]{PETZ-ENTROPY}.
\end{proof}

\begin{rem}
We have not been able to establish~\eqref{WAC} also for the remaining case of $F=\NCMi$, essentially because we lack a statement similar to~\eqref{better WAC} for $\NCM$. In turn, this is due to the fact that this latter quantity is not subadditive --- in fact, it is strongly superadditive!
\end{rem}

The application of Lemma~\ref{WAC lemma} that is of interest to us is as follows.

\begin{cor}\label{approximation ncm cor}
Let $\rho,\sigma\in \D{m}$ be two $m$-mode states, and set $\epsilon \coloneqq \frac12 \left\|\rho-\sigma\right\|_1$. Assume that $\Tr \left[\rho \left( \sumno_{j=1}^m a_j^\dag a_j\right)\right]\leq E$ and also $\Tr \left[\sigma \left( \sumno_{j=1}^m a_j^\dag a_j\right)\right]\leq E$. Then, for $F=\NCM, \NCi, \NC$ it holds that
\bb
\left| F(\rho) - F(\sigma) \right| \leq m\, \epsilon\, g\left(\frac{2E}{m\,\epsilon} \right) + g(\epsilon)\, ,
\label{WAC commuting}
\ee
where again $g(x) = (1+x) \log_2(1+x) - x \log_2 x$. In particular, denoting with $\rho=\sum_k p_k \ketbra{e_k}$ the spectral decomposition of $\rho$, the sequence of spectral truncations $\rho_n\coloneqq \left( \sumno_{k\leq n} p_k\right)^{-1} \sum_{k\leq n} p_k \ketbra{e_k}$ satisfies that
\bb
F(\rho) = \lim_{n\to\infty} F(\rho_n)\, .
\label{approximation}
\ee
\end{cor}

\begin{proof}
Thanks to Lemma~\ref{WAC lemma}, in order to prove~\eqref{WAC commuting} it suffices to show that $\Tr \left[|\rho-\sigma| \left( \sumno_{j=1}^m a_j^\dag a_j\right)\right] \leq 2E$. Indeed, if $\rho=\sum_k p_k \ketbra{e_k}$ and $\sigma=\sum_k q_k \ketbra{e_k}$ then 
\bbb
|\rho-\sigma| = \sum_k |p_k-q_k| \ketbra{e_k} \leq \sum_k (p_k+q_k) \ketbra{e_k}=\rho+\sigma\, ,
\eee
so that 
\bbb
\Tr \left[|\rho-\sigma| \left( \sumno_{j=1}^m a_j^\dag a_j\right)\right] \leq \Tr \left[(\rho+\sigma) \left( \sumno_{j=1}^m a_j^\dag a_j\right)\right] \leq 2E\, .
\eee

To deduce~\eqref{approximation}, note that $[\rho,\rho_n]=0$, with $\epsilon_n \coloneqq \frac12 \left\|\rho-\rho_n\right\|_1\tendsn{} 0$. Also, for sufficiently large $n$ the mean photon number of $\rho_n$ is at most twice that of $\rho$ (call it $E$), so that
\bbb
\left| F(\rho) - F(\rho_n) \right| \leq m\, \epsilon_n\, g\left(\frac{4E}{m\, \epsilon_n} \right) + g(\epsilon_n) \tendsn{} 0\, ,
\eee
where we used the well-known fact that $\lim_{\epsilon\to 0^+} \epsilon\, g(\delta/\epsilon)=0$ for all $\delta>0$.
\end{proof}

\section{A technical lemma}

The following result was invoked in the proof of Proposition~\ref{upper and lower bound for ncm prop} (more precisely, in step~3 of~\eqref{upper and lower bound for ncm eq1}) as an alternative to~\cite[Proposition~5.23(iv)]{PETZ-ENTROPY}. For the sake of completeness, we include a proof.

\begin{lemma}\label{stupid lemma}
Let us consider two $m$-mode states $\rho,\sigma$, with $S_W(\rho)<\infty$. Then, it holds that
\bbb
\DM(\rho\|\sigma)\geq \DKL(Q_\rho\|Q_\omega)\,.
\eee
\end{lemma}
\begin{proof}
At fist sight, the result might seem a trivial application of the definitions of measured relative entropy as given in~\eqref{measured re} and of heterodyne detection as given in~\cite[5.4.2]{BUCCO}. The only obstacle is that, according to our definition, a POVM must be a finite collection of operators; this is also crucial for proving Lemma~\ref{Berta variational lemma}, and hence cannot be simply removed. The idea of the proof is precisely to show that an
heterodyne detection, which corresponds to the decomposition of the identity $\id= \int_{\C^m} \frac{d^{2m}\alpha}{\pi^m} \, \ketbra{\alpha}$, can be approximated, inside the Kullback--Leibler divergence, by a sequence of finite POVMs. Note that the hypothesis $S_W(\rho)<\infty$ is needed to ensure that $\DKL(\rho\|\sigma)$ is a well-defined Lebesgue integral for any $\sigma\in\D{}$ (possibly diverging to $+\infty$).

Let us fix $\epsilon,\,\eta,\,r_1,\,r_2,\,\ell>0$ and consider the family of subsets of $\C^m$, $\pazocal{A}=\{A_n\}_{n=0,\ldots,N+1}$, with $N$ a finite positive integer, constructed as follows:
\begin{itemize}
    \item $A_0\coloneqq\left\{\alpha\in\C^m:Q_\rho(\alpha)<\eta\text{ or }Q_\sigma(\alpha)<\eta\right\}\,.$
    \item For $n=1,\ldots,N$, $A_n$ is a subset of $\C^m$ such that:
    \begin{itemize}
        \item it is contained in a ball of radius $\ell$;
        \item $A_n\cap A_{n'}=\emptyset$ if $n\neq n'$;
        \item $B_m(r_1)\subset\bigcup_{n=0}^{N}A_n$ and $\bigcup_{n=1}^{N}A_n\subset B_m(r_2)$, where $B_m(r)$ is the $m$-dimensional ball with radius $r$ and centered at the origin of $\C^m$.
    \end{itemize}
    \item $A_{N+1}\coloneqq \C^m\setminus\bigcup_{n=0}^{N}A_n$.
\end{itemize}
Notice that we do not specify the dependence of $\pazocal{A}$ upon $\epsilon,\,\eta,\,r_1,\,r_2$ and $\ell$ for ease of notation. It is clear from the definition that the set $\pazocal{M}=\{E_n\}_{n=1,\ldots,N}$, with
\bbb
E_n= \int_{A_n} \frac{d^{2m}\alpha}{\pi^m}\, \ketbra{\alpha}
\eee
is a proper, finite POVM.

For a fixed state $\rho$ we can compute the classical probability distribution
\bbb
\pazocal{P}^{\pazocal{M}}_\rho(n)=\Tr[\rho E_n]= \int_{A_n} \frac{d^{2m}\alpha}{\pi^m} \,\braket{\alpha|\rho|\alpha}=\int_{A_n} \frac{d^{2m}\alpha}{\pi^m} \,Q_\rho(\alpha)\,.
\eee
Since $Q_\rho\in L^1(\C^m)$, we can always take $r_1$ big enough so that $\pazocal{P}^{\pazocal{M}}_\rho(N+1)<\epsilon$. Similarly we can always choose $\eta$ such that $\pazocal{P}^{\pazocal{M}}_\rho(0)<\epsilon$. Now, given another fixed state $\sigma$, we have:
\begin{align*}
D^M(\rho\|\sigma)&\geq\DKL\left(\pazocal{P}^{\pazocal{M}}_\rho\|\pazocal{P}^{\pazocal{M}}_\sigma\right)\\
&=\sum_{n=0}^{N+1}\int_{A_n}d^{2m}\alpha \,Q_\rho(\alpha)\left[\log_2\left(\int_{A_n}d^{2m}\beta \,Q_\rho(\beta)\right)-\log_2\left(\int_{A_n}d^{2m}\beta \,Q_\sigma(\beta)\right)\right]\\
&\textgeq{1}\sum_{n=0}^{N+1}\int_{A_n}d^{2m}\alpha \,Q_\rho(\alpha)\log_2\left(\int_{A_n}d^{2m}\beta \,Q_\rho(\beta)\right)-\sum_{n=1}^{N}\int_{A_n}d^{2m}\alpha \,Q_\rho(\alpha)\log_2\left(\int_{A_n}d^{2m}\beta \,Q_\sigma(\beta)\right)\\
&\texteq{2}\sum_{n=1}^{N}\int_{A_n}d^{2m}\alpha \,Q_\rho(\alpha)\left[\log_2\left(\int_{A_n}d^{2m}\beta \,Q_\rho(\beta)\right)-\log_2\left(\int_{A_n}d^{2m}\beta \,Q_\sigma(\beta)\right)\right]+\pazocal{O}(\epsilon)\\
&\textgeq{}\sum_{n=1}^N\int_{A_n}d^{2m}\alpha \,Q_\rho(\alpha)\left[\log_2\left(\cancel{\text{Vol}(A_n)}\inf_{\beta\in A_n} \,Q_\rho(\beta)\right)-\log_2\left(\cancel{\text{Vol}(A_n)}\sup_{\beta\in A_n} \,Q_\sigma(\beta)\right)\right]+\pazocal{O}(\epsilon)\\
&\texteq{3}\sum_{n=1}^N\int_{A_n}d^{2m}\alpha \,Q_\rho(\alpha)\left[\log_2\left(Q_\rho(\alpha)+\pazocal{O}(\ell)\right)-\log_2\left(Q_\sigma(\alpha)+\pazocal{O}(\ell)\right)\right]+\pazocal{O}(\epsilon)\\
&\texteq{4}\sum_{n=1}^N\int_{A_n}d^{2m}\alpha \,Q_\rho(\alpha)\left(\log_2Q_\rho(\alpha)-\log_2Q_\sigma(\alpha)\right)+\pazocal{O}(\epsilon)+\pazocal{O}(\ell)\\
&=\int_{\bigcup_{n=1}^NA_n}d^{2m}\alpha \,Q_\rho(\alpha)\left(\log_2Q_\rho(\alpha)-\log_2Q_\sigma(\alpha)\right)+\pazocal{O}(\epsilon)+\pazocal{O}(\ell)\,.
\end{align*}
Here, in~1 we have discarded two positive terms, in~2 we have used the fact that $\pazocal{P}_\rho^{\pazocal{M}}(0),\pazocal{P}_\rho^{\pazocal{M}}(N+1)<\epsilon$ and $\lim_{x\to0}x\log_2x=0$, in~3 we have used the fact that $Q_\rho$ and $Q_\sigma$ are uniformly continuous on any compact set and finally in~4 we have used the fact that also the logarithm is uniformly continuous on $\bigcup_{n=1}^NA_n$.
Now, by taking the proper limits, i.e., $\epsilon,\eta,\ell\to0$ and $r_1,r_2\to\infty$, we obtain precisely the generalized Riemann integral over $\C^m$ of the function $Q_\rho(\alpha)\left(\log_2Q_\rho(\alpha)-\log_2Q_\sigma(\alpha)\right)$. Since $S_W(\rho)<\infty$, the function is absolutely integrable, and the generalized Riemann integral is well defined and coincides with the Lebesgue one.
\end{proof}

\bibliographystyle{apsrev4-1}
\bibliography{mainbiblio}

\end{document}